\documentclass[letterpaper]{article}

\usepackage{geometry}
\usepackage[utf8x]{inputenc}
\usepackage[T1]{fontenc}
\usepackage{graphicx}
\usepackage{subcaption}
\usepackage{longtable} 
\usepackage{multirow}
\usepackage{listings}
\usepackage{makecell}
\usepackage{colortbl,booktabs}
\usepackage{longtable}
\usepackage{float}
\usepackage[dvipsnames]{xcolor}

\usepackage{amsmath}
\usepackage{bm}
\usepackage{amssymb}
\usepackage{amsthm}
\usepackage{graphicx}
\usepackage{natbib}
\usepackage[ruled]{algorithm2e}
\usepackage{longtable}
\usepackage[colorinlistoftodos]{todonotes}
\usepackage{hyperref}
\hypersetup{
    colorlinks,%
    citecolor=blue,%
    filecolor=blue,%
    linkcolor=black,%
    urlcolor=blue
}

\usepackage{enumitem}
\usepackage{mathtools}
\usepackage{mathrsfs}
\usepackage{tikz}
\usetikzlibrary{arrows.meta}
\usetikzlibrary{positioning}
\usetikzlibrary{decorations.pathreplacing}
\usepackage{pgf}
\usepackage{dsfont}
\DeclareMathOperator{\ind}{\mathds{1}}  

\theoremstyle{plain}

\newtheorem{definition}{Definition}
\newtheorem{theorem}{Theorem}
\newtheorem{lemma}{Lemma}

\newtheorem{corollary}{Corollary}

\newtheorem{proposition}{Proposition}

\newcommand{\dd}{\mathrm{d}}

\newcommand{\R}{\mathbb{R}}
\newcommand{\real}{\mathbb{R}}

\newcommand{\E}{\mathbb{E}}
\newcommand{\p}{\mathbb{P}}
\newcommand{\N}{\mathcal{N}}

\newcommand{\indep}{\rotatebox[origin=c]{90}{$\models$}}

\newcommand{\indc}[1]{{\mathbf{1}_{\left\{{#1}\right\}}}}
\newcommand{\Indc}{\mathbf{1}}

\newcommand{\defn}{:=}

\usepackage{xspace}

\newcommand{\floor}[1]{{\left\lfloor {#1} \right \rfloor}}
\newcommand{\ceil}[1]{{\left\lceil {#1} \right \rceil}}

\definecolor{myblue}{rgb}{.8, .8, 1}
\definecolor{mathblue}{rgb}{0.2472, 0.24, 0.6} 
\definecolor{mathred}{rgb}{0.6, 0.24, 0.442893}
\definecolor{mathyellow}{rgb}{0.6, 0.547014, 0.24}

\newcommand{\tX}{{\widetilde{X}}}

\newcommand{\calC}{{\mathcal{C}}}

\newcommand{\calG}{{\mathcal{G}}}
\newcommand{\calH}{{\mathcal{H}}}

\newcommand{\calN}{{\mathcal{N}}}

\newcommand{\calS}{{\mathcal{S}}}

\long\def\comment#1{}

\newcommand{\algoname}{derandomized knockoffs}
\newcommand{\eqd}{\stackrel{\textnormal{d}}{=}}

\newcommand{\Null}{\mathcal{H}_0}
\newcommand{\Exs}{\mathbb{E}}

\newcommand{\bmx}{\bm{X}}
\newcommand{\bmtx}{\widetilde{\bm{X}}}

\newcommand{\bmY}{\bm Y}

\newcommand{\kFWER}{k\operatorname{-FWER}}

\title{Derandomizing Knockoffs} 

\author{
Zhimei Ren\thanks{Department of Statistics, Stanford University, Stanford, CA 94305} 
\and Yuting Wei\thanks{Statistics \& Data Science Department, Carnegie Mellon University, Pittsburgh, PA 15213}
\and Emmanuel Cand\`es\thanks{Department of Mathematics and Department of Statistics, Stanford University, Stanford, CA 94305}\\[0.2cm]
}

\date{\today}

\begin{document}
\maketitle

\begin{abstract}
  Model-X knockoffs is a general procedure that can leverage any
  feature importance measure to produce a variable selection
  algorithm, which discovers true effects while rigorously controlling
  the number or fraction of false positives. Model-X knockoffs is a
  randomized procedure which relies on the one-time construction of
  synthetic (random) variables. This paper introduces a
  derandomization method by aggregating the selection results across
  multiple runs of the knockoffs algorithm. The derandomization step
  is designed to be flexible and can be adapted to any variable
  selection base procedure to yield \emph{stable} decisions without
  compromising statistical power.  When applied to the base procedure
  of \cite{janson2016familywise}, we prove that derandomized knockoffs
  controls both the per family error rate (PFER) and the $k$
  family-wise error rate ($k$-FWER).  Further, we carry out extensive
  numerical studies demonstrating tight type-I error control and
  markedly enhanced power when compared with alternative variable
  selection algorithms.  Finally, we apply our approach to multi-stage
  genome-wide association studies of prostate cancer and report
  locations on the genome that are significantly associated with the
  disease.  When cross-referenced with other studies, we find that the
  reported associations have been replicated.


\end{abstract}



\vspace*{0.5cm}
\section{Introduction}
\label{sec:intro}
There has been a surge of interest in the design of trustworthy
inferential procedures for massive data applications with a focus on
the development of variable selection algorithms that are flexible,
and at the same time, possess clear performance guarantees. Among
them, the method of knockoffs, or knockoffs for short
\citep{barber2015controlling,candes2018panning}, has proved
particularly effective in a variety of applications
\citep{gao2018model,srinivasan2019compositional,sesia2020multi}. Imagine
a scientist wishes to infer which of the many covariates she has
measured are truly associated with a response of interest; for
instance, which of the many genetic variants influence the
susceptibility of a disease. At a high-level, the knockoffs selection
algorithm begins by synthesizing `fake' copies of the covariates (fake
genetic variants in our example), which can be thought of as serving
as a control group for the features. By contrasting the values a
feature importance statistic takes on when applied to a true variable
and a fake variable, it becomes possible to tease apart those features
which have a true effect on the response. This can be achieved via a
clever filter while controlling either the expected fraction of false
positives \citep{barber2015controlling} or simply the number of false
positives \citep{janson2016familywise}.

A frequently discussed issue is that knockoffs is a randomized
procedure; that is, the fake covariates (the knockoffs) are
stochastic. Therefore, different runs of the algorithm produce
different knockoffs (unless we use the same random seed), and in big
data applications, researchers have observed that the selection
algorithm may each time return overlapping, yet, different selected
sets. This has led researchers to report those features whose
selection frequency exceeds a threshold along with the corresponding
frequencies \citep{candes2018panning,sesia2019gene}.  While
statisticians are accustomed to randomized procedures---after all, any
procedure based on data splitting is randomized in the sense that
different splits typically yield different outcomes---it is still
desirable to derandomize the knockoffs selection algorithm as to
produce consistent results. This paper achieves this goal by running
the knockoffs algorithm several times and aggregating results
across all runs.

\paragraph{An overview of our contributions}
Our derandomization scheme is inspired by the stability selection
framework of \cite{meinshausen2010stability} and
\cite{shah2013variable}, which finds its roots in bootstrap
aggregating
\citep{efron1983leisurely,breiman1996bagging,breiman1999using},
subbagging \citep{buhlmann2002analyzing} and random forest
\citep{breiman2001random}.  In a nutshell, our scheme consists in
applying a knockoffs selection algorithm multiple times, each time
with a new matrix of knockoffs, and proceeds by aggregating the
results using the same rationale behind the stability selection
criterion, which as its name suggests, puts a premium on stability and
consistency. We will demonstrate that this indeed reduces the
variability of the outcome.  In Section \ref{sec:methodology}, we will
however explain why the similarities between stability selection and
derandomized knockoffs stop here, and why the interpretation and
properties of the two procedures are very different. Moving on, we
empirically demonstrate that derandomized knockoffs achieves tight
type-I error control and markedly enhanced power when compared with
alternative variable selection algorithms, including `vanilla'
knockoffs.  We establish theoretical support for derandomized
knockoffs by proving per family error rate (PFER) control and $k$
family-wise error rate ($k$-FWER) control.

Besides methodological developments, a fair fraction of this paper is
concerned with applying our ideas to genome-wide association studies
(GWAS) in Section \ref{sec:application}. We make two contributions.
\begin{itemize}
\item In previous applications of knockoffs to GWAS, the base
  procedure is typically applied multiple times with different random
  seeds. While each run comes with a type-I error guarantee, the
  authors often report genetic variants together with their selection
  frequency to identify variants which are consistently discovered,
  see \citet{ren2020knockoffs} for some examples. One issue is that we
  would not know how to interpret a `meta-set' of variants whose
  selection frequency is above a given threshold. By this, we mean
  that we would not be able to give this meta-set type-I error
  guarantees. The methods from this paper offer a remedy.

\item We design a general and scientifically sound workflow for
  multi-stage GWAS. Suppose we have a family of SNPs
  $X_1, \ldots, X_p$ and are interested in determining whether the
  distribution of a phenotype $Y$ {\em conditional} on
  $X_1, \ldots, X_p$ depends on $X_j$ or not; that is, we want to know
  whether $Y$ depends on $X_j$ controlling for all the other variables
  $X_{-j}$.  We will show how to achieve this in a multi-stage
  approach, where one can use a first study to determine a set of
  candidate SNPs and a second study for confirmatory analysis.
\end{itemize}

\section{A framework for derandomizing knockoffs}
\label{sec:methodology}
\paragraph*{Knockoffs} To set the stage for derandomized knockoffs,
imagine we are given a response variable $Y$ and potential explanatory
variables $X = (X_1,\ldots,X_p)$. We would like to identify those
variables that truly influence the response; that is, we would like to
discover those $X_j$'s on which the distribution
$Y \mid X_1, \ldots, X_p$ depends.  Formally, a variable $X_j$ is said
to be \emph{null} if the response $Y$ is independent of $X_j$ given
all other variables; i.e.,
\begin{align}\label{eq:hypothesis}
	Y~ \indep~ X_j\mid X_{-j} 
\end{align}
(throughout, $X_{-j}$ is a shorthand for all $p$ features except the
$j$-th). Our goal is of course to test each of the $p$ nonparametric
hypotheses \eqref{eq:hypothesis}.

In this setting, the key idea underlying knockoffs is to generate
`fake' covariates $\tX = (\tX_1,\ldots,\tX_p)$ whose distribution
roughly matches that of the true covariates, except that knockoffs are
designed to be conditionally independent of the response, and hence
should never be selected by a feature selection procedure.  Assemble
the covariates in an $n \times p$ matrix $\bmx$ and the responses in
an $n \times 1$ vector $\bmY$. Then we say that the new set of
variables $\bmtx \in \R^{n \times p}$ is a knockoff copy of
$\bmx$ if the following two properties hold: first,
\begin{equation}
  \label{eq:exch}
  \bmx_{j},  \bmtx_j \mid \bmx_{-j}, \bmtx_{-j} \, \eqd \, \bmtx_{j},
  \bmx_j \mid \bmx_{-j}, \bmtx_{-j}. 
\end{equation}
This says that by looking at $\bmx$ and $\bmtx$ we cannot tell whether
the $j$th column is a true variable or a knockoff. (The point is that
if $\bmx_j$ is non null, then we can tell by looking at $\bmY$.) The
second property is that $\bmY~\indep~ \bmtx \mid \bmx$. This says that
knockoffs provide no further information about the response (knockoffs
are constructed without looking at $\bmY$).

To perform variable selection, the researcher applies her favorite
{\em feature importance statistic} to the augmented data set
$(\bmx,\bmtx,\bmY)$ and scores each of the original and knockoff
variables. For example, she can score each variable by recording the
magnitudes of the Lasso coefficients for a value of the regularization
parameter chosen by cross-validation. The scores are then combined to
produce a test statistic for each feature.  This can be as simple as
the difference between the feature importance statistics, e.g.~the
difference between the magnitude of the Lasso coefficient of the
original feature and that of its knockoff. In the sequel, we refer to
this test statistic as the {\em Lasso coefficient difference} (LCD,
\citealt{candes2018panning}).  Finally the test statistics are passed
through the knockoff filter (e.g.~SeqStep,
\citealt{barber2015controlling}) and a selection set $\hat{\calS}$ is
generated.

\paragraph*{Derandomized knockoffs} With these preliminaries, our
derandomized procedure to stabilize the selection set over different
runs is as follows:
\begin{itemize}
\item Construct $M$ conditionally independent knockoff copies
  $\bmtx^1,\ldots,\bmtx^M \in \real^{n\times p}$.
  
\item For each $m \in [M] := \{1, \ldots, M\}$, apply a \emph{base
    procedure} to produce a rejection set $\hat{\calS}^m$.

\item For each feature $j$, compute its \emph{selection frequency}
  via 
\begin{align}
\label{EqnPi}
\Pi_j \defn \dfrac{1}{M}\sum^M_{m=1}\ind \{j\in \hat{\calS}^m\}.
\end{align}
\item Lastly, given a threshold $\eta > 0$, return the final selection
  set
\begin{align}
\label{EqnSetSDK}
\hat{\calS} \,\defn\, \{j\in [p]:\Pi_j\geq \eta\}.
\end{align}
\end{itemize}
The above derandomized knockoffs procedure is summarized in
Algorithm~\ref{algo.skn}. Readers will recognize that \eqref{EqnPi}
and \eqref{EqnSetSDK} are borrowed from stability selection (please
see below for a detailed comparison).  The parameter $\eta$ controls
how many times a variable needs to be selected to be present in the
final selection set. The larger $\eta$, the fewer variables will
ultimately be selected.  Unless otherwise specified, we fix $\eta$ to
be $0.5$ throughout the paper for simplicity. 

\begin{algorithm}[t]
  \DontPrintSemicolon  
  \SetAlgoLined
  \BlankLine
  \caption{Derandomized knockoffs procedure \label{algo.skn}}
  \textbf{Input:} Covariate matrix $\bmx\in \R^{n\times p}$; response variables $\bmY \in \R^{n}$; number of realizations $M$; a base procedure;
  selection threshold $\eta$.\;
  1. \For{$m=1,\ldots,M$}{
    i.  Generate a knockoff copy $\bmtx^m$. \;
      ii. Run the base procedure with  $\bmtx^m$ as knockoffs and obtain the selection set $\hat{\calS}^m$.\;
  }
  2. Calculate the selection probability
  \begin{equation*}
    \Pi_j = \dfrac{1}{M}\sum^M_{m=1}\ind \{j\in \hat{\calS}^m\}.
  \end{equation*}
  \textbf{Output}: selection set $\hat{\calS} \defn \{j\in[p]: \Pi_j\geq \eta\}$.
\end{algorithm}

At each iteration, the statistician is allowed to use a different
knockoffs generating distribution as well as a different test
statistic. That said, consider the scenario in which each copy
$\bmtx^m$ is identically distributed and that the same test statistics
are used (e.g.~LCD). Then the law of large numbers implies that each
$\Pi_j$ converges to \mbox{$\p(j\in \hat{\calS}^1 \mid \bmx,\bmY)$} as
$M$ increases to infinity. We thus see that in the limit of an
infinite number of knockoff copies, the procedure is fully
derandomized since the outcome is determined by $\bmx$ and $\bmY$.

\paragraph*{Reduced variability} When working on a specific data set
or application, researchers are typically interested in the fraction
of false positives (FDP) and/or the number $V$ of false
positives. Even in the case where one employs a procedure controlling
the false discovery rate (this is the expected value of the FDP) or
the PFER (this is the expected value of $V$), one would always prefer
a method which has lower variability in FDP and/or $V$ so that FDP and
$V$ are close to their expectations. The reason is that on any given
data set, we would like to be sure that the fraction and/or number of
false positives are not too high. The variability of these random
variables and others, such as whether a specific variable is selected
or not, originates from different sources. First, it comes from the
draw we got to see, i.e.~the sample $\bmx, \bmY$. In the case of
knockoffs, it also comes from the random nature of the algorithm
itself producing $\bmtx$. Clearly, the derandomization scheme removes
the second source of variability, which is a desirable trait.  

\paragraph*{Connections to prior literature}
Certainly, aggregating results from multiple runs of a random
procedure is not a new idea. A line of work develops methods to
represent the consensus over multiple runs of one algorithm, with the
aim of reducing its sensitivity to the initialization or the
randomness inherent in the algorithm, see
e.g.~\cite{bhattacharjee2001classification,monti2003consensus}.
Another line of prior work seeks to combine multiple different
learning algorithms to improve performance, which is often referred to
as \emph{ensemble learning}. A few examples would include
\cite{strehl2002cluster,rokach2010ensemble,polikar2012ensemble}.  Yet,
most of these methods are neither directly applicable to the knockoffs
framework nor come with a finite sample type-I error control.

\paragraph{Comparisons with stability selection} It is time to expand
on the similarities and differences with stability selection. To
facilitate this discussion, it is helpful to briefly motivate and
describe the stability selection algorithm. We are given data
$(\bmx, \bmY)$ and would like to find important variables by reporting
those variables which have a nonzero Lasso coefficient. How confident
are we that our selections will replicate in the sense that we would
get a similar result on an independent data set? How do we make sure
that the nonzero coefficients are not merely due to chance?  Stability
selection addresses this issue by sampling repeatedly
$\lfloor n/2 \rfloor$ observations without replacement from the
original data set as if they were independent draws from the
population.\footnote{$\floor{x}$ denotes the largest integer that is
  not greater than $x$, and $\ceil{x}$ denotes the smallest integer
  that is not less that $x$.}  Important variables are then determined
based on their selection frequencies just as in
\eqref{EqnPi}--\eqref{EqnSetSDK}. Despite evident similarities, there
are major differences with derandomized knockoffs.
\begin{itemize}
\item First, stability selection introduces randomness via data splits
  and it is precisely this extra source of randomness which permits
  inference. In contrast, vanilla knockoffs natively provides valid
  inference and the aim of the derandomized procedure is simply to
  remove the randomness of the knockoffs.
\item Second, while stability selection benefits from the bootstrap
  aggregating procedure preventing overfitting, it only operates on a
  random subset of the data at each step. This difference explains why
  our algorithm is particularly useful in the case where the samples
  size is comparable to the number of features we are assaying, since
  subsampling inevitably leads to a loss of power. We refer the reader
  to the numerical comparisons from Section~\ref{sec:sim} that
  illustrate this point.
\item Third, the theoretical guarantees for stability selection come
  with very strong assumptions---such as the exchangeability
  assumption of null statistics---which are nearly impossible to
  justify in practice. In contrast, our theoretical results hold under
  fairly mild
  assumptions. 
\end{itemize}

\paragraph{A representative base procedure} While our aggregating
framework can be easily applied to a wide range of base procedures,
the current paper focuses on that proposed by
\citet{janson2016familywise}---referred to as $v$-knockoffs
throughout---which has been shown to control the PFER. Informally,
suppose we wish to make at most $v$ false discoveries over the long
run.  Then this base procedure (1) sorts the features based on the
absolute value of their test statistic $|W_j|$;\footnote{In the case
  of the LCD, $W_j = |\hat{\beta}_j| - |\tilde{\beta}_j|$, where
  $\hat{\beta}_j$ (resp.~$\hat{\beta}_j$) is the lasso coefficient
  estimate for the variable $X_j$ (resp.~$\tilde{X}_j$) when
  regressing $\bmY$ on $\bmx$ and $\widetilde{\bmx}$ jointly.} (2) it
then examines the ordered features starting from the largest $|W_j|$
and selects those examined features with $W_j > 0$; (3) the procedure
stops the first time it sees $v$ features with negative values of
$W_j$.  For more details about $v$-knockoffs, we refer the readers to
Section~\ref{Sec:details} from the Appendix.

\section{Theoretical guarantees: controlling the PFER}
\label{sec:pfer}

We now tune derandomized knockoffs parameters to control the per family
error rate.  Formally, let $\Null \subset [p]:=\{1,\cdots,p\}$ denote
the set of null variables for which \eqref{eq:hypothesis} is true, and
consider a selection procedure producing a set of discoveries
$\hat{\calS} \subset [p]$. Letting $V$ be the number of false
discoveries defined as
\begin{align}
  V \defn \#\{j: j \in \Null \cap \hat{\calS} \}, 
\end{align}
the PFER is simply the expected number of false discoveries,
$\text{PFER} = \Exs [V]$ (see, e.g.~\cite{dudoit2007multiple}).

\begin{theorem}
\label{cor:monotone}
Consider derandomized knockoffs (Algorithm \ref{algo.skn}) with a base
procedure obeying $\operatorname{PFER}\le v$ (e.g.~$v$-knockoffs).
If the condition 
\begin{align}
  \label{eq:ratio_bnd}
  \p(\Pi_j \ge \eta)~\le~ \gamma \E[\Pi_j],
\end{align}
holds for every $j\in \calH_0$, then the PFER can be controlled as
\begin{align}
\label{eq:firstbound}
\E [V] ~\le~ \gamma v.
\end{align}
In particular, Markov's inequality ensures that we always have 
\begin{equation}
  \label{eq:markov}
\E [V] ~\le~ v/\eta.
\end{equation}
\end{theorem}
To prove Theorem \ref{cor:monotone}, observe that
\begin{align}
\label{eq:proof_pfer}
\E[V] = \E\left[\sum_{j\in\calH_0} \Indc\{\Pi_j\ge \eta\} \right] 
= \sum_{j\in \calH_0}\p\left(\Pi_j\ge \eta\right) {\le} \sum_{j\in \calH_0} \gamma \E[\Pi_j] 
= \gamma\E[V_1]\le \gamma v,
\end{align}
where $V_{1}$ denotes the number of false discoveries in
$\hat{\calS}^1$; the first inequality follows from
\eqref{eq:ratio_bnd} and the second from the property of the base
procedure.

Returning to the comparison with stability selection, we note that
PFER control holds regardless of the choice of $M$ and without
any assumption on the exchangeability of the selected variables.

\subsection{Guarantees under mild assumptions}
\label{sec:asst_theory}
Set $\eta = 1/2$. In this case, we have seen that \eqref{eq:ratio_bnd}
holds with $\gamma=2$. This is however too conservative in all the
cases we have ever encountered. In fact, we will be surprised to ever
see an example where the ratio $\p(\Pi_j \ge 1/2)/\E[\Pi_j]$ exceeds
one. Consider for instance the setting from Figure
\ref{fig:pfer_amp_gaussian_within}. In this case, Figure
\ref{fig:pfer_amp_gaussian_ratio} plots the realized ratios
$\p(\Pi_j\ge 1/2)/\E[\Pi_j]$ for each null variable, and we can
observe that all the ratios are below one.

Turning to formal statements, Proposition \ref{prop:asst_ratio} below
examines assumptions under which pairs $(\eta,\gamma)$ obey condition
\eqref{eq:ratio_bnd}.  The idea is very similar to that of
\cite{shah2013variable}, where a general bound is first established
and then followed by a sharpened version holding under constraints on
the shape of the distribution $\Pi_j$.

\begin{definition}
\label{def:monotonicity}
Let $M$ be a positive integer and $X$ a random variable supported on
$\{0,1/M,\ldots,1\}$. The probability mass function (pmf) of $X$ is
said to be monotonically non-increasing if for any
{$m_1 \le m_2\in \{0,1,\ldots,M\}$},
\[
  \p(X = {m_1}/{M}) ~\ge~ \p(X = {m_2}/{M}).
\]
\end{definition}

\begin{figure}
	\centering
	\includegraphics[width = .95\textwidth]{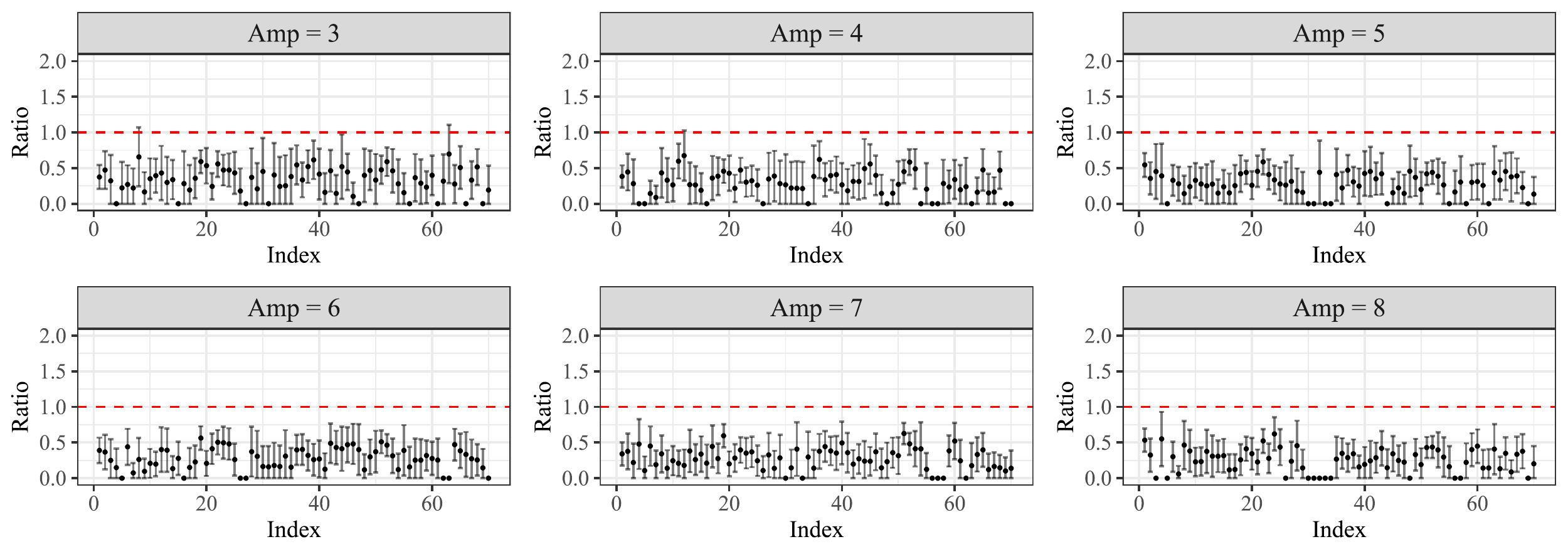}
        \caption{Realized ratios between $\p(\Pi_j\ge 1/2)$ and
          $\E[\Pi_j]$, along with $95\%$ confidence intervals,
          estimated from $1,000$ repetitions. The setting is that from
          Figure \ref{fig:pfer_amp_gaussian_within}. We caution the
          reader when interpreting the reported values of the ratios
          since there is a multiple selection issue at play here. The
          largest observed ratios indicate that the means are larger
          but also that the difference between the empirical means and
          the true means may also be large. To drive this point home,
          suppose we are to look at $m$ random variables all with
          means less than 0.8, say. Suppose each empirical mean has a
          standard deviation equal to 0.1. Then we would expect to see
          a few empirical means above 1 even though all theoretical
          means are below 0.8.}
	\label{fig:pfer_amp_gaussian_ratio}
\end{figure}

\begin{proposition}
  \label{prop:asst_ratio}
  \begin{enumerate}[label=(\alph*)]
    \item Assume the pmf of $\Pi_j$ is monotonically non-increasing
      for each $j\in\calH_0$, then condition \eqref{eq:ratio_bnd}
      holds with $\gamma$ being the optimal value of the following
      linear program (LP):
  \begin{equation}
    \label{eq:cvx_opt}
    \begin{array}{ll}
      \text{\em maximize} & \quad \sum_{m\ge M\eta} \,\, y_m \\
      \text{\em subject to} & \quad y_m \ge 0,\\
                      & \quad y_{m-1}\ge y_m,~m\in[M],\\
                      & \quad \sum^M_{m=0}y_m\, m/M = 1.
                      \end{array} 
  \end{equation}
\item Assume that for any $j\in \calH_0$, 
	\begin{align}
    \label{eq:asst_partialsum}
    \int^\eta_0\p (\Pi_j \in [\eta-u,\eta))\dd u \geq \int^{\eta-1/M}_0 \p (\Pi_j\in [\eta,\eta+u))\dd u,
	\end{align}
  then condition \eqref{eq:ratio_bnd} holds with $\gamma$ being the  optimal value of the following
      linear program:
      \begin{equation}
        \label{eq:opt_partialsum}
        \begin{array}{ll}
    \text{\em maximize} & \quad \sum_{m\ge M\eta} y_m\\
  \text{\em subject to} & \quad y_m\ge 0,~m\in \{0,1,2,\ldots,M\},\\
         & \quad \sum_{m=0}^{M} y_m\,\, {m}/{M}=1,\\
                    & \quad \sum_{m=0}^M y_m\ge 2,\\
          & \quad \sum^{\ceil{\eta M}-1}_{m=1}my_m \ge \sum^{\floor{2\eta
            M-1} - \ceil{\eta M}}_{m=0}\left(2\eta M -1  -\ceil{\eta
            M}-m\right) y_{\ceil{\eta M}+m}.
    \end{array}
  \end{equation}

\item As a special case of (b), assume that for any $j\in \calH_0$,
  the pmf of $\Pi_j$ is unimodal where the mode is less than or equal
  to $\eta$ and $\p(\Pi_j = 0) \ge \p(\Pi_j = \ceil{\eta M}/M)$. Then
  \eqref{eq:asst_partialsum} holds and, therefore,
  \eqref{eq:ratio_bnd} holds with $\gamma$ being the optimal value of
  \eqref{eq:opt_partialsum}.

\item Suppose there exists a constant $\beta \in [0,1]$ such
    that for any $j\in\calH_0$, the pmf of $\Pi_j$ satisfies
\begin{align}
\label{eq:asst_sharperdecay}
  \p(\Pi_j = m/M) \le \beta \cdot\p (\Pi_j = (m-1)/M),
\qquad\text{for }m\in[M].
\end{align}
Then condition \eqref{eq:ratio_bnd} holds with $\gamma$ being
the optimal value of the LP, 
\begin{equation}
\label{eq:cvx_opt_sharp}
  \begin{array}{ll}
    \text{\em maximize} & \quad \sum_{m\ge M\eta} \,\, y_m \\
    \text{\em subject to} & \quad y_m \ge 0,\\
                      & \quad \beta y_{m-1}\ge y_m,~m\in[M],\\
                      & \quad \sum^M_{m=0}y_m\, m/M = 1.
  \end{array} 
\end{equation}
\end{enumerate}
\end{proposition}
For illustration, Figure \ref{fig:ratio_M} plots the optimal value of
\eqref{eq:cvx_opt} versus $M$ with $\eta=0.501$, $0.751$ and $1$,
respectively. Taking $M=31$ and $\eta=0.501$ for example, we see that
\eqref{eq:ratio_bnd} holds with $\gamma = 1$. Also, and this is
important for later, \eqref{eq:ratio_bnd} holds with $\gamma = 1$ for
$M = 31$, $\eta = 1/2$.

The proof of Proposition~\ref{prop:asst_ratio} is deferred to Appendix
\ref{PfThmFB} and we pause here to parse the claims. The
monotonicity assumption in part (a) states that the chance that a null
variable gets selected 50 times is at most that it gets selected 49
times, which is at most that it gets selected 48 times and so
on. (When we say chance, recall that the probability is taken over
$\bmx,\bmY,\bmtx^1, \ldots, \bmtx^M$.) In part (b),
\eqref{eq:asst_partialsum} is a relaxed version of the monotonicity
condition. To be sure, if the pmf of $\Pi_j$ is monotonically
non-increasing, then \eqref{eq:asst_partialsum} holds.  Setting
$F_-(x) := \p(X<x)$, condition \eqref{eq:asst_partialsum} says this:
the area between the two curves $y=F_-(x)$ and $y=F_-(\eta)$ (the
latter does not vary with $x$) over the interval $[0, \eta]$---the
blue area in Figure \ref{fig:ex_dist}---is larger than the area
between the same two curved curves over $[\eta, 2\eta-1/M]$---the red
area in Figure \ref{fig:ex_dist}.
In other words, the pmf of $\Pi_j$ is skewed towards the left as
illustrated in Figure \ref{fig:ex_dist}(b). 
Part (d) shows that we can sharpen the bound (as illustrated in Figure
\ref{fig:ratio_M_gamma_1}) if the pmf of $\Pi_j$ decays at a faster
rate---$\p(\Pi_j = m/M)\le \beta \p(\Pi_j=(m-1)/M)$---where the
smaller $\beta$, the faster the decay. In this paper, we only consider
$\beta = 1$ (the weakest possible condition), which just says that the
pmf is monotonically non-increasing.

\begin{figure}[ht]
  \centering
  \begin{minipage}{0.49\textwidth}
    \centering
  \begin{tikzpicture}[scale=0.8]
    \draw[fill = blue!10,draw = blue!0](0,0.5) rectangle (0.5,2.3);
    \draw[fill = blue!10,draw = blue!0](0.5,1.1) rectangle (1,2.3);
    \draw[fill = blue!10,draw = blue!0](1,1.6) rectangle (1.5,2.3);
    \draw[fill = blue!10,draw = blue!0](1.5,2) rectangle (2,2.3);
    \draw[fill = red!10,draw = blue!0](2.5,2.6) rectangle (3,2.3);
    \draw[fill = red!10,draw = blue!0](3,2.8) rectangle (3.5,2.3);
    \draw[fill = red!10,draw = blue!0](3.5,3) rectangle (4.5,2.3);
    \draw[thick,color=gray,step=.5cm,dashed] (0,0) grid (5,3);
    \draw[->] (-1,0) -- (5.5,0)
    node[below right] {$x$};
    \draw[->] (0,-1) -- (0,3.5)
    node[left] {$F_-(x)$};
    \node at (2.5,-0.5) {$\eta$};
    \node[scale=.7] at (4.5,-0.5) {$2\eta-\frac{1}{M}$};
    \draw[thick,color=blue](0,.5) -- (0.5,.5);
    \draw[thick,color=blue](0.5,1.1) -- (1,1.1);
    \draw[thick,color=blue](1,1.6) -- (1.5,1.6);
    \draw[thick,color=blue](1.5,2) -- (2,2);
    \draw[thick,color=blue](2,2.3) -- (2.5,2.3);
    \draw[thick,color=red](2.5,2.6) -- (3,2.6);
    \draw[thick,color=red](3,2.8) -- (3.5,2.8);
    \draw[thick,color=red](3.5,3) -- (4.5,3);
    \draw[thick,color=red!60,densely dashdotted](0,2.3) -- (5,2.3);
    \draw[thick,color=red!60,densely dashdotted](2.5,0) -- (2.5,3);
  \end{tikzpicture}\\
  (a) $F_-(x) = \mathbb{P}(X<x)$.
  \label{fig:cdf}
  \end{minipage}
  \hfill
  \begin{minipage}{0.49\textwidth}
    \centering
  \begin{tikzpicture}[scale=0.8]
    \draw[fill = blue!10,draw = blue!50](0,0) rectangle (0.5,2);
    \draw[fill = blue!10,draw = blue!50](0.5,0) rectangle (1,2.5);
    \draw[fill = blue!10,draw = blue!50](1,0) rectangle (1.5,2);
    \draw[fill = blue!10,draw = blue!50](1.5,0) rectangle (2,1.5);
    \draw[fill = blue!10,draw = blue!50](2,0) rectangle (2.5,1);
    \draw[fill = blue!10,draw = blue!50](2.5,0) rectangle (3,1);
    \draw[fill = blue!10,draw = blue!50](3,0) rectangle (3.5,0.5);
    \draw[fill = blue!10,draw = blue!50](3.5,0) rectangle (4,0.5);
    \draw[->] (-1,0) -- (6,0)
    node[below right] {$x$};
    \node at (0.25,-0.5) {$0$};
    \node at (0.75,-0.5) {$\frac{1}{M}$};
    \node at (2.75,-0.5) {$\eta$};
    \draw[->] (0,-1) -- (0,3.5)
    node[left] {$f(x)$};
    \draw[thick,color=blue](0,2) -- (0.5,2);
    \draw[thick,color=blue](0.5,2.5) -- (1,2.5);
    \draw[thick,color=blue](1,2) -- (1.5,2);
    \draw[thick,color=blue](1.5,1.5) -- (2,1.5);
    \draw[thick,color=blue](2,1) -- (2.5,1);
    \draw[thick,color=blue](2.5,1) -- (3,1);
    \draw[thick,color=blue](3,.5) -- (3.5,0.5);
    \draw[thick,color=blue](3.5,.5) -- (4,.5);
    \draw[thick,color=blue](4,0) -- (4.5,0);
    \draw[thick,color=blue](4.5,0) -- (5,0);
    \draw[thick,color=blue](5,0) -- (5.5,0);
    \draw[thick,color=red!60,densely dashdotted](3,0) -- (3,3);
  \end{tikzpicture}\\
  (b) Probability mass function.
  \label{fig:pmf}
  \end{minipage}
  \caption{An example of a distribution obeying
    \eqref{eq:asst_partialsum}.}
  \label{fig:ex_dist}
\end{figure}
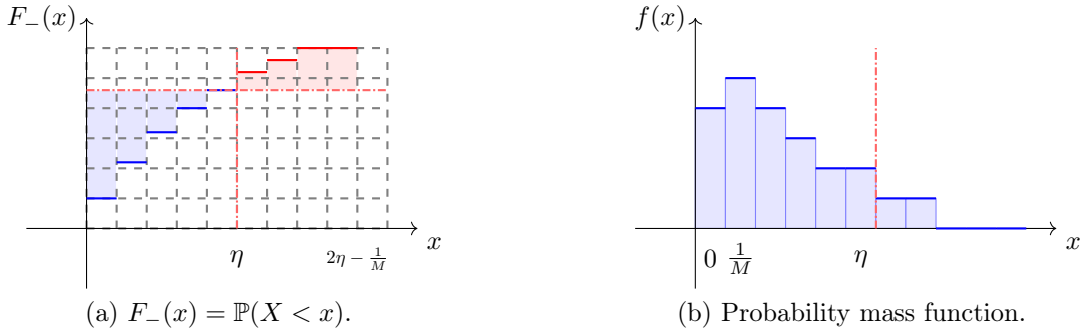

\begin{figure}[ht]
\centering
\begin{subfigure}{0.3\textwidth}
  \includegraphics[width=\textwidth]{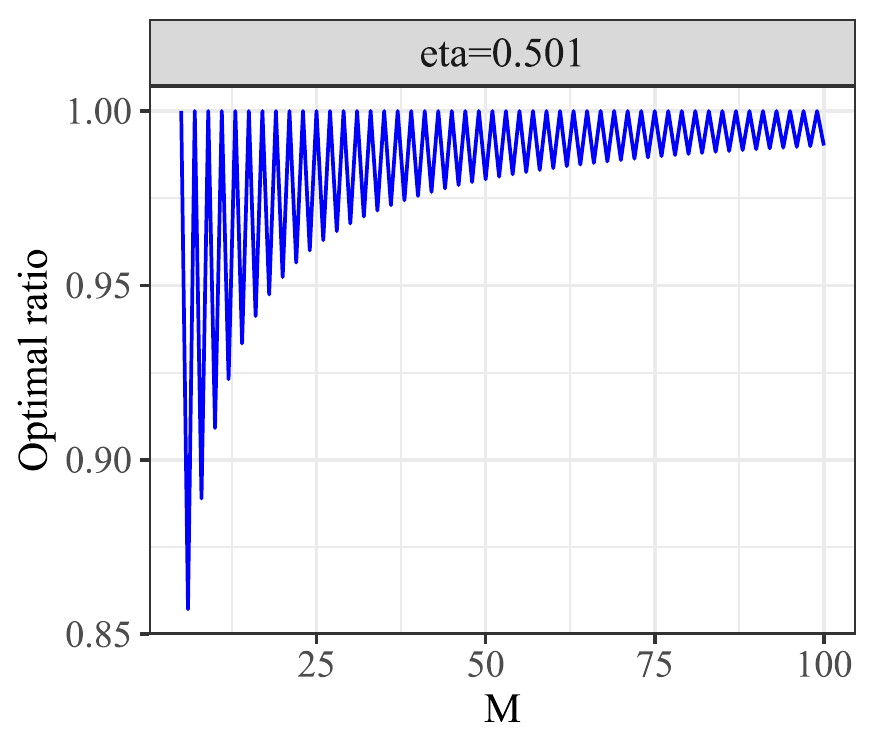}
  \caption{$\eta=0.501$}
\end{subfigure}
\hfill
\begin{subfigure}{0.3\textwidth}
  \includegraphics[width=\textwidth]{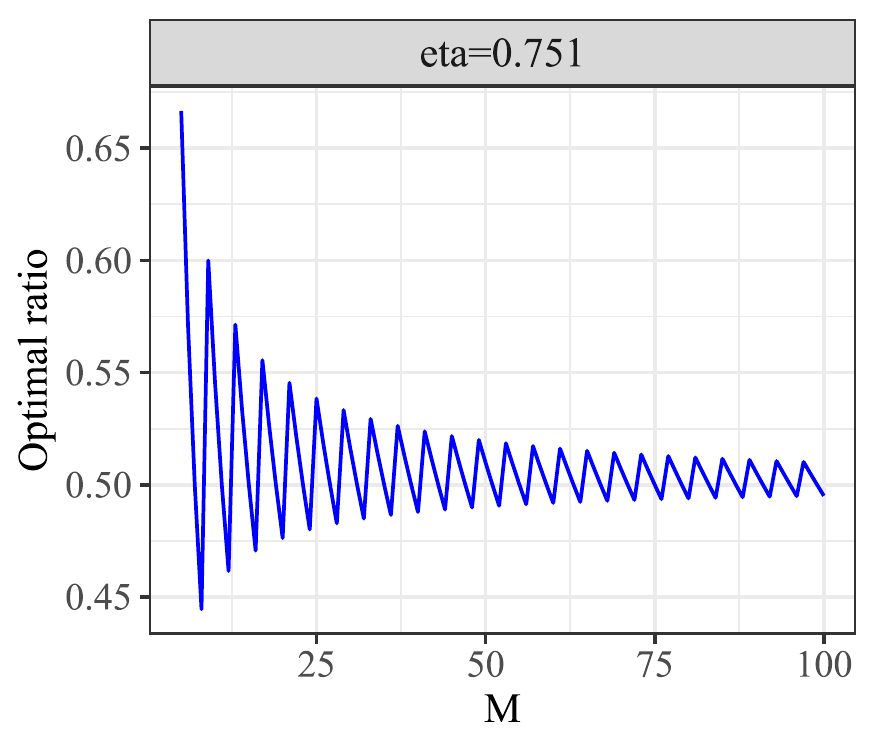}
  \caption{$\eta=0.751$}
\end{subfigure}
\hfill
\begin{subfigure}{0.3\textwidth}
  \includegraphics[width=\textwidth]{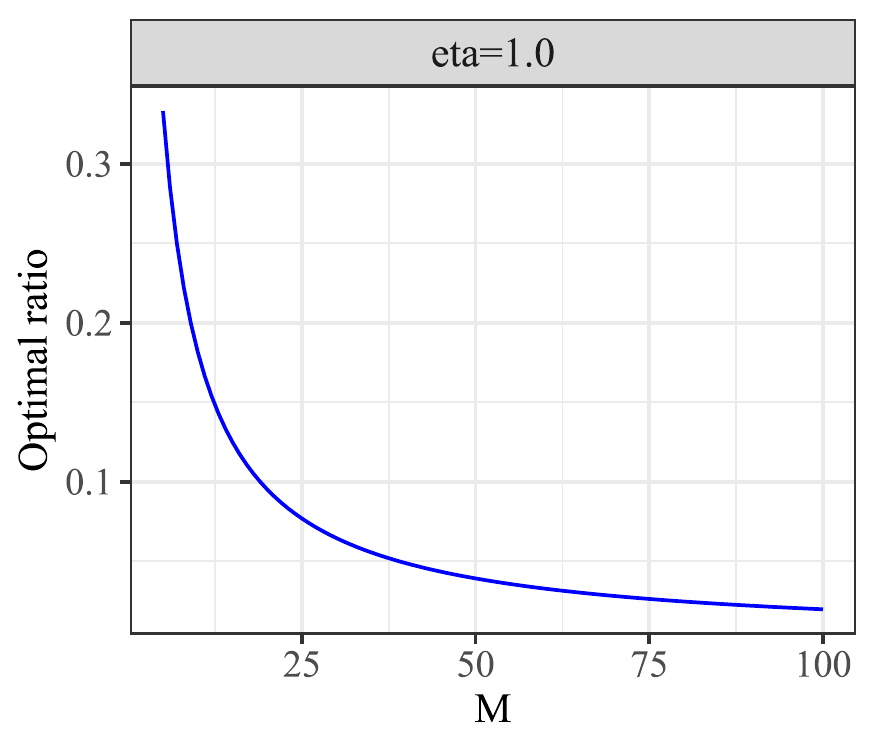}
  \caption{$\eta=1$}
\end{subfigure}
\caption{Optimal value of \eqref{eq:cvx_opt} as a function of the
number $M$ of replicates.}
\label{fig:ratio_M}
\end{figure}

\begin{figure}[ht]
\centering
\begin{subfigure}{0.3\textwidth}
  \includegraphics[width=\textwidth]{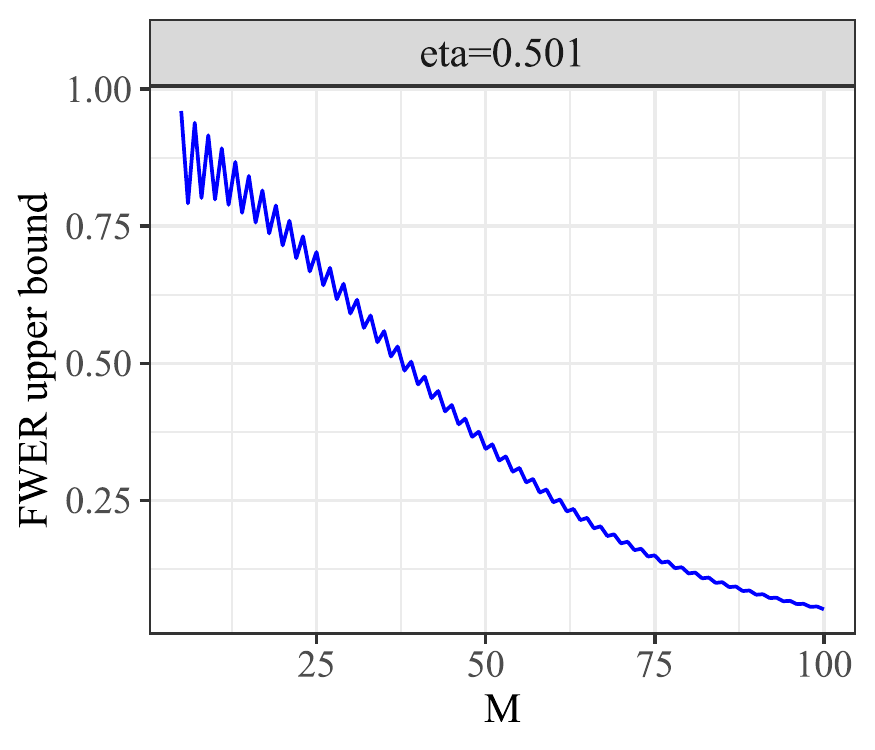}
  \caption{$\eta=0.501$}
\end{subfigure}
\hfill
\begin{subfigure}{0.3\textwidth}
  \includegraphics[width=\textwidth]{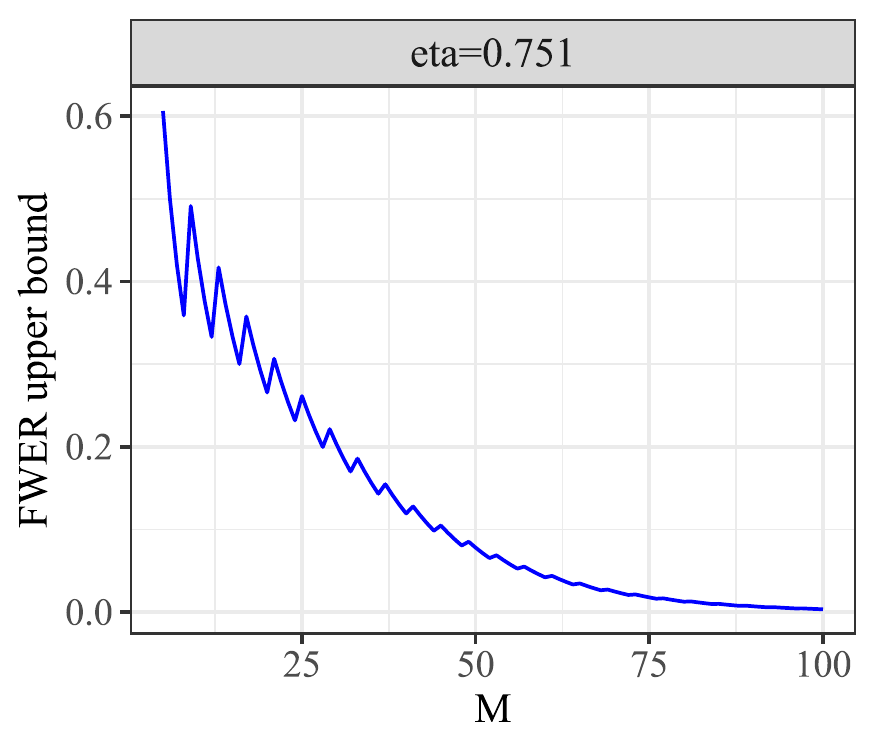}
  \caption{$\eta=0.751$}
\end{subfigure}
\hfill
\begin{subfigure}{0.3\textwidth}
  \includegraphics[width=\textwidth]{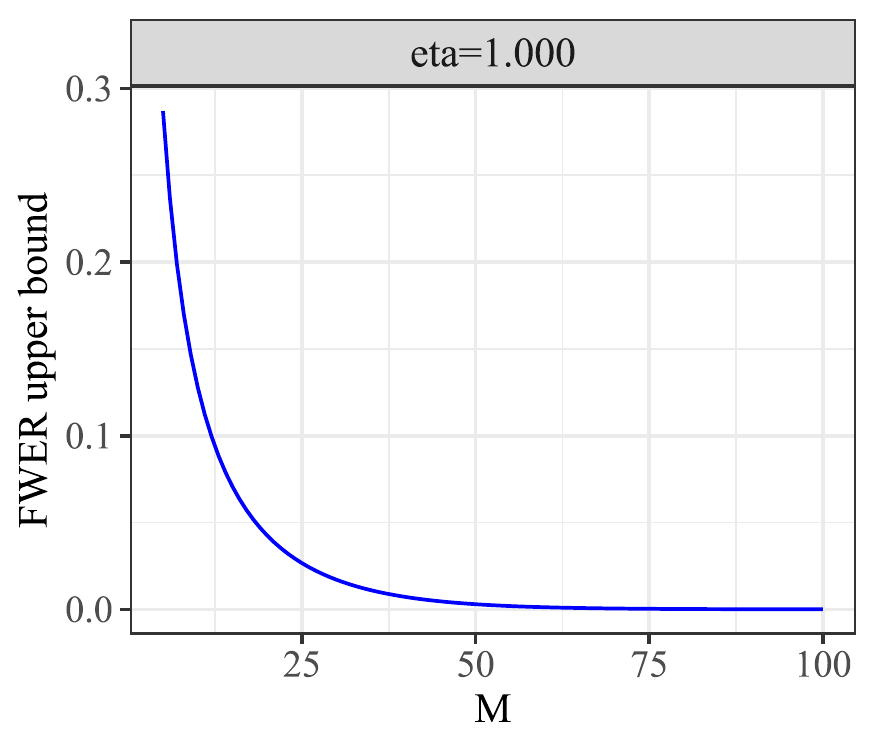}
  \caption{$\eta=1$}
\end{subfigure}
\caption{Optimal value of \eqref{eq:cvx_opt_sharp}, with $\beta=0.9$, as a function of the
number $M$ of replicates.}
\label{fig:ratio_M_gamma_1}
\end{figure}

\subsection{Numerical evaluation of the ``derandomization'' effect}
\label{sec:simulation_pfer}
To illustrate the effect of ``derandomization'', we compare
\algoname~with vanilla knockoffs in a small-scale and a large-scale
simulation study.  {Our method is implemented in the \textsf{R
    derandomKnock} package, available at
  \url{https://github.com/zhimeir/derandomKnock}; code to reproduce
  all the numerical results from this paper can be found at
  \url{https://github.com/zhimeir/derandomized_knockoffs_paper}}.  We
evaluate the difference in the type-I error, the power and the
stability of the selection set.  Throughout this section, we set
$\eta=0.5$ and $M=31$, which according to Proposition
\ref{prop:asst_ratio} yields $\E [V] \le v$ (recall that $v$ is the
nominal level of the base procedure) under the monotonicity
assumption. Here, $Y$ is generated from a linear model conditional on
the feature vector $X$, namely,
\begin{align}
Y\mid X_1,\ldots,X_p \sim \N(\beta_1 X_1+\ldots+\beta_p X_p,1).
\end{align}
As for the covariates, $X$ is drawn from a multivariate Gaussian
distribution with parameters to be specified below. We remark
that under this model, testing conditional independence is the same as testing
whether $\beta_j = 0$.


\begin{figure}[ht]
\centering
\includegraphics[width = \textwidth]{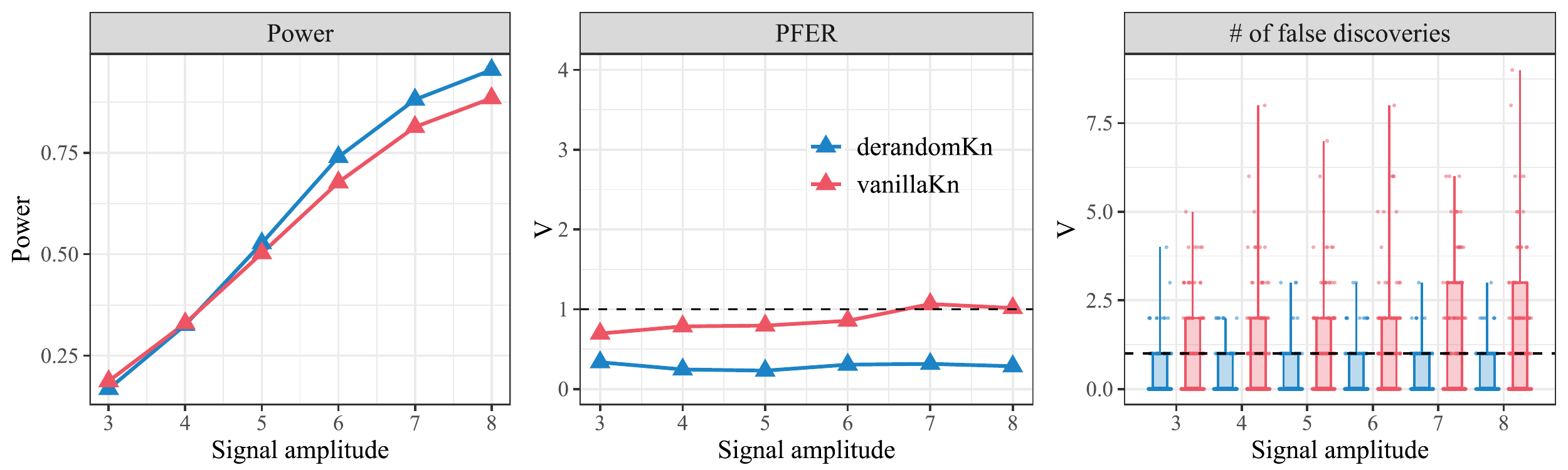}
\caption{Performance of derandomized and vanilla knockoffs in the
  small-scale study. Here, $n=200$, $p=100$,
  $X\sim\calN(\bm 0,\bm \Sigma)$ with $\Sigma_{ij}=0.6^{|i-j|}$, and
  $Y\mid X$ is generated from a linear model with $30$ non-zero
  coefficients. Each nonzero coefficient $\beta_j$ takes value
  $\pm A/\sqrt{n}$ where the signal amplitude $A$ ranges in
  $\{3,4,\ldots,8\}$ and the sign is determined by i.i.d.~coin
  flips. The locations of the non-zero signal are randomly chosen from
  $[p]$. We show the averaged results over 200 trials. The parameter
  $\bm \beta$ is fixed across trials so that the distribution of
  $(X, Y)$ does not vary. The dashed black line indicates the target
  PFER level $v=1$. In the boxplot, the box is drawn from the $10$th
  quantile to the $90$th quantile; the whiskers represent the maximum
  and the minimum of the data; each jittered dot represents a raw data
  points outside of the $[10, 90]$th percentile range.}
\label{fig:pfer_amp_gaussian_within}
\end{figure}

Figure \ref{fig:pfer_amp_gaussian_within} compares the performance of
derandomized and vanilla knockoffs in the small-scale study. The
construction of knockoffs in this study is based on a version
suggested by \citet{Spector2020private}, and we use the LCD statistic
to tease the signal and noise apart.  We can see that both procedures
control the PFER, while the power of \algoname~is slightly better than
that of vanilla knockoffs. The boxplot shows that derandomization
significantly decreases the \emph{marginal} selection variability as
claimed earlier (we additionally provide the frequencies of the number
of false discoveries resulting from both methods in Table
\ref{tab:pfer_amp_gaussian} in Appendix \ref{sec:extra_tables}).

PFER control is theoretically guaranteed with our parameter choices
since the ratio between $\p(\Pi_j\ge 1/2)$ and $\E[\Pi_j]$ is below
one for all null variables $j$, as seen earlier in Figure
\ref{fig:pfer_amp_gaussian_ratio}.  A different way to establish
validity is to check the monotonicity condition from Proposition
\ref{prop:asst_ratio} (which in turn implies that none of the ratios
exceed one).
Figure~\ref{fig:pfer_amp_gaussian_pidist} shows the {\em pooled}
histograms of all (nonzero) null $\Pi_j$'s; the non-increasing
property of the pooled distributions is clear.


\begin{figure}[ht]
\centering
\includegraphics[width = \textwidth]{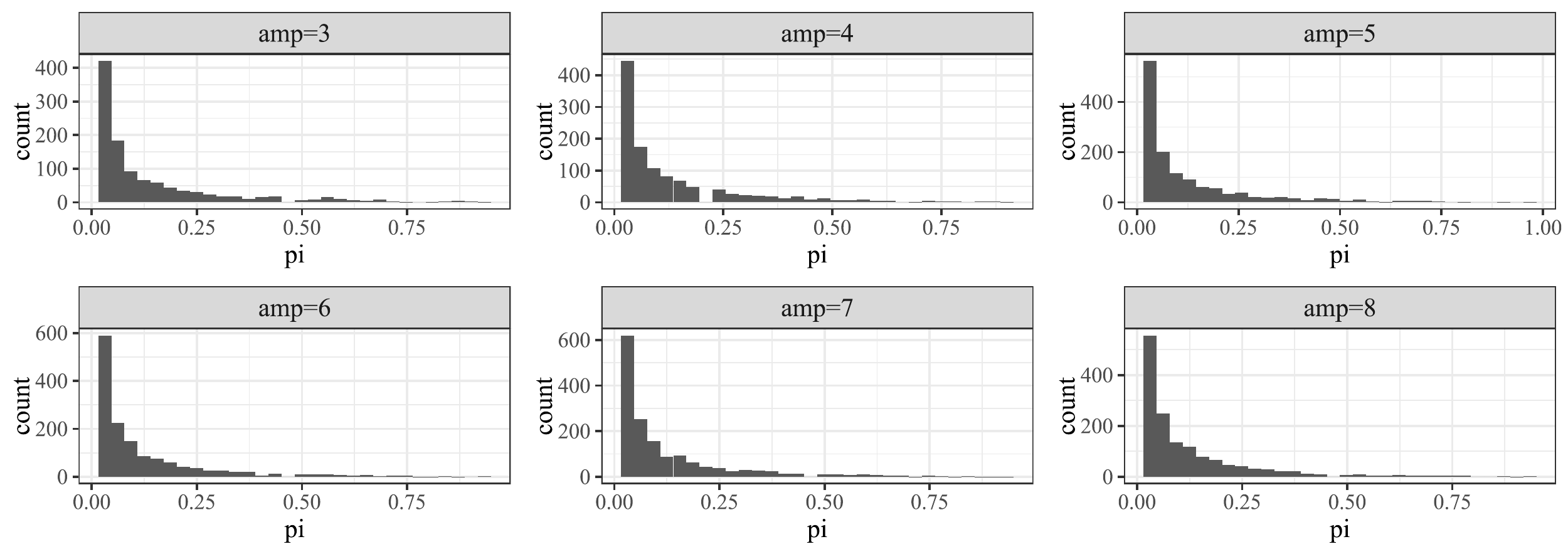}
\caption{Pooled histograms of all nonzero null $\Pi_j$'s under different signal amplitude. The experiment setting is the same as in Figure \ref{fig:pfer_amp_gaussian_within}.}
\label{fig:pfer_amp_gaussian_pidist}
\end{figure}

Our large-scale study uses the same LCD statistic and $M=31$, with the
knockoff construction based on the version described in
\citet{candes2018panning}.  As shown in Figure
\ref{fig:pfer_amp_gaussian_large_within}, we observe similar results:
both derandomized and vanilla knockoffs control the PFER as expected;
the power of \algoname~is slightly higher than that of the vanilla
knockoffs; and \algoname~has a lower marginal variability than vanilla
knockoffs.  In addition, the frequencies of the number of false
discoveries $V$ are recorded in Table
\ref{tab:pfer_amp_gaussian_large} for both methods respectively; the
pooled histograms of all null $\Pi_j$'s are constructed in Figure
\ref{fig:pfer_amp_gaussian_large_pidist} and we refer the readers to
Appendix \ref{sec:additional_figs} for other diagnostic plots.

\begin{figure}[ht]
	\centering
	\includegraphics[width = \textwidth]{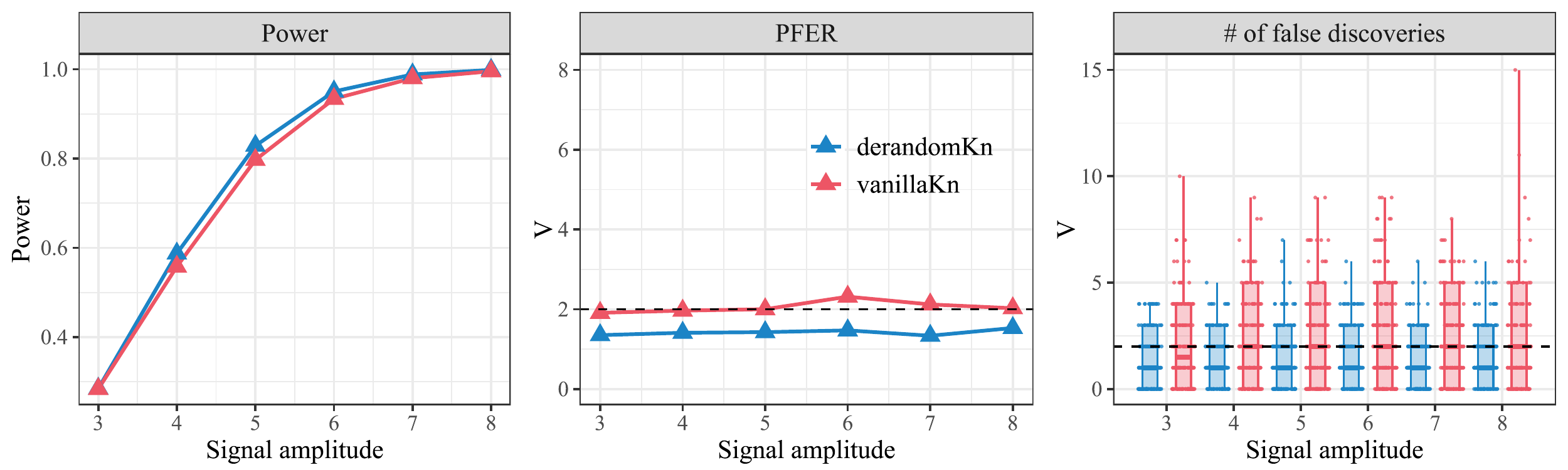}
        \caption{Performance of derandomized and vanilla knockoffs in
          the large-scale setting.  Here, $n=2000$, $p=1000$ and
          $\Sigma_{ij}=0.5^{|i-j|}$. $Y\mid X$ is generated from a
          linear model with $60$ non-zero coefficients. The experiment
          settings are otherwise the same as in Figure
          \ref{fig:pfer_amp_gaussian_within}. The dashed black line
          corresponds to the target PFER level $v=2$. The construction
          of boxplots is as in Figure
          \ref{fig:pfer_amp_gaussian_within}.
}
  \label{fig:pfer_amp_gaussian_large_within}
\end{figure}

\subsection{Improving assumption-free guarantees}
\label{sec:free_theory}
We are primarily interested in a  relatively large number 
$M$ of repetitions in order to enable stable decision making.
In the case
where $M$ is low, it is possible to significantly improve on the
assumption-free bound  $\E [V] ~\le~ v/\eta$. 
Indeed, assuming that the knockoff features are
conditionally i.i.d., then the number of selections $M\Pi_j$ is binomial
conditional on $\bmx$ and $\bmY$,
i.e.~$M\Pi_j | \bmx, \bmY \sim \operatorname{Bin}(M,\mathbb{P}(j\in
\hat{\calS}^1 \mid \bmx,\bm{Y}))$. Therefore, the PFER can be directly computed via
\begin{align*}
  \E [V]  = \Exs \left[\sum_{j\in \calH_0} \p(j\in \hat{\calS}\mid \bmx, \bm{Y})\right]
   &= \Exs \left[ \sum_{j\in \calH_0} \p (M\Pi_j\geq M\eta \mid \bmx, \bm{Y}) \right].
\end{align*}
For example, let us consider the case $\eta = 0.5$ and $M = 3$.
Setting $p_j \defn \mathbb{P}(j\in \hat{\calS}^1\mid \bmx, \bm{Y} )$
gives
\begin{align*}
\E [V]  &= \Exs \left[\sum_{j\in \calH_0}  p^3_j + 3 p_j^2(1 - p_j) \right] = 
\Exs \left[\sum_{j\in \calH_0}p_j \left(p_j^2 + 3 p_j(1 - p_j)\right) \right] \leq 1.125 v,
\end{align*}
where the last inequality follows from two facts, namely,
\begin{align}
\label{eq:max}
  \max_{x\in [0,1]}(x^2 + 3x(1-x)) = 1.125 ~\text{ and }~
  \sum_{j\in \calH_0} \mathbb{P}(j\in \hat{\calS}^1) \leq v.
 \end{align} 
 This bound is of course better than $\E [V] \le 2v$ from
 \eqref{eq:markov}.  Such calculations for any value of
 $\eta$ and $M$ provide a tighter control and can be carried out in a
 completely offline fashion. In this spirit, Figure
 \ref{fig:assptfree_bounds} shows a valid assumption-free bound for a
 number of repetitions ranging from $2$ to $50$.
\begin{figure}[h]
\centering
\begin{subfigure}{0.47\textwidth}
\centering
  \includegraphics[width =.8\textwidth]{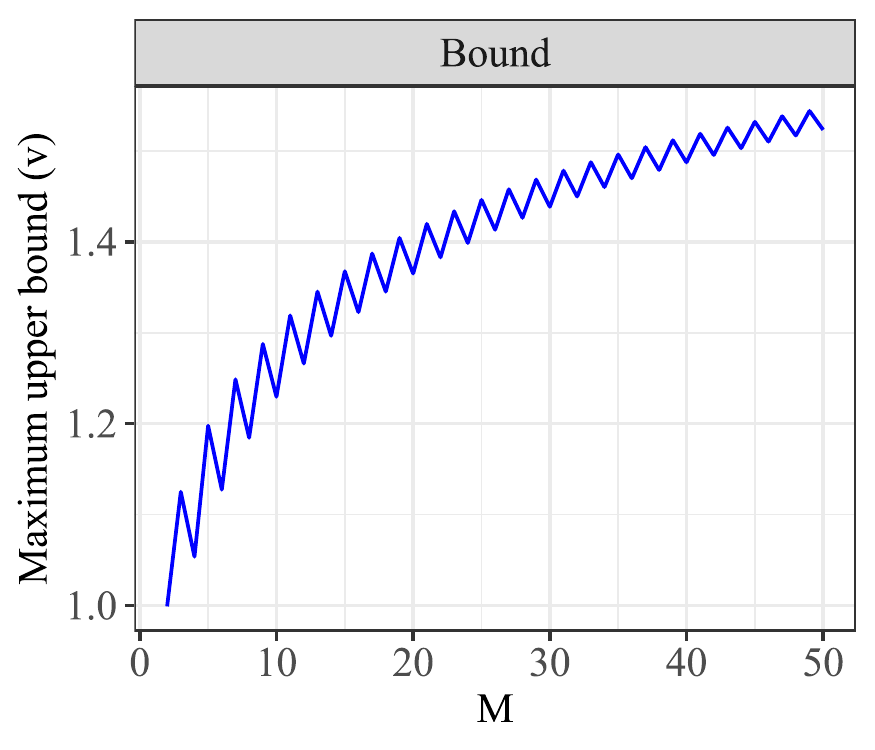}
\end{subfigure}
\begin{subfigure}{0.47\textwidth}
\centering
  \includegraphics[width = .8\textwidth]{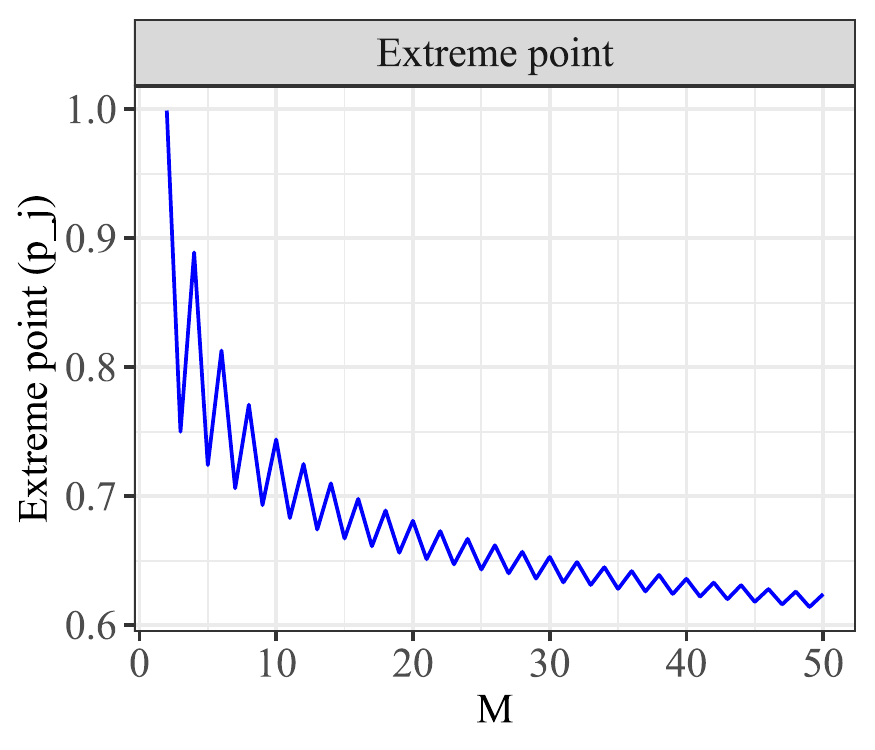}
\end{subfigure}

\caption{(a) Refined assumption-free PFER bounds as a function of the
  number of repetitions. The PFER bound is the number reported on the
  y-axis times the PFER of the base procedure. (b) Value of $p_j$ at
  which the upper bound as in \eqref{eq:max} (shown for $M = 3$)
  is achieved.}
\label{fig:assptfree_bounds} 
\end{figure}

\section{Theoretical guarantees: controlling the $k$-FWER}
\label{sec:fwer}
Another widely used type-I error measure is the \emph{$k$ family-wise
  error rate} ($k$-FWER): defined as the probability of making at
least $k$ false discoveries, $\kFWER = \p(V\ge k)$.  Dating back to
Bonferroni \citep{dunn1961multiple} and \cite{holm1979simple}, many
procedures guaranteeing $\kFWER$ control have been proposed. Most
operate on p-values and many require various assumptions on the
dependence structure between these p-values (see,
e.g.~\citet{karlin1980classes,hochberg1988sharper,benjamini2001control,romano2010balanced}).
We refer the readers to \cite{guo2014further,duan2020familywise}
and the references therein for a survey of these methods.

We now demonstrate how to tune the parameters for derandomized
knockoffs to control the $\kFWER$. Our exposition parallels that from
the previous section.
\begin{theorem}
\label{thm:fdx_finerbnd}
Let $V$ be the number of false discoveries after applying derandomized
knockoffs (Algorithm \ref{algo.skn}) with a base procedure obeying
$\text{PFER} \le v$. Suppose condition~\eqref{eq:ratio_bnd} holds and
that for each $k\ge 1$, 
\begin{align}
  \label{eq:condition_v}
  \p(V\ge k) ~\le~ \dfrac{\rho \E[V]}{k}.
\end{align}
Then the $\kFWER$ is controlled via 
\begin{align}
  \label{eq:explicitbnd}
  \p(V\ge k)~\le~ \dfrac{\rho \gamma v}{k}.
\end{align}
In particular, by Markov's inequality one always has $\rho=1$, {and consequently,} 
\[\p(V\ge k)~\le~ \gamma v/k.\]
\end{theorem}
The proof of this result is straightforward since we have 
\[
  \p(V\ge k)\le \dfrac{\rho \E[V]}{k} \le \dfrac{\rho \gamma v}{k}, 
\]
where the last inequality follows from Theorem \ref{cor:monotone}.  

Set $h(x) := x$ and let $Z\sim\mathrm{NB}(v,1/2)$, where
$\mathrm{NB}(m,q)$ denotes a negative binomial random variable, which
counts the number of successes before the $m$-th failure in a sequence
of independent Bernoulli trials with success probability $q$. With
this, the right-hand side of \eqref{eq:explicitbnd} can be expressed
as $\rho \gamma\E[h(Z)]/k$ (by simply observing that $\Exs[Z] = v$).
This leads to the following extension:
\begin{corollary}
\label{cor:kfwer_genbnd}
Let $h :\R\mapsto \R$ be a convex, non-negative and non-decreasing
function.  In the setting of Theorem \ref{thm:fdx_finerbnd}, suppose
that
\begin{align}
  \label{eq:condition_hofv}
  \p(V\ge k) \le \dfrac{\rho \E[h(V)]}{h(k)}.
\end{align}
Then the $\kFWER$ obeys 
\begin{align}
  \label{eq:kfwer_generalbnd}
  \p(V\ge k)\le \dfrac{\rho \E[h(Z/\eta)]}{h(k)},
  \qquad 
  Z \sim \mathrm{NB}(v,1/2).
\end{align}
In particular, Markov's inequality shows that \eqref{eq:kfwer_generalbnd}
always holds with $\rho = 1$.
\end{corollary}
The proof of Corollary \ref{cor:kfwer_genbnd} is deferred to Appendix \ref{sec:PfCorfdx}.

\subsection{Guarantees under mild assumptions}
\label{sec:fwer_mildassumption}



While \eqref{eq:condition_v} holds with $\rho = 1$, we observe in
simulations that this value is often quite conservative and we give an
example where \eqref{eq:condition_v} holds with $\rho = 1/2$. The
proof and an extension of the proposition below are given in Appendix
\ref{sec:proof_of_finerfwer}.
\begin{proposition}
\label{prop:fwer_finerbnd}
In the setting of Theorem \ref{thm:fdx_finerbnd}, suppose the pmf of
$V$ is skewed to the left of $k$ in the sense that
\begin{align}
  \label{eq:asst_v}
   \sum_{u=1}^{k-1} \, \p (V\in [k-u,k)) ~\geq~ \sum_{u=1}^k \,\p (V\in[k,k+u))
\end{align}
(observe the similarity with \eqref{eq:asst_partialsum}). Then
condition \eqref{eq:condition_v} holds with $\rho = 1/2$.
\end{proposition}

In applications, $k$ and $\alpha$ are supplied and we provide below
some guidance on the selection of $v$ and $\eta$ to control the
$k$-FWER at level $\alpha$. In order to do so, however, we must extend
the base procedure to control the PFER at levels $v$ which may not be
integer valued. 

\paragraph{Non-integer $v$}
Let $\floor{v}$ be the integer part of $v$ and sample a random
variable $U\sim \mathrm{Bern}(v-\floor{v})$. If $U=1$, run the
$(\floor{v}+1)$-knockoffs and $\floor{v}$-knockoffs otherwise. It is
easy to see that such an algorithm controls the PFER at level $v$.

\paragraph{Choices of parameters for larger values of $k$}
When $k\ge \max(1/(4\alpha),2)$, we shall fix $\eta$ to be $0.5$ and
choose $M$ such that the optimal value of \eqref{eq:cvx_opt} is
one. 
In this way, all the reported variables are selected at least half of
the times by the base procedure. We then set $v = 2k\alpha$.  For
example, when $k=3$ and $\alpha=0.1$, we simply set $v =0.6$ and use
the base procedure at level 0.6 as above. 

\paragraph{Choice of parameters for smaller values of $k$}
In the case where $k=1$ or $k<1/(4\alpha)$, we fix $v=1$.  When $k=1$,
note that the LHS in \eqref{eq:asst_v} vanishes so we cannot make use
of this condition. In order to control the FWER at the nominal level
$\alpha$, it suffices to control $\gamma$ by $\alpha$ since
$\p(V\ge 1)\le \E[V] \le \gamma$. As shown in Figure
\ref{fig:ratio_M}, under the monotonicity assumption of each $\Pi_j$,
pairs $(\eta, M)$ yield different values of $\gamma$. Among those
pairs for which $\gamma\le \alpha$, we shall select the pair with the
smallest value of $\eta$. For example, suppose we wish to have
$\text{FWER} \le 0.1$. Then Figure \ref{fig:alpha_bnd}(a) plots the
admissible pairs $(\eta,M)$ (under the monotonicity condition), among
which we shall use $\eta = 0.95$ and $M=20$ (one may opt to use a
larger value of $M$ to increase stability if the computational cost
allows it).  Similarly, when $k<1/(4\alpha)$, we cannot use $\eta=0.5$
and $v=2k\alpha$ because such a parameter combination would always
yield no selection since $v<\eta$. Figure \ref{fig:alpha_bnd}(b) plots
the pairs $(\eta,M)$ controlling the $2$-FWER by $0.1$ (once again,
under the monotonicity condition).  According to this plot, we can
pick $\eta = 0.81$ and $M$ within the computational limit as large as
possible in order to enhance stability. In our simulation, we shall
take $M=30$.

\begin{figure}[h]
  \centering
  \begin{minipage}{0.49\textwidth}
  \centering
  \includegraphics[width = 0.8\textwidth]{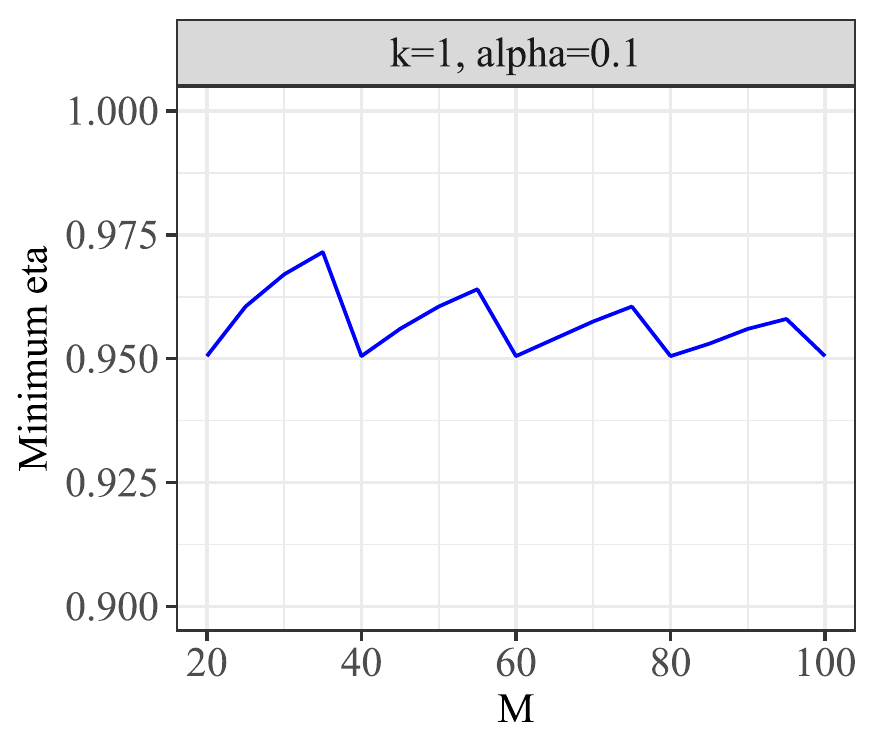}
  \caption*{(a) $k=1$ and $\alpha=0.1$.}
  \end{minipage}
  \begin{minipage}{0.49\textwidth}
  \centering
    \includegraphics[width = 0.8\textwidth]{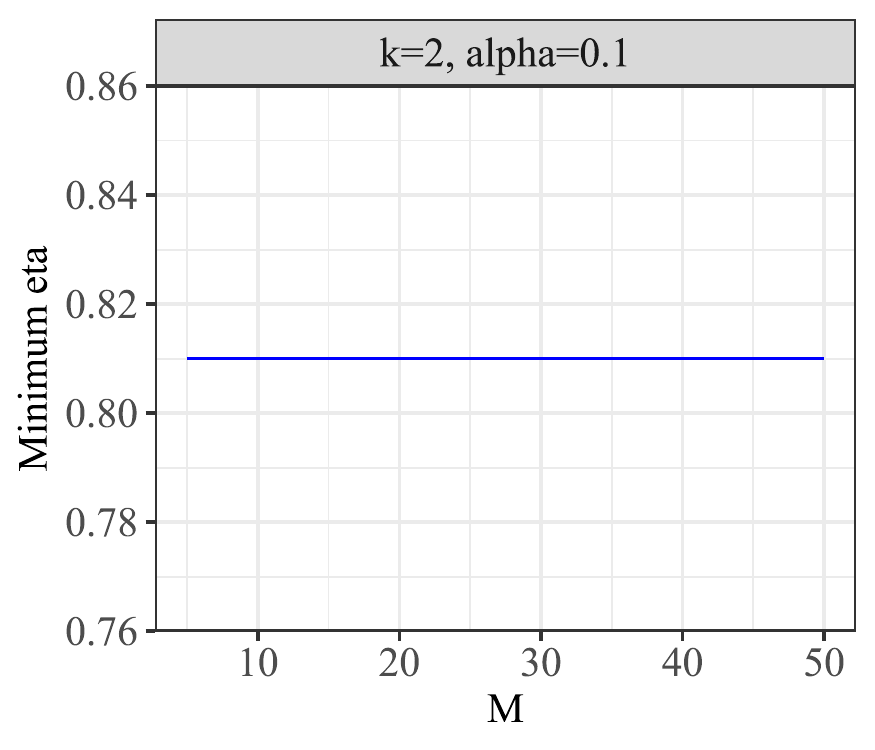}
    \caption*{(b) $k=2$ and $\alpha=0.1.$}
  \end{minipage}
  \caption{Pairs $(\eta,M)$ controlling the $k$-FWER control at level $\alpha$ with $v=1$.}
  \label{fig:alpha_bnd}
\end{figure}

\subsection{Numerical evaluation of the derandomization effect} 
We perform two numerical experiments to gauge the performance of
derandomized knockoffs. In this study, the response $Y$ is sampled
from a logistic model
\begin{align}
\label{eq:logismdl}
  Y\mid X_1,\ldots,X_p\sim \mathrm{Bern} 
~\left(\dfrac{\exp(\beta_1X_1+\ldots+\beta_pX_p)}{1+\exp(\beta_1X_1+\ldots+\beta_pX_p)}\right)
\end{align}
and $X$ is drawn from a multivariate Gaussian distribution with
parameters to be specified later on.  As in Section
\ref{sec:simulation_pfer}, the vector of regression coefficients is
sparse so that most of the hypotheses are actually null; 
under this model, testing conditional independence is the same as testing
whether $\beta_j = 0$. 

We evaluate derandomized knockoffs on a small-scale and a large-scale
data set.  In the small-scale study, the knockoff construction is the
same as that from Section \ref{sec:simulation_pfer}, and the LCD
statistic is used as our importance statistic.  $M = 30$ knockoff
copies generated in each run and the selection threshold is
$\eta = 0.81$.  Under the monotonicity constraint, the value of
\eqref{eq:cvx_opt} is $0.39$.  According to Theorem
\ref{thm:fdx_finerbnd}, we thus control the $2$-FWER at level $0.1$.
Figure \ref{fig:fwer_amp_gaussian_within} displays the results of the
small-scale experiment, where derandomized and vanilla knockoffs obey
$\text{$2$-FWER} \le 0.1$. As before, the boxplot shows that
\algoname~exhibits less marginal randomness than vanilla knockoffs. At
the same time, we can clearly see a substantial power gain.

We empirically verify the monotonicity assumption and the skewness
property \eqref{eq:asst_v} by plotting the histograms of null
$\Pi_j$'s and false discoveries $V$ respectively in Figure
\ref{fig:fwer_amp_gaussian_pidist} and
\ref{fig:fwer_amp_gaussian_Vdist}.  Under the monotonicity assumption,
we expect to see the ratios obeying
$\p(\Pi_j\ge 0.81)/\E[\Pi_j]\leq 0.39$, which is indeed the case as
shown in Figure \ref{fig:fwer_amp_gaussian_ratio}.

In the large-scale experiment, the construction of knockoffs and the
feature importance statistics are the same as in Section
\ref{sec:simulation_pfer}; the number of knockoff copies is $M = 31$
and the selection threshold is $\eta=0.5$ (this yields $\gamma = 1$
under the monotonicity constraint). Figure
\ref{fig:fwer_amp_gaussian_large_within} presents the results from
which we observe that derandomized knockoffs achieves comparable
power, lower $k$-FWER, and reduced variability when compared with
vanilla knockoffs.  Additionally, we plot the ratios between
$\p(\Pi_j\ge \eta)$ and $\E[\Pi_j]$, histograms of the $\Pi_j$'s and
the number of false discoveries in Figure
\ref{fig:fwer_amp_gaussian_large_ratio}-\ref{fig:fwer_amp_gaussian_large_Vdist}
in Appendix \ref{sec:additional_figs}.

\begin{figure}[ht]
	\centering
	\includegraphics[width=\textwidth]{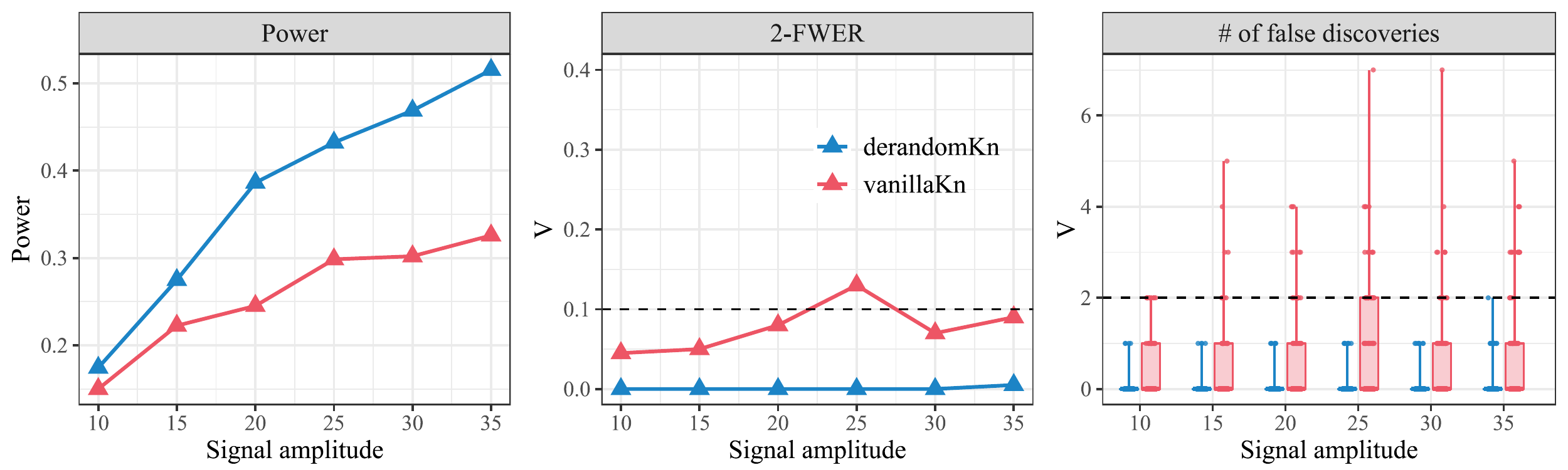}
        \caption{Performance of \algoname~($\eta = 0.81$ and $v = 1$)
          and vanilla knockoffs. The target $2$-FWER level is $0.1$.
          In this setting, $n=300$ and $p=50$,
          $X\sim \calN(0,\bm \Sigma)$ with
          $\Sigma_{ij} = 0.5^{|i-j|}$. $Y\mid X$ is sampled from a
          logistic model \eqref{eq:logismdl} with $30$ non-zero
          entries in $\bm \beta$.  These nonzero entries take values
          $\pm A/\sqrt{n}$, where the signal amplitude $A$ ranges in
          $\{10,15,\ldots,35\}$ and the sign is determined by
          i.i.d.~coin flips. The setting is otherwise the same as in
          Figure \ref{fig:pfer_amp_gaussian_within}. We indicate the
          target $2$-FWER level $\alpha=0.1$ and $\text{PFER} = 2$
          with a dashed line. Each point in the first two panels
          represents an average over $200$ replications.  The
          construction of boxplots is as in Figure
          \ref{fig:pfer_amp_gaussian_within}. Exact frequencies of the
          number of false discoveries are provided in 
          Table \ref{tab:fwer_amp_gaussian}, Appendix
          \ref{sec:extra_tables}.}
  \label{fig:fwer_amp_gaussian_within}
\end{figure}

\begin{figure}[ht]
	\centering
	\includegraphics[width = 0.95\textwidth]{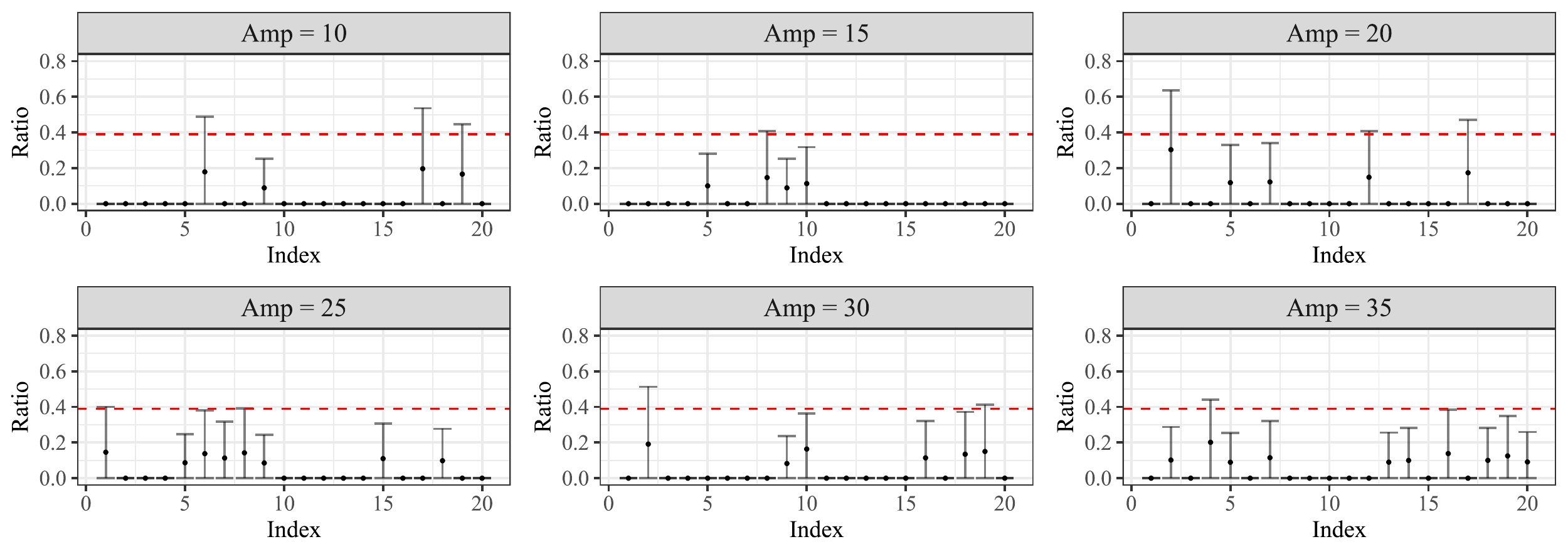}
        \caption{Realized ratios between $\p(\Pi_j\ge 0.81)$ and
          $\E[\Pi_j]$ with corresponding $95\%$ confidence intervals
          estimated from $200$ repetitions. The setting is the same as
          in Figure \ref{fig:fwer_amp_gaussian_within}. The red dashed
          line corresponds to the target upper bound of $0.39$.}
	\label{fig:fwer_amp_gaussian_ratio}
\end{figure}
\begin{figure}[ht]
	\centering
	\includegraphics[width = 0.95\textwidth]{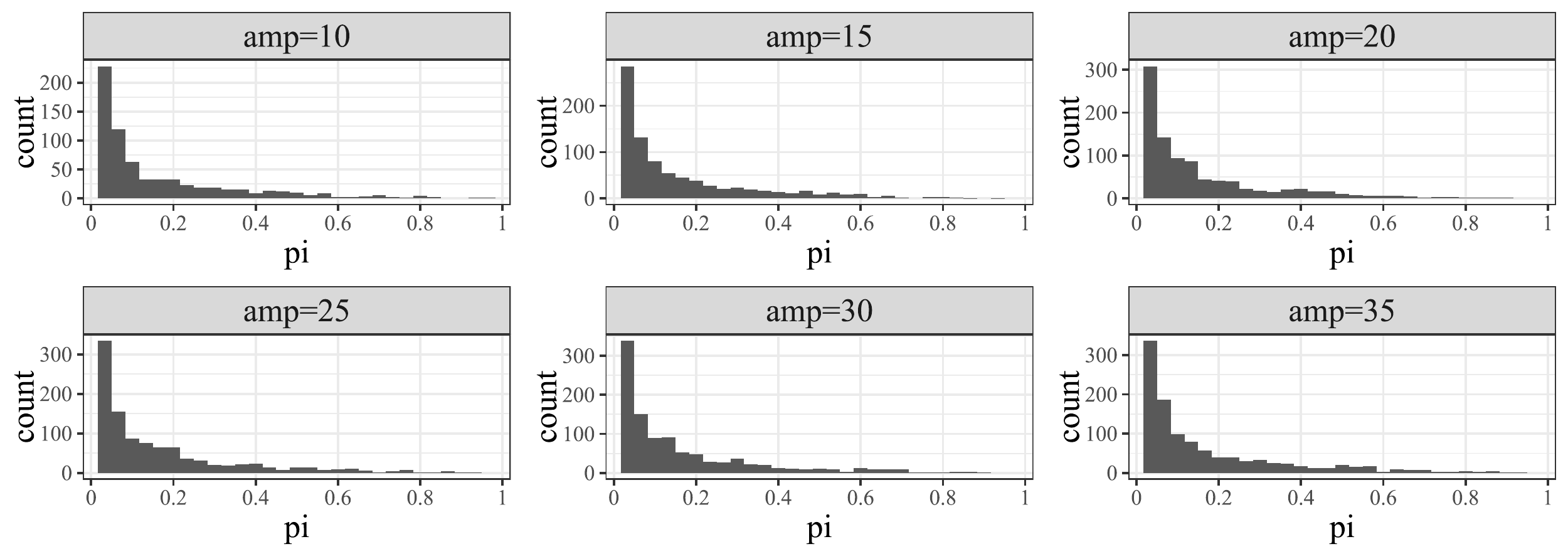}
	\caption{Pooled histograms of all nonzero null $\Pi_j$'s. The experiment setting is the same as in Figure \ref{fig:fwer_amp_gaussian_within}.}
	\label{fig:fwer_amp_gaussian_pidist}
\end{figure}

\begin{figure}[ht]
	\centering
	\includegraphics[width = 0.95\textwidth]{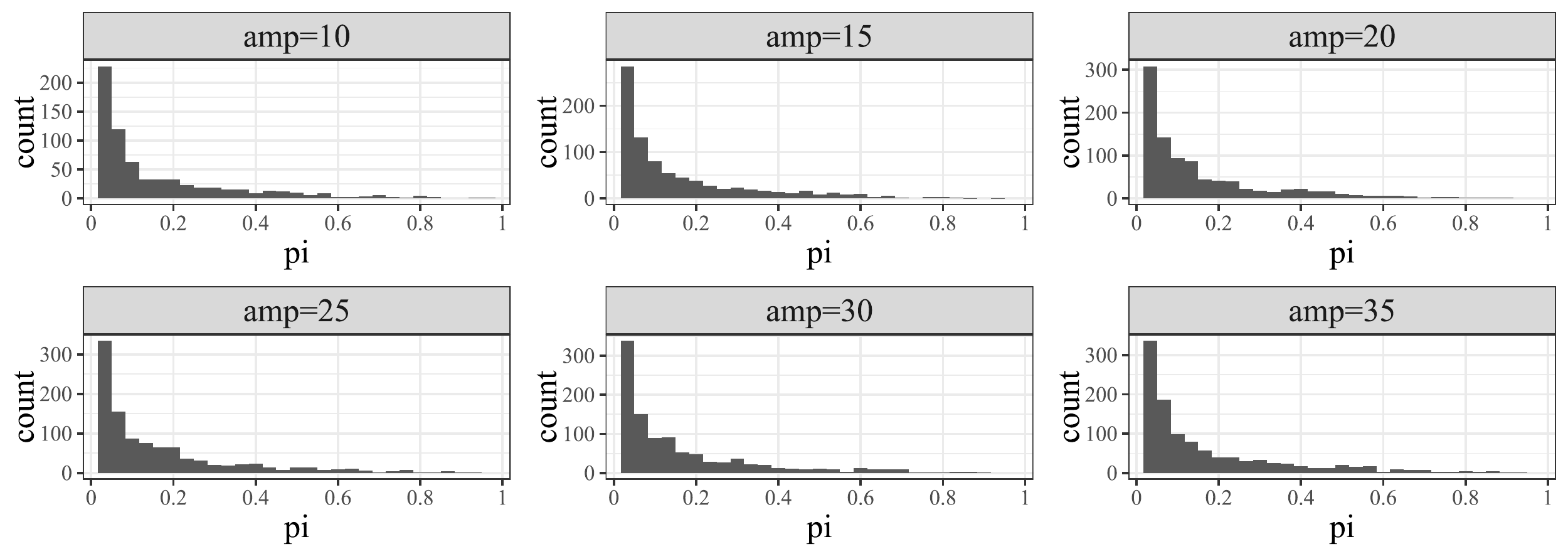}
	\caption{Histograms of the number of false discoveries
          $V$. The setting is the same as in Figure
          \ref{fig:fwer_amp_gaussian_within}.}
	\label{fig:fwer_amp_gaussian_Vdist}
\end{figure}



\begin{figure}[ht]
	\centering
	\includegraphics[width=\textwidth]{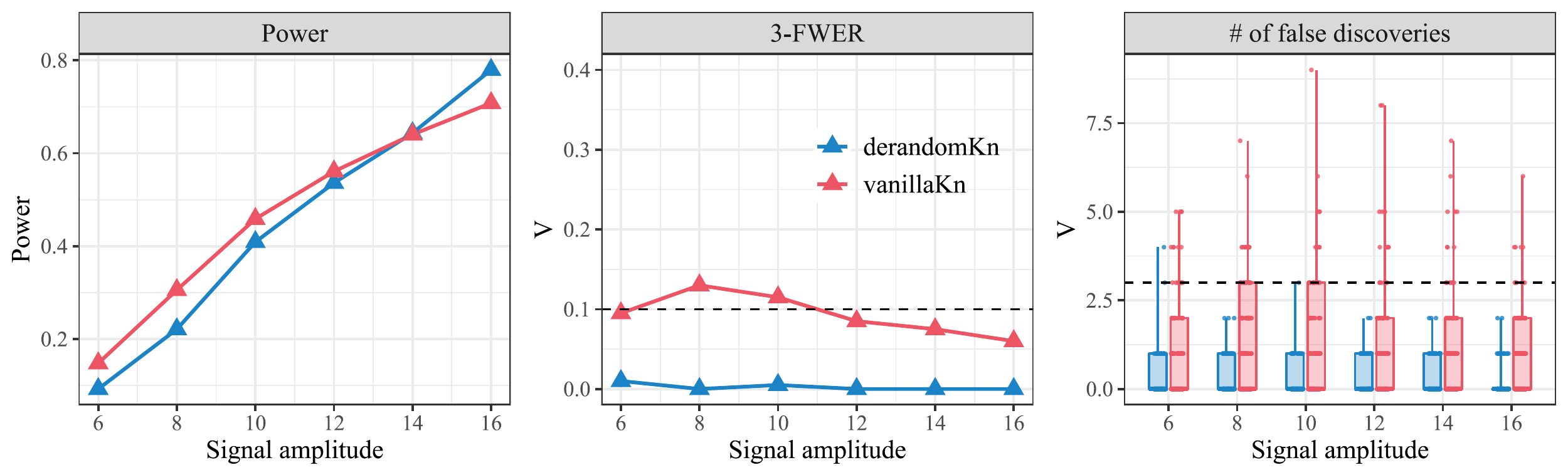}
  \caption{Results from \algoname~($\eta = 0.5$ and $v = 0.6$)
  and vanilla knockoffs. In this setting, $n=2000$, $p=1000$
  and $\Sigma_{ij} = 0.2^{|i-j|}$. The number of signals is
  $60$, and the signal amplitude $A$ ranges in $\{6,8,\ldots,16\}$.  
  The setting is otherwise the same as in Figure
  \ref{fig:fwer_amp_gaussian_within}.  We indicate the target
  3-FWER level $\alpha=0.1$ and $\text{PFER} = 3$ with a
  dashed line.  Each point in the first two panels represents
  an average over $200$ replications. The
          construction of boxplots is as in Figure
          \ref{fig:pfer_amp_gaussian_within}. }
	\label{fig:fwer_amp_gaussian_large_within}
\end{figure}

\section{Numerical simulations}
\label{sec:sim}
Thurs far, we merely demonstrated the enhanced stability of randomized
knockoffs when compared to the base procedure.  It is, however,
unclear whether this stability improvement comes at a price of a
potential power loss, as is often seen in the literature on stability
selection. This section dispels this concern by comparing our
procedure with its competitors via further numerical studies.

\subsection{PFER control}
\label{sec:numerical_pfer}
In our first experiment, we keep the two simulation settings from
Section \ref{sec:pfer} and compare \algoname~with a p-value based
method applying a Bonferroni correction and with stability selection.
The p-values used in the Bonferroni's procedure are computed from the
conditional randomization test (CRT; \citealt{candes2018panning}) in
the small-scale study, and from multivariate linear regression in the
large-scale study. The CRT p-values are in general more powerful but
require intensive computations. This is why we only calculate them in
the small-scale study. For stability selection, we use the
\textsf{stabs} R-package \citep{hofner2015stabs} and take the
selection threshold to be $0.75$ suggested by the example from the
R-package.
 
\paragraph{Small-scale study}
We can see in Figure \ref{fig:pfer_amp_gaussian} that all three
procedures control the PFER as expected.  The power of
\algoname~{is slightly better than} that of the Bonferroni procedure with CRT
p-values. Stability selection suffers a huge power loss vis a vis the
other methods due to sample splitting and a conservative choice of
threshold. 

\begin{figure}[ht!]
	\centering
  \begin{minipage}{0.45\textwidth}
    \centering
	  \includegraphics[width =.8 \textwidth]{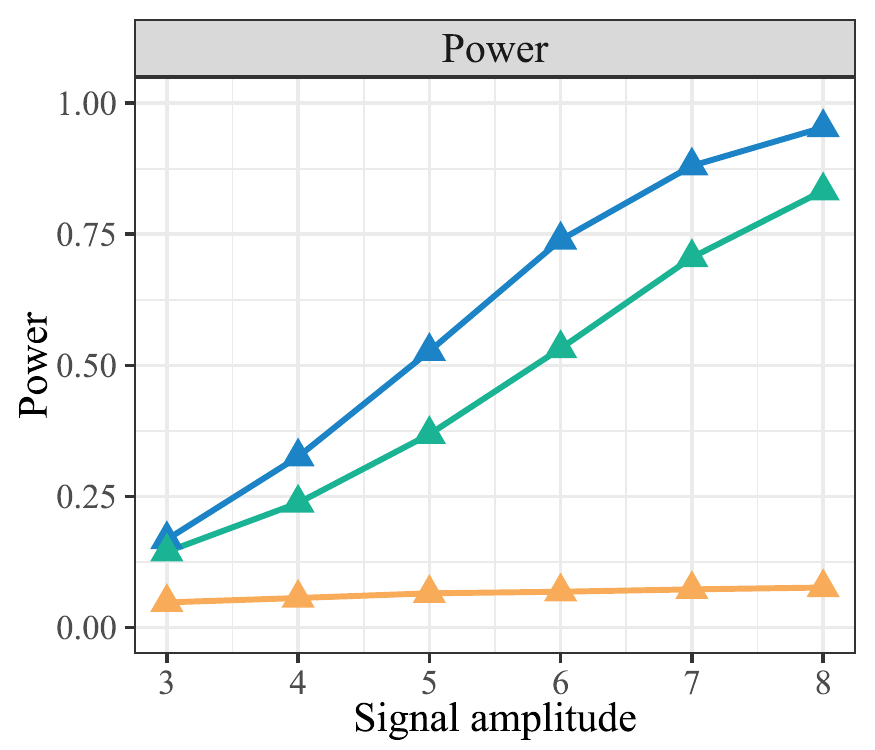}
  \end{minipage}
  \begin{minipage}{0.45\textwidth}
    \centering
  	\includegraphics[width =.8 \textwidth]{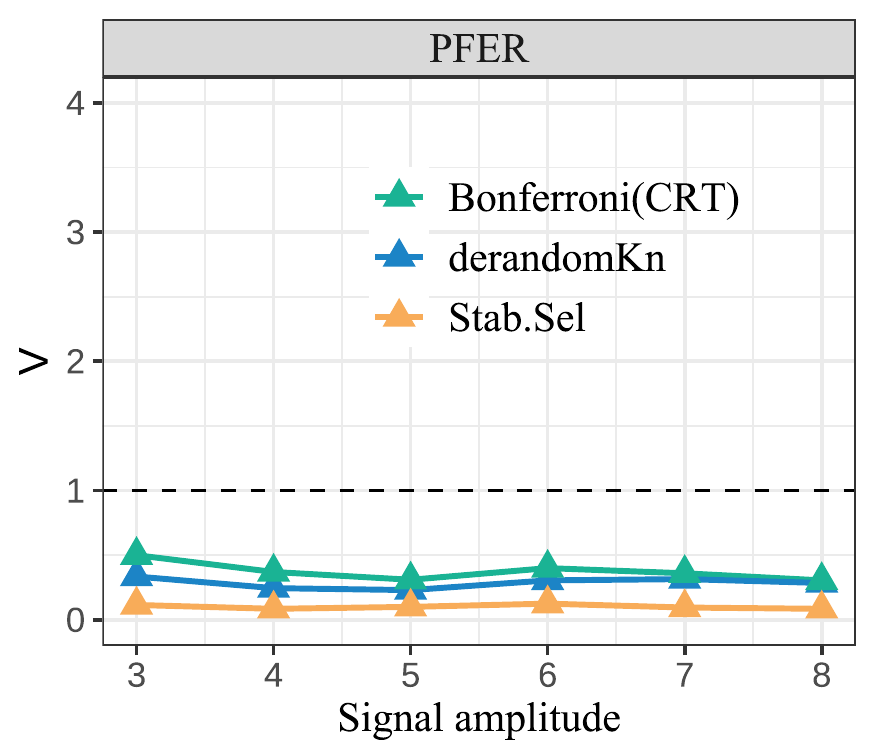}
  \end{minipage}
  \caption{Results from \algoname, stability selection and the
    Bonferroni correction with CRT p-values. The simulation setting is
    the same as that in Figure \ref{fig:pfer_amp_gaussian_within}.}
	\label{fig:pfer_amp_gaussian}
\end{figure}

\paragraph{Large-scale study}
In this case, we see from Figure \ref{fig:pfer_amp_gaussian_large}
that the three procedures control the PFER (up to
fluctuations). Derandomized knockoffs has a higher power than the
other methods.  The Bonferroni correction suffers a power loss because
of the low quality of the p-values. Once again, stability selection
loses power due to sample splitting and a conservative choice of
threshold.

\begin{figure}
	\centering
  \begin{minipage}{.45\textwidth}
    \centering
	\includegraphics[width = .8\textwidth]{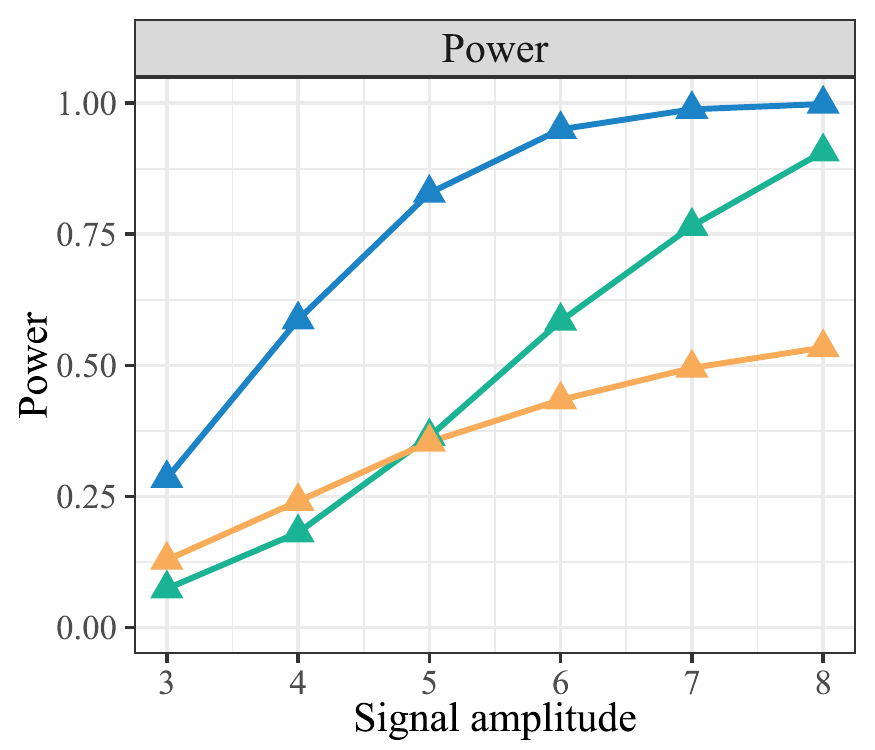}
  \end{minipage}
  \begin{minipage}{.45\textwidth}
  \centering
	\includegraphics[width = .8\textwidth]{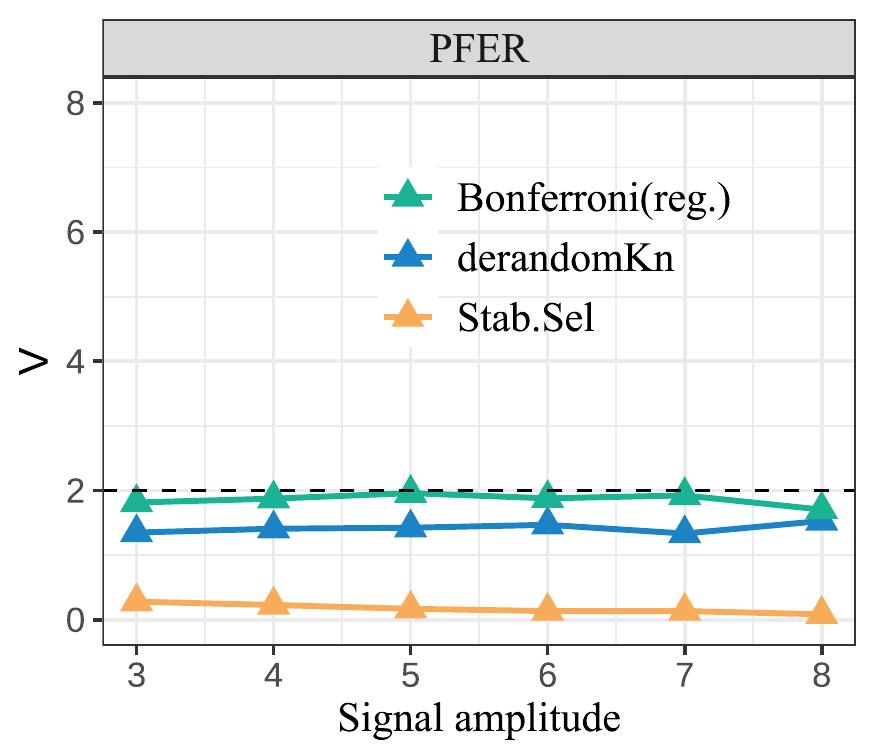}
  \end{minipage}
  \caption{Results from \algoname, stability selection and the Bonferroni correction with regression p-values.  The simulated setting is the same as that in Figure \ref{fig:pfer_amp_gaussian_large_within}.}
	\label{fig:pfer_amp_gaussian_large}
\end{figure}

\subsection{$k$-FWER control}
We now turn our attention to evaluating $k$-FWER control, following
the set of simulation settings from Section \ref{sec:fwer}.  In the
small-scale experiment, we compare \algoname~with stability selection
(the implementation is the same as in Section
\ref{sec:numerical_pfer}) and the Bonferroni's method where the
p-values are computed via CRT.  In the large-scale experiment, it is
computationally very expensive to compute the CRT p-values as observed
earlier. How about other p-values? The truth is that p-value based
methods face a serious problem since, to the best of our knowledge, it
is totally unclear how to compute valid p-values for conditional
hypothesis testing; that is, for determining whether $\beta_j = 0$.
The reader may consider that in our high-dimensional logistic
regression setting, the maximum likelihood estimator does not even
exist.\footnote{Even if we were to reduce the dimensionality,
  classical high-dimensional likelihood theory is plain wrong
  \citep{candes2020phase} so that the p-value based methods are
  extremely problematic.}  This is the reason why in the large-scale
experiment, we forego the comparison with the Bonferroni procedure and
only compare with the stability selection procedure (with the same
parameters used in the small-scale experiment).

Figure \ref{fig:fwer_amp_gaussian} concerns the small-scale study and
show the realized power and FWER.  Figure
\ref{fig:fwer_amp_gaussian_large} displays the same statistics in the
large-scale study.
As before, we see that both procedures control the FWER. We can also
see that \algoname~yields a much higher power.

\begin{figure}[ht]
	\centering
  \begin{minipage}{.45\textwidth}
    \centering
	  \includegraphics[width = 0.8\textwidth]{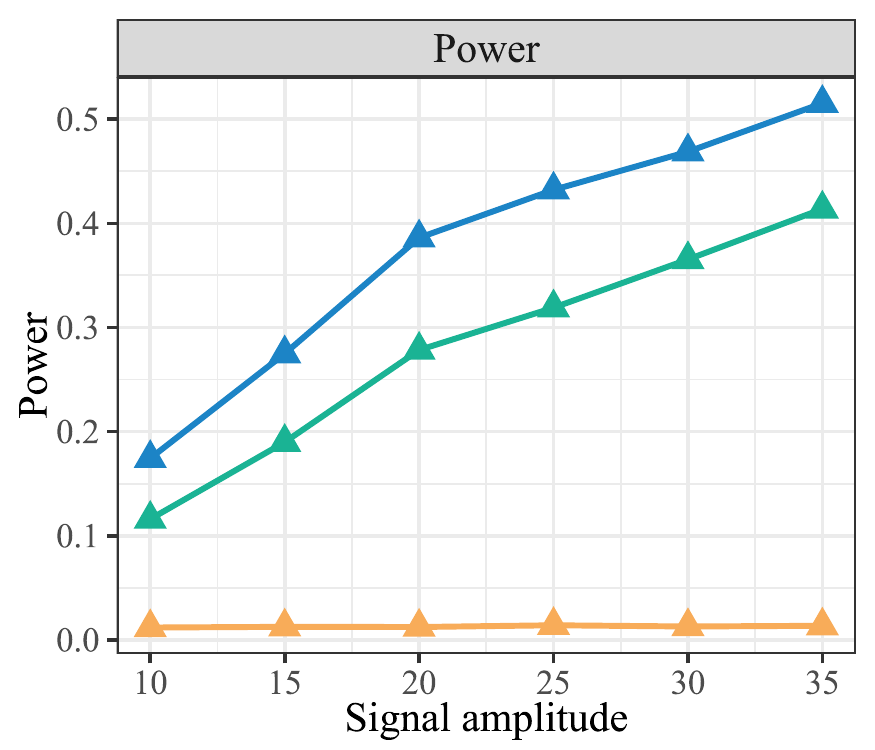}
  \end{minipage}
  \begin{minipage}{.45\textwidth}
    \centering
  	\includegraphics[width = 0.8\textwidth]{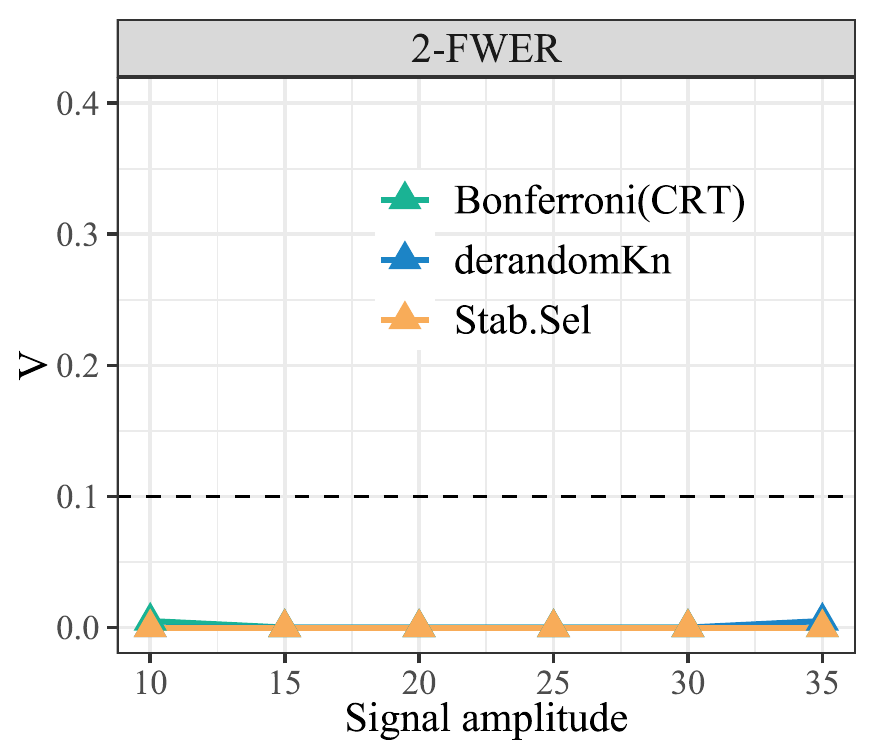}
  \end{minipage}
  \caption{Results from \algoname~and stability selection.  
    The simulation setting is the same as that in Figure
    \ref{fig:fwer_amp_gaussian_within}.}
	\label{fig:fwer_amp_gaussian}
\end{figure}

\begin{figure}[ht]
	\centering
  \begin{minipage}{.45\textwidth}
    \centering
	  \includegraphics[width = 0.8\textwidth]{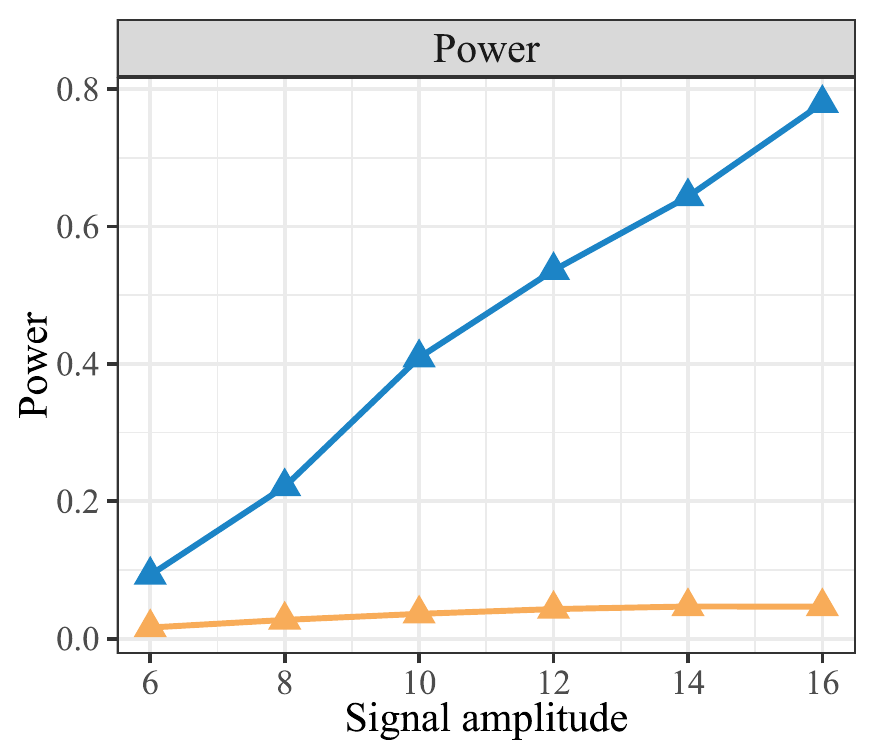}
  \end{minipage}
  \begin{minipage}{.45\textwidth}
    \centering
	  \includegraphics[width = 0.8\textwidth]{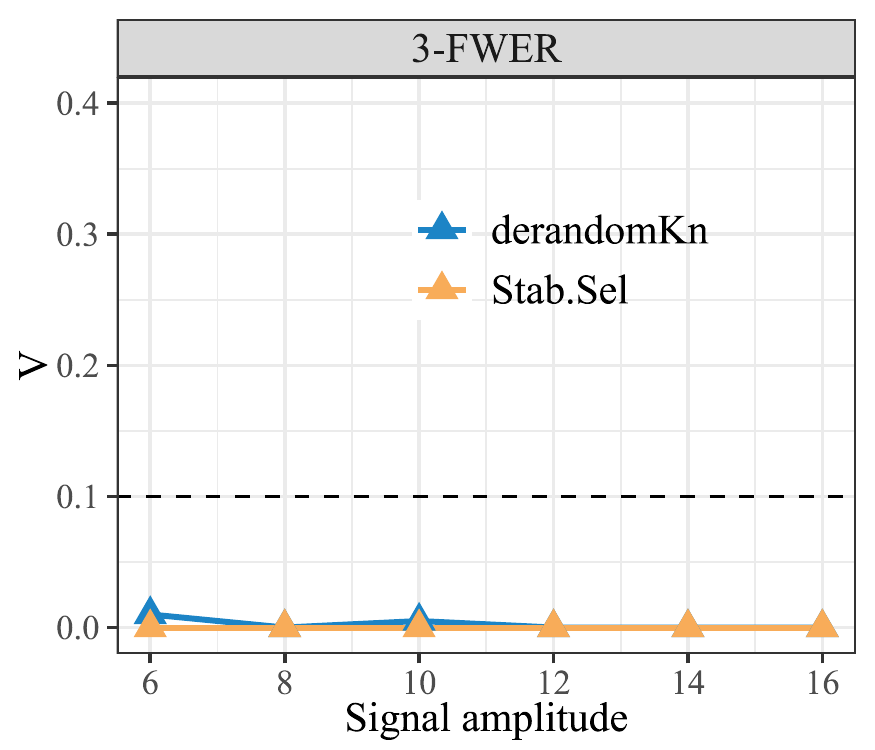}
  \end{minipage}
  \caption{Results from \algoname~and stability selection. 
  The simulation setting is the same as that in Figure
    \ref{fig:fwer_amp_gaussian_large_within}.}
	\label{fig:fwer_amp_gaussian_large}
\end{figure}

\section{Application to multi-stage GWAS}
\label{sec:application}
\subsection{Background}

The main goal of GWAS 
is to detect single-nucleotide polymorphisms (SNPs) associated with
certain phenotypes. The task is commonly carried out in multiple
stages, see e.g.,
\cite{kote2011seven,thomas2009multistage,lambert2013meta}. The purpose
of the early stages is often exploratory so that researchers tend to
consider more liberal type-I error criteria (such as FDR) to allow for
the inclusion of more candidates.  The end-stage study is, in
contrast, confirmatory, thus asking for a more stringent type-I error
criterion (such as FWER). Informally, we can say that the early stage
study narrows down the choices to a subset of ``candidate SNPs'',
whereas the end-stage study pins down the final discoveries.  Figure
\ref{fig:multistage_gwas} provides a pictorial description of a
typical multi-stage GWAS workflow. Here, we would like to apply
derandomized knockoffs to the end stage, paving the way to reliable
and stable decision making.

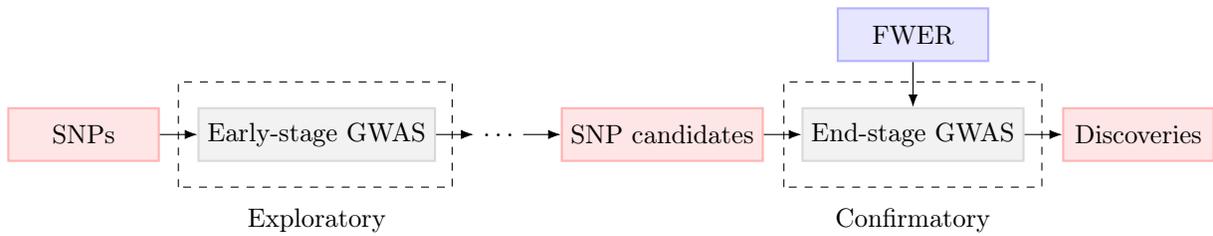
\begin{figure}[ht]
	\centering
  \begin{tikzpicture}[scale = 0.7,
    bluenode/.style = {rectangle,draw=blue!30, fill=blue!10, thick,minimum width=20mm,minimum height=7mm},
    rednode/.style = {rectangle,draw=red!30, fill=red!10, thick,minimum width=20mm,minimum height=7mm},
    graynode/.style = {rectangle,draw=gray!30, fill=gray!10, thick,minimum width=20mm,minimum height=7mm},
    ]
    \node[rednode] (snp) {SNPs};
    \draw[dashed](1.8,-1) rectangle (7,1); 
    \node[graynode] (stage1) [right=5mm of snp] {Early-stage GWAS};
    \node [below=5mm of stage1] {Exploratory};
    \node (mid) [right=5mm of stage1] {$\ldots$};
    \node[rednode] (C) [right=5mm of mid] {SNP candidates};
    \node[graynode] (stage2) [right=5mm of C] {End-stage GWAS};
    \draw[dashed](13.3,-1) rectangle (18.2,1); 
    \node[rednode] (dis) [right=5mm of stage2] {Discoveries};
    \node [below=5mm of stage2] {Confirmatory};
    \node[bluenode] (fwer) [above=6mm of stage2] {FWER};
    \draw[-Latex] (snp) -- (stage1);
    \draw[-Latex] (stage1) -- (mid);
    \draw[-Latex] (mid) -- (C);
    \draw[-Latex] (C) -- (stage2);
    \draw[-Latex] (stage2) -- (dis);
    \draw[-Latex] (fwer) -- (stage2);
	\end{tikzpicture}
	\caption{A typical workflow of multi-stage GWAS.}
	\label{fig:multistage_gwas}
\end{figure}

\subsection{Answering the \emph{same} question in stages}

We pause to discuss challenges associated with multi-stage
studies. Although our methods apply regardless of the relationship
between a phenotype $Y$ and genetic variants $X_1, \ldots, X_p$, and
always yield type-I error control, it may simplify the discussion to
consider a standard linear model relating the quantitative $Y$ to $X$
to bring the reader onto familiar grounds (recall this is purely
hypothetical).  Consider a geneticist who has genotyped a number of
sites. She wants to know whether the coefficient $\beta_j$ associated
with the variant $X_j$ vanishes or not.  Suppose now that in a first
stage---e.g.~after analyzing the results of a first study---she thins
out the list of possibly interesting variants, those for which she
suspects $\beta_j$ may not be zero. In a later confirmatory study, we
want her to determine whether the coefficients of the screened
variables in a model that still includes all the variants
$X_1, \ldots, X_p$ she was originally interested vanish or not. It
might be tempting to test in the second stage whether coefficients
vanish in the {\em reduced model} only including those variables that
passed screening. However, note that this would lead to test
hypotheses that are different from those we started with, not merely a
subset of them.  To bring this point home, imagine that only one
variable passed screening. Then in the second stage, this strategy
would lead to test a {\em marginal} test of hypothesis, which is not
what our geneticist wants (she wants a {\em conditional} test). This
change of hypotheses so strongly influenced by the random results of
the selection of the first stage seems hardly coherent with the goal
of the scientific study. (For more discussion about full versus
reduced model inference, we refer the reader to
\citet{wu2010screen,wasserman2009high,barber2019knockoff} and
\citet{fan2008sure,voorman2014inference,belloni2014inference,ma2017semi}.)

In light of this, this paper proposes a pipeline for multi-stage GWAS
that answers the \emph{same} question throughout the stages. More
specifically, from the very first stage, we are committed to testing
the conditional independence hypothesis:
\begin{align*}
\calH_j:~Y~\indep~X_j\mid X_{-j},
\end{align*}
where $X_{-j}$ corresponds to \emph{all} the SNPs except $X_j$. This
means that if $\calC$ is the candidate set selected by previous
stages, we test $\calH_j$ for each $j\in \calC$ in the end stage.  We
do this by applying \algoname, which controls type-I errors regardless
of the procedures used in the previous stages. This is very different
from existing approaches which switch the inferential target and would
test whether a variable $j \in \calC$ is significant in a model that
only includes variables in $\calC$
\citep{lee2013exact,tibshirani2016exact,tian2018selective,fithian2014optimal,barber2019knockoff}.

\subsection{A synthetic example with real genetic covariates}
\label{sec:syntheticex}
We now rehearse the pipeline for multi-stage GWAS and showcase the
performance of \algoname~using a synthetic example with \emph{real
  genetic covariates}.  Recall that the key assumption underlying the
knockoffs procedure is the knowledge of the distribution of $X$.  For
real data analysis, the distribution of $X$ can only be approximated,
which raises the concern of whether such an approximation invalidates
FWER control.  With this in mind, we set up a simulation with real
covariates and synthetic phenotypes generated from a known
model.  The knockoffs are constructed with the approximated
distribution of $X$ and we verify whether or not FWER is controlled at
the desired level since we know the true conditional model. (Recall
that FWER constrol does not in any away depend upon the unknown
relationship between $Y$ and $X$.)

\newcommand{\Var}{\operatorname{Var}}

\paragraph{Data acquisition}
We obtain the covariate matrices by subsetting the genoypte matrix
from the UK Biobank data set \citep{bycroft2018uk}, where we only
include the $9{,}537$ SNPs from chromosome $22$.  A subset of
$n=6,000$ samples is drawn for the early-stage GWAS and $20$ other
disjoint subsamples~(each of size $n=3{,}000$)~for the end-stage GWAS
(see Figure \ref{fig:realx} for illustration), where all samples are
unrelated British individuals.  The FWER and power reported are
averaged over these $20$ independent subsets.  Conditional on $X$, the
variable $Y$ is generated from a linear model:
\begin{align}
	Y \mid X_1,\ldots,X_p \sim \N(\beta_1X_1+\ldots+\beta_pX_p,1).
\end{align}
The $40$ non-null features are equally divided into $20$ clusters,
where the clusters are evenly spaced on the chromosome.  The
magnitudes and signs of the coefficients vary across clusters, and are
the same within each cluster.  The relative absolute values of the
magnitudes are chosen uniformly at random so that the ratio between
the smallest and the largest is $1/19$; the signs of the coefficients
are determined by independent coin flips.  The heritability, defined
as $h^2 :=\Var(X^\top \beta)/\Var(Y)$, is used as a control parameter;
informally, this is the fraction of the variance of the phenotype
explained by genetic factors.


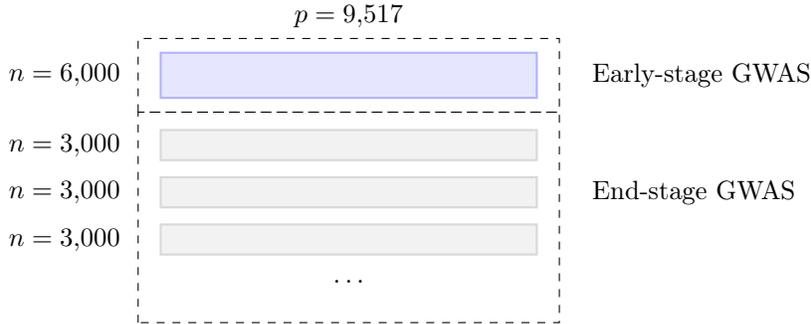
\begin{figure}[ht]
	\centering
  \begin{tikzpicture}[scale = 0.7,
    bluenode/.style = {rectangle,draw=blue!30, fill=blue!10, thick,minimum width=50mm,minimum height=6mm},
    graynode/.style = {rectangle,draw=gray!30, fill=gray!10, thick,minimum width=50mm,minimum height=4mm},
    ]
    \node[bluenode] (pre) {};
    \node [above=2mm of pre] {$p=9{,}517$};
    \node [left=4mm of pre] {$n=6{,}000$};
    \node [right=6mm of pre] {Early-stage GWAS};
    \draw[dashed](-4,-0.7) rectangle (4,0.7);
    \draw[dashed](-4,-4.7) rectangle (4,-0.7);
    \node[graynode] (end1) [below=4mm of pre] {};
    \node [left=4mm of end1] {$n=3{,}000$};
    \node[graynode] (end2) [below=2mm of end1] {};
    \node [right=6mm of end2] {End-stage GWAS};
    \node [left=4mm of end2] {$n=3{,}000$};
    \node[graynode] (end3) [below=2mm of end2] {};
    \node [left=4mm of end3] {$n=3{,}000$};
    \node [below=2mm of end3] {$\ldots$};
	\end{tikzpicture}
\caption{Splitting the UK Biobank data set and obtaining covariate matrices.}
\label{fig:realx}
\end{figure}
\paragraph{Pre-processing and clustering}
We follow the pre-processing steps in \citet{sesia2020multi}, and only
consider the biallelic SNPs with minor allele frequency above $0.1\%$
and in Hardy-Weinberg equilibrium ($10^{−6}$).  A typical challenge in
GWAS is the presence of linkage disequilibrium, which makes it
essentially difficult to pin-point the important SNPs.  To address
this difficulty, we use the treatment suggested in
\citet*{sesia2020multi} and partition the genetic variants into
groups---we shall use the words ``groups'' and ``clusters''
interchangeably---via adjacency-constrained hierarchical clustering.
After clustering, we treat the clusters as the inferential object and
test whether or not a cluster of SNPs is independent of the phenotype
conditional on all other clusters of SNPs.  Formally, let $G$ be a
partition of $[p]$. For each cluster $g\in G$, the hypothesis to be
tested is
\begin{align}
  \calH_{g}:~Y~\indep~X_g\mid X_{-g}.\label{eq:group_cond_ind}
\end{align}
For example, suppose as in Figure \ref{fig:cluster} that we have $12$ SNPs
divided into four clusters. With $g = \{5,6\}$, hypothesis 
\eqref{eq:group_cond_ind} translates into
$(X_5,X_6)~\indep~Y\mid X_{-\{5,6\}}$.  
The sizes of the clusters determine how fine our discoveries are.  In
the experiment we consider five levels of resolution: 
$2\%$, $10\%$, $20\%$, $50\%$ and $100\%$, where the resolution is
defined as the number of groups divided by the number of SNPs.
\begin{figure}[ht]
	\centering
  \begin{tikzpicture}[scale = 0.7,
    bluenode/.style = {circle,draw=blue!20, fill=blue!0, very thick, minimum size = 0.5mm},
    bluenodef/.style = {circle,draw=blue!20, fill=blue!10, very thick, minimum size = 0.5mm},
    ]
    \node(0)  {};
    \node[bluenode] (1) [right=5mm of 0] {};
    \node[below = 0.5mm of 1] {1};
    \node[bluenode] (2) [right=5mm of 1] {};
    \node[below = 0.5mm of 2] {2};
    \node[bluenode] (3) [right=5mm of 2] {};
    \node[below = 0.5mm of 3] {3};
    \node[bluenode] (4) [right=5mm of 3] {};
    \node[below = 0.5mm of 4] {4};
    \node[bluenodef] (5) [right=5mm of 4] {};
    \node[below = 0.5mm of 5] {5};
    \node[bluenodef] (6) [right=5mm of 5] {};
    \node[below = 0.5mm of 6] {6};
    \node[bluenode] (7) [right=5mm of 6] {};
    \node[below = 0.5mm of 7] {7};
    \node[bluenode] (8) [right=5mm of 7] {};
    \node[below = 0.5mm of 8] {8};
    \node[bluenode] (9) [right=5mm of 8] {};
    \node[below = 0.5mm of 9] {9};
    \node[bluenode] (10) [right=5mm of 9] {};
    \node[below = 0.5mm of 10] {10};
    \node[bluenode] (11) [right=5mm of 10] {};
    \node[below = 0.5mm of 11] {11};
    \node[bluenode] (12) [right=5mm of 11] {};
    \node[below = 0.5mm of 12] {12};
    \node(13) [right=5mm of 12] {};
    \draw (0) -- (1);
    \draw (1) -- (2);
    \draw (2) -- (3);
    \draw (3) -- (4);
    \draw (4) -- (5);
    \draw (5) -- (6);
    \draw (6) -- (7);
    \draw (7) -- (8);
    \draw (8) -- (9);
    \draw (9) -- (10);
    \draw (10) -- (11);
    \draw (11) -- (12);
    \draw (12) -- (13);
    \draw [decorate,decoration={brace},xshift=0pt,yshift=-4pt]
      (0.6,0.5) -- (5.3,0.5) node [black,midway,xshift=0cm,yshift = .3cm]
      {\footnotesize $g_1$};
    \draw [decorate,decoration={brace},xshift=0pt,yshift=-4pt]
      (5.7,0.5) -- (7.8,0.5) node [black,midway,xshift=0cm,yshift = .3cm]
      {\footnotesize $g_2$};
    \draw [decorate,decoration={brace},xshift=0pt,yshift=-4pt]
      (8.2,0.5) -- (11.5,0.5) node [black,midway,xshift=0cm,yshift = .3cm]
      {\footnotesize $g_3$};
    \draw [decorate,decoration={brace},xshift=0pt,yshift=-4pt]
      (11.8,0.5) -- (15.3,0.5) node [black,midway,xshift=0cm,yshift = .3cm]
      {\footnotesize $g_4$};
  \end{tikzpicture}
	\caption{Visual illustration of the partition of SNPs.}
	\label{fig:cluster}
\end{figure}
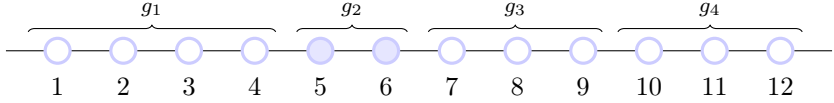

\paragraph{Early-stage GWAS}
We apply two procedures to identify candidate hypotheses:
\begin{enumerate}
\item The first is the group model-X knockoffs procedure
  \citep{sesia2020multi} with FDR controlled at level $0.5$. The
  selected clusters then become the candidate SNP clusters for the end
  stage.
\item The second is motivated by the practice of early-stage studies,
  which often only provide summary statistics such as p-values and,
  lead to the selection of candidate SNPs via p-value cutoffs. To evaluate the
  performance of \algoname~in such a context, we screen candidate SNPs
  as follows: compute a marginal regression p-value for each SNP and
  select those SNPs with a p-value below
  $5\times 10^{-6}~(\approx 0.05/p)$.  Then a SNP cluster is
  considered a ``candidate'' if at least one of its SNPs passes the
  p-value threshold.
\end{enumerate}
Figure \ref{fig:multistage_gwas_kn} summarizes the basic steps of our
hypothetical multi-stage GWAS.
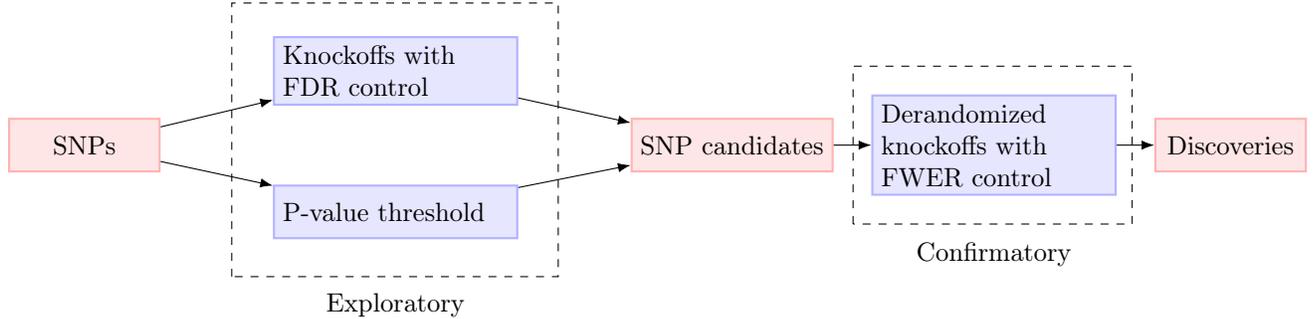
\begin{figure}[ht]
	\centering
  \begin{tikzpicture}[scale = 0.7,
    bluenode/.style = {rectangle,draw=blue!30, fill=blue!10, thick,minimum width=20mm,minimum height=7mm},
    rednode/.style = {rectangle,draw=red!30, fill=red!10, thick,minimum width=20mm,minimum height=7mm},
    graynode/.style = {rectangle,draw=gray!30, fill=gray!10, thick,minimum width=20mm,minimum height=7mm},
    ]
    \node[rednode] (snp) {SNPs};
    \draw[dashed](2.8,-2.5) rectangle (9,2.7);
    \node (auxnode) [right=5mm of snp,text width=5cm] {};
    \node[bluenode] (stage10) [above=4mm of auxnode,text width=3cm] {Knockoffs with FDR control};
    \node[bluenode] (stage11) [below=4mm of auxnode,text width =3cm] {P-value threshold};
    \node [below=6mm of stage11] {Exploratory};
    \node[rednode] (C) [right=5mm of auxnode] {SNP candidates};
    \node[bluenode] (stage2) [right=5mm of C,text width=3cm] {Derandomized knockoffs with FWER control};
    \draw[dashed](14.6,-1.5) rectangle (19.9,1.5); 
    \node[rednode] (dis) [right=5mm of stage2] {Discoveries};
    \node [below=5mm of stage2] {Confirmatory};
    \draw[-Latex] (snp) -- (stage10);
    \draw[-Latex] (snp) -- (stage11);
    \draw[-Latex] (stage10) -- (C);
    \draw[-Latex] (stage11) -- (C);
    \draw[-Latex] (C) -- (stage2);
    \draw[-Latex] (stage2) -- (dis);
	\end{tikzpicture}
	\caption{Schematic representation of the multi-stage GWAS used
          in the numerical study.}
	\label{fig:multistage_gwas_kn}
\end{figure}

\paragraph{Knockoffs and feature importance statistics construction}
Let $\calG$ be the set of candidate SNP \emph{clusters} identified
during the earlier stage. We proceed to construct a \emph{conditional
  knockoff} copy \emph{only} for $X_{\calG}$: the knockoff copy
$\tX_{\calG}$ satisfies the knockoff properties conditional on
$X_{-\calG}$; that is,
$\tX_{\calG}~\indep~ Y\mid (X_{\calG},X_{-\calG})$ and
\begin{align}
  \label{eq:cond_ex}
  (X_{\calG},\tX_{\calG})_{\text{swap}(g)}|X_{-\calG} \stackrel{\mathrm{d}}{=} (X_{\calG},\tX_{\calG})|X_{-\calG},
\end{align}
where $g$ denotes a SNP cluster that belongs to $\calG$. Above, $(X_{\calG},\tX_{\calG})_{\text{swap}(g)}$ is 
obtained from $(X_{\calG},\tX_{\calG})$ by swapping the features $X_j$ and $\tX_j$ for each SNP $j\in g$.
We model the distribution of $X$ via a hidden Markov model (HMM), and
construct (group) HMM knockoffs based on a variant of an existing
procedure \citep{sesia2019gene,sesia2020multi}.  The HMM used here
describes the distribution of SNPs in unrelated British individuals
\citep{sesia2020multi}.\footnote{See \cite{sesia2020controlling} for
  more sophisticated models capable of handling population structure.}
Details about the construction of $\tX_\calG$ are deferred to Appendix
\ref{Sec:construct_kn}.  As for the feature importance statistic, we
use the (group) LCD statistics introduced in \citet{sesia2020multi}.


\paragraph{Results} 
Figure \ref{fig:pval_illustration} presents the marginal p-values for
a subset of SNPs in the form of a Manhattan plot, where it is easily
seen that the selected SNPs are clustered together.  Figure
\ref{fig:real_x_result_2} shows the realized power, end-stage power
(the number of true discoveries divided by the number of non-nulls
selected by pre-GWAS) and FWER with different resolutions in the case
where the screening step thresholds marginal p-values.  Figure
\ref{fig:real_x_result_1} presents the same statistics in the case
where the screening step uses knockoffs.  In both cases, it can
readily be observed that even with an approximate distribution of $X$,
our procedure successfully controls the FWER. Furthermore, while the
candidates selected by p-value thresholding are clustered together
(which potentially affects power), the two-stage procedure still makes
a reasonable number of findings. 

\begin{figure}
  \centering
  \includegraphics[width=\textwidth]{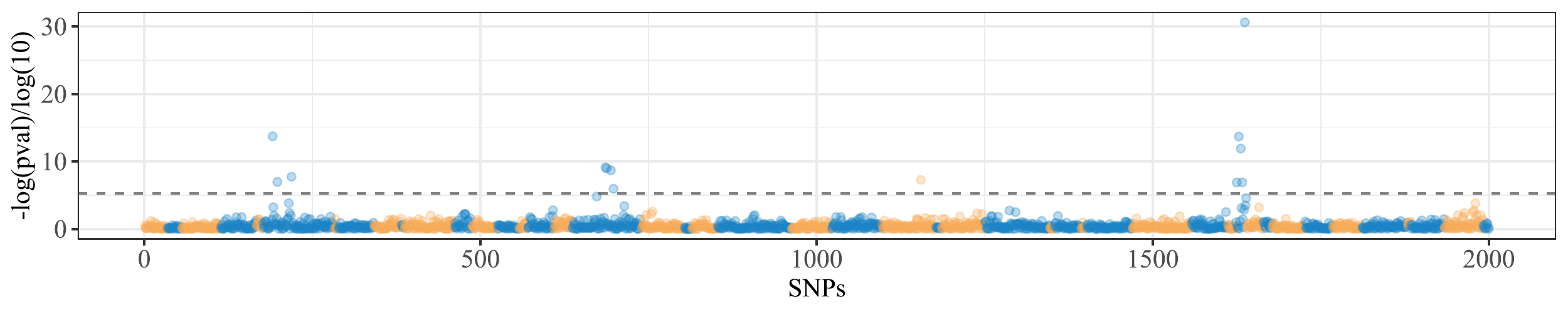}
  \caption{The Manhattan plot of a subset of $2000$ SNPs in the
    synthetic example. The colors correspond to the clusters at 2\%
    resolution.}
  \label{fig:pval_illustration}
\end{figure}

\begin{figure}[ht]
\centering
\begin{minipage}{1\textwidth}
\includegraphics[width = \textwidth]{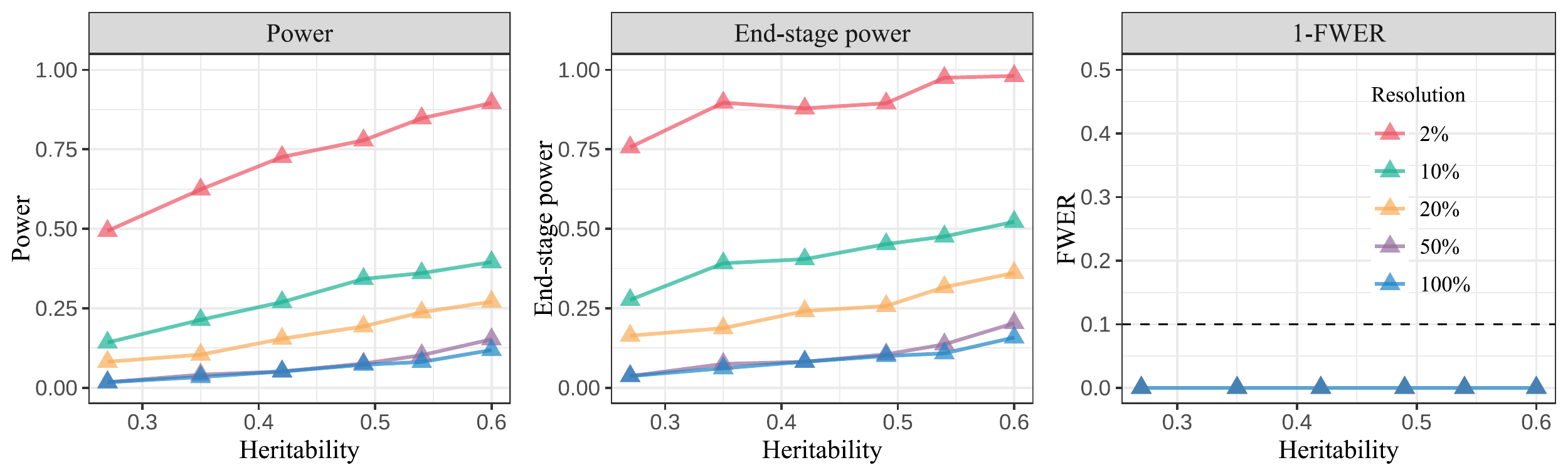}
\caption{Realized power (left), end-stage power (middle) and FWER (right) in the synthetic multi-stage GWAS. The early-stage GWAS selects SNPs with p-values less than $5\times 10^{-6}$. The target FWER level of \algoname~is $0.1$.}
\label{fig:real_x_result_2}
\end{minipage}

\begin{minipage}{1\textwidth}
\includegraphics[width = \textwidth]{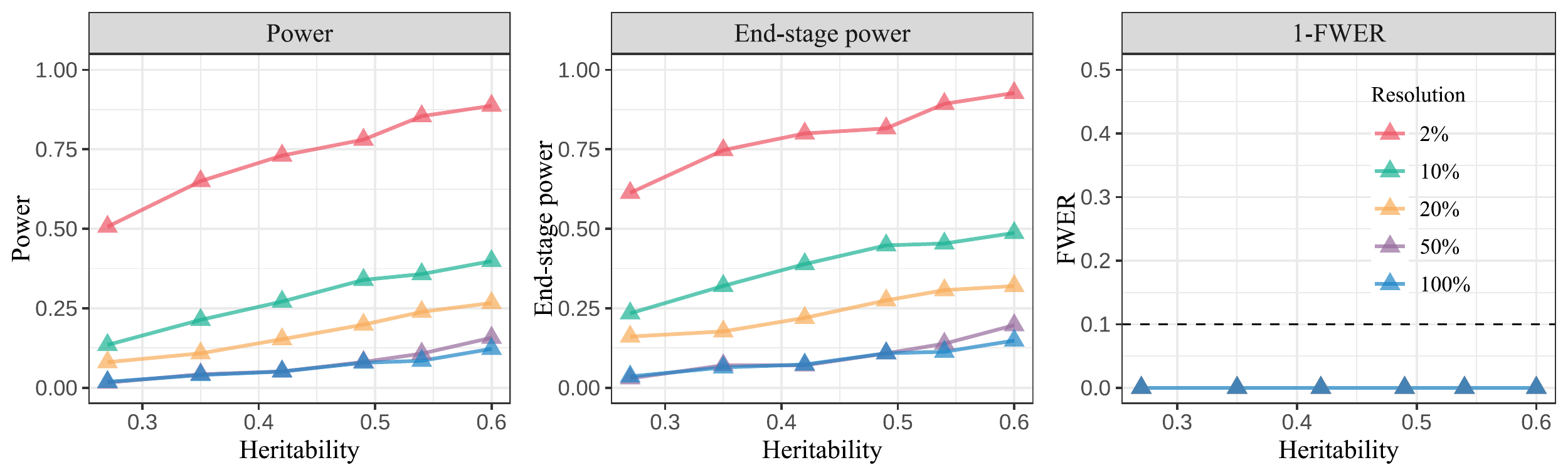}
\caption{Realized power (left), end-stage power (middle) and FWER (right) in the synthetic multi-stage GWAS. The early-stage GWAS uses model-X knockoffs with target FDR level $0.5$. The target FWER level of \algoname~is $0.1$.}
\label{fig:real_x_result_1}
\end{minipage}
\end{figure}

\subsection{End-stage GWAS of prostate cancer}
We finally apply our procedure to an end-stage GWAS of prostate
cancer.  We take the meta-analysis conducted by
\cite{schumacher2018association} as the early-stage study, and apply
\algoname~on a data set from UK Biobank for a confirmatory analysis.
 The UK biobank data set contains genetic information on $161$K
unrelated British male individuals and their disease status, 
i.e., whether or not a participant has reported being diagnosed with
prostate cancer.

After selecting p-values from
\citet{schumacher2018association} below $10^{-3}$, we end up with
$4072$ pre-selected SNPs. (The set of SNPs recorded in
\citet{schumacher2018association} can be different from that in the UK
Biobank data set. Here, we only consider the intersection of the two
sets.)

As in our earlier numerical study, the next step is to partition
\emph{a priori all} the SNPs into clusters at a level of resolution
$2\%$. The resulting average length of the clusters is $0.226$ Mb.  A
cluster is called a candidate cluster if at least one of its SNPs is a
candidate SNP.  Ten runs of conditional group HMM knockoffs are
constructed for the candidate clusters. We compute the group LCD statistics
as in the synthetic example from Section
\ref{sec:syntheticex}. Six additional covariates, namely, age and the top five principal
components of the genotypes are included in the knockoffs
predictive model as follows: instead of using the phenotypes as the
response, we use the residuals of the
phenotypes {\em after} regressing out these six additional covariates.
The inclusion of these covariates allows us to account for the
(remaining) population structure in the data, which increases the
detection power.  Finally, we apply \algoname~with
target FWER level $0.1$.  
Table \ref{tab:dis_res2} provides detailed information on the final list of clusters when the
resolution is $2\%$.

We compare our findings with those from the existing literature. Since
our discoveries are SNP clusters and different studies may contain
different sets of SNPs, we cannot directly compare the results across
different studies.  Here we consider findings to be confirmed by
another study if the latter reports a SNP whose position is within the
genomic locus spanned by a cluster we discovered.  With this, it turns
out that \emph{all} of our $8$ findings are confirmed by other
studies. Further, $7$ matches are exact in the sense that the leading
SNP of a discovered cluster is reported significant in the
literature. Specifically, clusters represented by rs12621278,
rs1512268, 
rs6983267, rs7121039, rs10896449
and rs1859962 are replicated by \citet{wang2015large}, which is a
large GWAS conducted in the Asian population; rs1016343 is confirmed
by \citet{hui2014study}---a study specifically investigating the
associations of six SNPs including rs1016343 in a Chinese population;
the association between rs7501939 and prostate cancer
is in \citet{elliott2010evaluation}, which is a study focusing on the
association of two SNPs including rs7501939 with several diseases. 

What would happen if we were a little more liberal?  To find out, we
also run the \algoname set to control the $3$-FWER at level $0.1$: all
the SNPs discovered earlier appear in the new discovery set. The more
liberal procedure makes seven {\em additional} discoveries and the
corresponding SNPs are listed in Table \ref{tab:dis_res2_k3}.


\begin{table}[ht]
  \centering
  \begin{tabular}{ccccc}
    \textbf{Lead SNP} & \textbf{Chromosome} & \textbf{Position range (Mb)} & \textbf{Size} & \textbf{Confirmed by?}\\
    \hline
    rs12621278 & 2 & 173.28-173.58 & 68 & \citet{wang2015large}\\
    rs1512268 & 8 & 23.39-23.55 & 48 & \citet{wang2015large}\\
    rs1016343 & 8 & 128.07-128.24 & 45 & \citet{hui2014study}\\
    rs6983267 & 8 & 128.40-128.47 & 37 &\citet{wang2015large}\\
    rs7121039 & 11 & 2.18-2.31 & 40 & \citet{wang2015large}*\\
    \hline
    rs10896449 & 11 & 68.80-69.02 & 62 &\citet{wang2015large}\\
    rs7501939 & 17 & 36.05-36.18 & 55 & \citet{elliott2010evaluation}\\
    rs1859962 & 17 & 69.07-69.24 & 40 & \citet{wang2015large}\\
    \end{tabular}
    \caption{Discoveries made by \algoname~at 2\% resolution and
      the target FWER level set to $0.1$. The asterisk indicates the
      following: the
      lead SNP is not found in the literature; however,  the
      literature reports a SNP within the position range of our reported cluster.}
  \label{tab:dis_res2}
\end{table}

\begin{table}[ht]
  \centering
  \begin{tabular}{ccccc}
  \textbf{Lead SNP} & \textbf{Chromosome} & \textbf{Position range (Mb)} & \textbf{Size} & \textbf{Confirmed by?}\\
  \hline
  rs905938 & 1 & 154.80-155.04 & 66 & \citet{wang2015large}*\\
  rs6545977 & 2 & 62.81-63.60 & 65 & \citet{wang2015large}*\\
  rs77559646 & 2 & 242.07-242.16 & 35 & \citet{kaikkonen2018ano7}\\
  rs2510769 & 4 & 95.28-95.60 & 59 & \citet{wang2015large}*\\
  rs10486567 & 7 & 27.56-28.04 & 87  &\citet{haiman2011characterizing}\\
  \hline
  rs17762878 & 8 & 127.85-128.07 & 70 & \citet{wang2015large}* \\
  rs4242382 & 8 & 128.47-128.56 & 34&\citet{zhao20148q24} \\
  \end{tabular}
  \caption{{\em Additional} discoveries made by \algoname~at 2\% resolution and
    the target $3$-FWER level set to $0.1$ (all the SNPs reported in Table~\ref{tab:dis_res2} is also discovered by this procedure). The setting is otherwise the same 
            as in Table \ref{tab:dis_res2}.}
  \label{tab:dis_res2_k3}
\end{table}


\section{Discussion}
\label{sec:discussion}
We proposed a framework for derandomized knockoffs inspired by
stability selection. By exploiting multiple runs of the knockoffs
algorithm, our method offers a more stable solution for selecting
non-null variables. 
Leveraging a base procedure with controlled PFER, we show how to
achieve PFER and $\kFWER$ control, these being error metrics perhaps
more suitable than FDR for confirmatory stage studies
\citep{tukey1980we,goeman2011multiple} as well as more
resource-consuming applications such as end-stage GWAS
\citep{meijer2016multiple,sham2014statistical}, clinical trials
\citep{crouch2017controlling} and neuro-imaging
\citep{eklund2016cluster}.  Furthermore, we execute our methodology on
a GWAS example and and find that \emph{all} our findings are confirmed
by related studies.


\paragraph{Future work}
While the current paper empirically demonstrates enhanced statistical
power, it would be of interest to theoretically validate power gains,
at least in some simple settings.  Also, the machinery described here
is general and can be applied to a variety of base procedures---even
procedures that do not come with controlled PFER. One would,
therefore, ask whether our theoretical framework can be adapted to
accommodate such base procedures.  Finally, perhaps the most natural
question is whether our ideas can be adapted to the more liberal FDR
criterion or related error rates such as the false discovery
exceedence.



\subsection*{Acknowledgment}
The authors would like to thank Malgorzata Bogdan, Lihua Lei,
and Chiara Sabatti for helpful discussions.  The authors also would like
to thank the PRACTICAL consortium, CRUK, BPC3, CAPS, PEGASUS for
providing GWAS summary statistics (more information can be found at
\url{http://practical.icr.ac.uk/blog/?page_id=8164}). Z.~R.~is
supported by the Math + X award from the Simons Foundation, the JHU
project 2003514594, the ARO project W911NF-17-1-0304 , the NSF grant
DMS 1712800 and the Discovery Innovation Fund for Biomedical Data
Sciences.  Y.~W.~is supported partially by the NSF grants CCF 2007911
and DMS 2015447. E.~C.~is partially supported by NSF via grants DMS
1712800 and DMS 1934578 and by the Office of Naval Research grant
N00014-20-12157.

\bibliographystyle{apalike}
\bibliography{ref}

\vspace{1cm}

\appendix
\section{More details on $v$-knockoffs}
\label{Sec:details}
\subsection{Algorithm description}
Suppose at the $m$-th round, the feature importance statistic $\bm{W}^m$ has been computed. The $v$-knockoffs algorithm starts by ordering the features 
according to the magnitudes of $W^m_j$, i.e.
\begin{align*}
	|W^m_{r_1}|\geq \ldots\geq|W^m_{r_j}|\geq\ldots\geq|W^m_{r_p}|,
	\qquad \text{for some permutation }~ r_1,\ldots,r_p.
\end{align*}
Given a  pre-specified integer $v$, the procedure examines the features one by one from $r_1$ to $r_p$ according to the ordering above until the first time there are $v$ negative $W^m_j$'s. Formally, the stopping criterion is defined as 
\begin{align*}
	T_v^m \defn \min \, \Big\{k\in[p] \,\mid\, \sum^k_{j=1}\indc{W^m_{r_j}<0}\geq v \Big\}.
\end{align*}
The selected set is then defined to be the collection of the features examined before $T_v^m$ which have positive signs of feature importance statistics, i.e.,
\begin{align}
\label{eqn:vknockoffs}
  \hat{\calS}^m \defn \{r_j \,\mid\, j<T_v^m,~ W^m_{r_j}>0\}.
\end{align}
\begin{figure}[h!]
	\centering
	\begin{tikzpicture}
	\filldraw[color=blue!60, fill=blue!5, very thick](0,0) circle (0.4);
	\node[align = center] at (0,0){$\bmx$};
	\draw[-Latex] (0.5,0) -- (1,0);
	\filldraw[color=blue!60, fill=blue!5, very thick](1.5,0) circle (0.4);
	\node[align = center] at (1.5,0){$\bmtx^m$};
	\draw[-Latex] (2,0) -- (2.5,0);
	\filldraw[color=blue!60, fill=blue!5, very thick](3,0) circle (0.4);
	\node[align = center] at (3,0){$\bm{W}^m$};
  \node[align = center] at (2.3,0.3){$\mathbf{Y}$};
	\draw[-Latex] (3.5,0) -- (6.5,0);
	\filldraw[color=blue!60, fill=blue!5, very thick](7,0) circle (0.4);
  \node[align = center] at (7,0){$\hat{\calS}^m$};
	\node[align = center] at (5,0.3){base procedure};
	\end{tikzpicture}
	\caption{Visual illustration of the base procedure.}
	
	\label{fig:framework_illustration}
\end{figure}
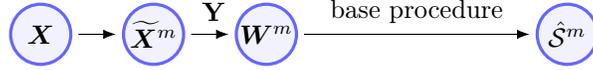

\subsection{Theoretical properties}

\begin{lemma}[\cite{janson2016familywise}]
\label{LemVknock}
For any integer $v$, the number of false discoveries produced by the \emph{$v$-knockoffs} procedure 
\begin{align}
  V_m \defn \#\{ j \in \Null \cap \hat{\calS}^m \},
\end{align}
is stochastically dominated by a negative binomial variable
$\mathrm{NB}(v,1/2)$.
\end{lemma}
This lemma is originally stated in \cite{janson2016familywise} for
fixed knockoffs in linear models. However, the result is a direct
consequence of the coin-flip property (see~\citet[Lemma 3.3]{candes2018panning}); here the coin-flip property means that conditional on $(|W^m_1|,\ldots , |W^m_p|)$, the signs of the null $W^m_j$’s ($j \in \mathcal{H}_{0}$) are i.i.d.~coin flips.
Therefore, Lemma~\ref{LemVknock} is also applicable in the model-X setting.

\section{Proofs of the main results}
\subsection{Proof of Proposition \ref{prop:asst_ratio}}
\label{PfThmFB}
\paragraph*{Proof of (a)}

For every $j\in \calH_0$, the selection probability $\Pi_j$ is supported on $\{0,1/M,2/M,\ldots,1\}$. Letting $p_m:=\p\left(\Pi_j = m/M\right)$ for $m=0,1,\ldots,M$, we can express the ratio as 
\begin{align*}
  \frac{\p(\Pi_j\ge \eta)}{\E[\Pi_j]} = \Big(\sum_{m\ge \eta M}p_m\Big)/\Big(\sum^M_{m=0} p_m\cdot \dfrac{m}{M}\Big).
\end{align*}
By definition, $p_m\ge 0$ and $\sum^M_{m=0}p_m=1$. In addition, by the
coin-flip property (see~\citet[Lemma 3.3]{candes2018panning}),
$\E[\Pi_j]\le 1/2$ since $j\in \calH_{0}$, which translates into
$\sum^M_{m=0}p_m\cdot(m/M) \le 1/2$. Together with the monotonicity
assumption, the optimization problem can be written as

\begin{equation}
  \label{eq:raw_opt}
  \begin{array}{ll}
    \text{maximize} & \quad \Big(\sum_{m\ge \eta M} p_m \Big) / \Big( \sum^M_{m=0} p_m\cdot \dfrac{m}{M}\Big) \\
    \text{subject to} & \quad p_m\ge 0,\\
                    & \quad p_{m-1}\ge p_m,~m\in[M],\\
                    & \quad \sum^M_{m=0} p_m =1,~\sum^M_{m=0} p_m\cdot \dfrac{m}{M}\le 1/2.
                    \end{array} 
\end{equation}
Setting $y_m := p_m/\Big( \sum^M_{m=0} p_m\cdot \dfrac{m}{M}\Big)$,
optimization problem~\eqref{eq:raw_opt} reduces to
\begin{equation*}
  \begin{array}{ll}
    \text{maximize} & \quad \sum_{m\ge \eta  M} y_m \\
    \text{subject to} & \quad y_m\ge 0,\\
                    & \quad y_{m-1}\ge y_m, ~m\in[M],\\
                    & \quad \sum^M_{m=0} y_m\cdot \dfrac{m}{M} =1,
                    \end{array} 
\end{equation*}
which proves (a). Here, we use the fact that
$$\sum^M_{m=0} p_m\cdot \dfrac{m}{M} \leq (\sum^M_{m=0} p_m) \cdot
(\frac{1}{M+1} \sum^M_{m=0} \frac{m}{M}) = \frac{1}{2}$$
holds for any non-increasing sequence $\{p_{m}\}$.

\paragraph*{Proof of (b), (c) and (d)}
\begin{itemize}
\item The proof of (b) is similar except that we replace the
  monotonicity condition by the condition on the partial summation.
\item For (c), note that if the pmf of $\Pi_j$ is unimodal where the
  mode is less than or equal to $\eta$ and
  $\p(\Pi_j = 0)\ge \p(\Pi_j = \ceil{\eta M}/M)$, then
  \eqref{eq:asst_partialsum} holds.
\item The proof of (d) is essentially the same as that of part (a), thus omitting here. 
\end{itemize}

\subsection{Proof of Corollary \ref{cor:kfwer_genbnd}}
\label{sec:PfCorfdx}
For any monotonically non-decreasing and non-negative function
$h:\R\mapsto \R$, Markov's inequality gives 
\begin{align*}
  \p(V\geq k) &\leq \dfrac{\E [h(V)]}{h(k)}.
\end{align*}
By the definition of $V$, 
\begin{align}
  \label{eq:V_upperbnd}
  V = \sum_{j\in \calH_0} \Indc\{\Pi_j\ge \eta \} 
  \le \sum_{j\in \calH_0} \dfrac{\Pi_j}{\eta} 
  = \sum_{j\in\calH_0}\sum_{m=1}^M \dfrac{\Indc\{j\in\hat{\calS}^m\}}{\eta M}
  = \sum^M_{m=1}\dfrac{V_m}{\eta M}.
\end{align}
As a consequence,
\begin{align*}
  \frac{\Exs[h(V)]}{h(k)} 
  &\stackrel{(i)}{\leq} \frac{\Exs \Big[ h\left(\sum^M_{m=1}V_m/\eta M\right) \Big]}{h(k)} 
  \stackrel{(ii)}{\leq} \sum_{m=1}^M\frac{\Exs[ h(V_m/\eta) ]}{Mh(k)}
  \stackrel{(iii)}{=} \frac{\Exs[h(V_1/\eta)]}{h(k)}. 
\end{align*}
Step (i) holds since $h$ is non-decreasing and $V$ obeys
\eqref{eq:V_upperbnd}; step (ii) follows from the convexity of $h$ and
Jensen's inequality; step (iii) uses the fact that each $V_m$ has the
same marginal distribution.  Putting things together, the $k$-FWER
satisfies
\begin{align*}
  \p(V\geq k) & ~\leq~ \frac{\Exs[ h(V_1/\eta) ]}{h(k)}.
\end{align*}
To complete the proof, use the fact that $V_1$ is stochastically
dominated by $\mathrm{NB}(v,1/2)$ by Lemma~\ref{LemVknock}.

\subsection{Proof and extension of Proposition \ref{prop:fwer_finerbnd}}
\label{sec:proof_of_finerfwer}
\paragraph*{Proof of Proposition \ref{prop:fwer_finerbnd}}
The proof of Proposition \ref{prop:fwer_finerbnd} is a direct 
consequence of the following result, whose proof is provided in 
Section~\ref{SecPfLemTighter}. 

\begin{lemma}[\cite{huber2018halving}]\label{lem:tighterMarkov}
	Let $X$ denote a non-negative random variable. Assume there exists $\xi>0$, such that 
	\begin{align*}
	\int^\eta_0 \p (X\in [\eta-u,\eta))\dd u \geq \int^\xi_0 \p (X\in [\eta,\eta+u))\dd u,
	\end{align*}
	then the following Markovian type tail bound holds:   
	\begin{align*}
	\p (X\geq \eta)\leq \dfrac{\E [X]}{\eta+\xi}.
	\end{align*}
\end{lemma}

Applying Lemma~\ref{lem:tighterMarkov} to the random variable $V$ with
$\eta = \xi = k$ completes the proof of
Proposition~\ref{prop:fwer_finerbnd}.

\paragraph*{Extension of Proposition \ref{prop:fwer_finerbnd}}
The result below provides conditions under which the bounds in
Corollary \ref{cor:kfwer_genbnd} can be tightened.

\begin{proposition}
\label{prop:fwer_finerbnd_extension}
Consider the setting of Corollary \ref{cor:kfwer_genbnd}.
\begin{enumerate}[label = (\alph*)]
\item Let $h(x) = x^{\nu}$ for $\nu \ge 1$. Suppose the pmf 
of $V$ obeys 
\begin{multline}
\label{eq:asst_v_moment}
\sum^{\floor{k/\nu}-1}_{u=1}\p\left(V\in[k-u,k)\right)+
(k/\nu-\floor{k/\nu})\cdot\p(V\in [k-\floor{k/\nu},k)) \ge \\
\sum^{\floor{k/\nu}}_{u=1}\p(V\in [k,k+u))+ (k/\nu-\floor{k/\nu})\cdot
\p(V\in[k,k+\floor{k/\nu}+1)),
\end{multline}
then \eqref{eq:condition_hofv} holds with $\rho = 1/2$.

\item Let $h(x) = \exp{(\lambda x)}$ for $\lambda\in (0,1]$. 
Suppose the pmf of $V$ obeys 
\begin{multline}
\label{eq:asst_v_exp}
\sum^{\floor{1/\lambda}-1}_{u=1}\p\left(V\in[k-u,k)\right)+
  (1/\lambda-\floor{1/\lambda})\cdot\p(V\in [k-\floor{1/\lambda},k))
  \ge \\
  \sum^{\floor{1/\lambda}}_{u=1}\p(V\in [k,k+u))+
  (1/\lambda-\floor{1/\lambda})\cdot \p(V\in[k,k+\floor{1/\lambda}+1)),
    \end{multline}
    then \eqref{eq:condition_hofv} holds with $\rho = 1/2$.
\end{enumerate}
\end{proposition}

The proof of Proposition \ref{prop:fwer_finerbnd_extension} follows 
from the following lemma, whose proof is provided in Section~\ref{SecPfChebyshev}.

\begin{lemma}
Suppose $X$ is a non-negative random variable. 
\label{lem:tighterChebyshev}
  \begin{enumerate}
    \item[(a)] If
    \begin{align}
    \label{eqn:condition-moments}
      \int^{\frac{k}{\nu}}_0 \p (X\in [k-u,k))\dd u\geq 	\int^{\frac{tk}{\nu}}_0 \p (X\in [k,k+u))\dd u
    \end{align}
    for some $k,t>0$ and $\nu\ge 1$, then 
    \begin{align}
    \label{eqn:chebyshev}
    \p (X\geq k)\le \dfrac{\E [X^\nu]}{(1+t)k^\nu}.
    \end{align}
  
  \item[(b)] If 
  \begin{align*}
    \int^{\frac{1}{\lambda}}_0 \p (X\in [k-u,k))\dd u\geq \int^{\frac{t}{\lambda}}_0 \p (X\in[k,k+u)) \dd u
  \end{align*}
   holds for some $k,t,\lambda>0$, then 
  \begin{align}
  \label{eqn:chernoff}
    \p (X\geq k)\le \dfrac{\E[\exp(\lambda X)]}{(1+t)\exp(\lambda k)}.
  \end{align}
  \end{enumerate}
\end{lemma}

Applying Lemma~\ref{lem:tighterChebyshev} to $V$ with $t=1$ yields
Proposition~\ref{prop:fwer_finerbnd_extension}.

\section{Proofs of auxiliary lemmas}
\subsection{Proof of Lemma \ref{lem:tighterMarkov}}
\label{SecPfLemTighter}

For each $\xi > 0$, we introduce an auxiliary random variable
$U \sim \text{Unif}[-\xi, \eta]$ independent of $X$. We claim that
	\begin{align*}
	\p (X\geq \eta) \le \p (X+U\geq \eta)\leq \dfrac{\E [X]}{\xi+\eta},
	\end{align*}  
        which completes the proof.  To see this, let us start with
        establishing the second inequality. A direct computation
        yields
	\begin{align*}
	\p (X+U\ge \eta) &= \E[\p (U\geq \eta-X|X)]\\
	& =\E[\ind \{\eta-X\leq -\xi\}] +\E\left[\dfrac{X}{\eta+\xi}\ind\{-\xi < \eta -X \leq\eta \}\right]
	\leq \E \left[\dfrac{X}{\xi+\eta} \right].
	\end{align*}
By definition, 
	\begin{align*}
	\p(X+U\geq \eta) &=\dfrac{1}{\eta+\xi} \int^\eta_{-\xi} \p(X\geq \eta -u)\dd u\\
	& = \dfrac{1}{\eta+\xi} \left[\int^0_{-\xi}\p(X\geq \eta -u)\dd u+\int^\eta_0 \p(X\geq \eta -u)\dd u\right]\\
	& = \dfrac{1}{\eta+\xi}\left[\int^\eta_{-\xi}\p(X\geq \eta)\dd u-\int^0_{-\xi}\p(X\in[\eta,\eta-u))\dd u+\int^\eta_0\p (X\in [\eta-u,\eta))\dd u\right]\\
	& = \p(X\geq \eta) +\dfrac{1}{\eta+\xi}\left[-\int^\xi_{0}\p(X\in[\eta,\eta+u))\dd u+\int^\eta_0\p (X\in [\eta-u,\eta))\dd u\right]\\
	& \geq \p (X\ge \eta),
	\end{align*}
where the last inequality follows from our assumption.

\subsection{Proof of Lemma \ref{lem:tighterChebyshev}}
\label{SecPfChebyshev}
Let $U \sim \text{Unif}[-\frac{t k}{\nu},\frac{k}{\nu}]$ be
independent of $X$. We prove that
\begin{align*}
  \p (X\geq k)\leq \p (X+U\geq k) \leq \dfrac{\E [X^\nu]}{(1+t)k^\nu}.
\end{align*}
The second inequality follows from direct calculations:  
\begin{align*}
  \p (X+U\ge k) &=\E \left[\ind\left\{k-X\leq -\frac{tk}{\nu}\right\}\right] +\E \left[\dfrac{\frac{k}{\nu}-k+ X}{(1+t)\frac{k}{\nu}} \cdot \ind\left\{-\frac{tk}{\nu}< k-X\leq \frac{k}{\nu}\right\} \right]\\
                &\leq \E \left[\dfrac{X-\frac{(\nu-1)k}{\nu}}{(1+t)\frac{k}{\nu}}\right] \leq \dfrac{\E [X^\nu]}{(1+t)k^\nu},
\end{align*}
where the last inequality uses
$x-\frac{\nu-1}{\nu}k\leq \frac{x^\nu}{\nu k^{\nu-1}}$. As for the
first inequality, 
\begin{align*}
  \p (X+U\geq k)& = \dfrac{1}{(1+t)\frac{k}{\nu}}\left[\int^{\frac{k}{\nu}}_{-\frac{tk}{\nu}} \p (X\ge k-u)\dd u\right]\\
                & = \dfrac{1}{(1+t)\frac{k}{\nu}}\left[\int^{\frac{k}{\nu}}_{-\frac{tk}{\nu}} \p (X \ge k)\dd u +\int^{\frac{k}{\nu}}_0 \p (X\in [k-u,k))\dd u-\int^{\frac{tk}{\nu}}_{0}\p (X\in[k,k+u))\dd u\right]\\
&\geq \p (X\geq k)
\end{align*}
(the last inequality follows from our assumption).

For the second part, introduce an independent uniform random variable
$U$ supported on $[-\frac{t}{\lambda},\frac{1}{\lambda}]$. We prove
that
	\begin{align*}
    \p (X\geq k)\leq \p (X+U\geq k)\leq \dfrac{\E [\exp(\lambda X)]}{(1+t)\exp(\lambda k)}.
	\end{align*}
        The proof of the first inequality is exactly the same as that
        of Lemma \ref{lem:tighterMarkov} and
        \ref{lem:tighterChebyshev} and is omitted. Direct calculations
        give 
	\begin{align*}
	\p (X+U\geq k) & =\E \left[ \ind\left\{k-X\leq -\frac{t}{\lambda}\right\} \right] + 
	\E \left[\dfrac{X-k+\frac{1}{\lambda}}{\frac{1+t}{\lambda}}\ind\left\{-\frac{t}{\lambda}< k-X \leq \frac{1}{\lambda}\right\}\right]\\
	& \leq\frac{1}{1+t} \E [1+\lambda (X-k)]\\
  & \leq \frac{1}{1+t}\E[\exp(\lambda (X-k))].
	\end{align*}
        The last inequality uses $1+x\leq e^x$ for each $x\in \R$.



\section{Auxiliary details}

\subsection{Construction of the conditional HMM knockoffs}
\label{Sec:construct_kn}
In this section, we provide further details on the construction of the conditional HMM knockoff copies discussed in Section \ref{sec:syntheticex}. 
Suppose the feature vector $X=(X_1,\ldots,X_p)$ follows a hidden Markov model and let $M= (M_1,\ldots,M_p)$ denote the corresponding hidden Markov chain. 
Given a set $C\subset [p]$, we define the ``adjacent set'' to $C$ as 
\begin{align}
  A(C) \defn \{j\in [p]\mid j\notin C,~j+1 \text{ or } j-1\in C \}. 
\end{align}
The adjacent set contains all the SNPs that are adjacent to $C$ on the Markov chain. 
For example, in Figure \ref{fig:adjset}, the red nodes represent $C$
and the blue nodes correspond to $A(C)$.

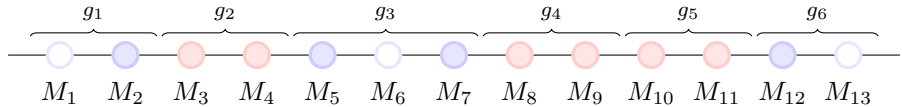
\begin{figure}[ht]
	\centering
  \begin{tikzpicture}[scale = 0.7,
    bluenode/.style = {circle,draw=blue!20, fill=blue!10, very thick, minimum size = 0.5mm},
    rednode/.style = {circle,draw=red!20, fill=red!10, very thick, minimum size = 0.5mm},
    graynode/.style = {circle,draw=blue!10, fill=blue!0, very thick, minimum size = 0.5mm},
    ]
    \node(0)  {};
    \node[graynode] (1) [right=5mm of 0] {};
    \node[below = 0.5mm of 1] {$M_1$};
    \node[bluenode] (2) [right=5mm of 1] {};
    \node[below = 0.5mm of 2] {$M_2$};
    \node[rednode] (3) [right=5mm of 2] {};
    \node[below = 0.5mm of 3] {$M_3$};
    \node[rednode] (4) [right=5mm of 3] {};
    \node[below = 0.5mm of 4] {$M_4$};
    \node[bluenode] (5) [right=5mm of 4] {};
    \node[below = 0.5mm of 5] {$M_5$};
    \node[graynode] (6) [right=5mm of 5] {};
    \node[below = 0.5mm of 6] {$M_6$};
    \node[bluenode] (7) [right=5mm of 6] {};
    \node[below = 0.5mm of 7] {$M_7$};
    \node[rednode] (8) [right=5mm of 7] {};
    \node[below = 0.5mm of 8] {$M_8$};
    \node[rednode] (9) [right=5mm of 8] {};
    \node[below = 0.5mm of 9] {$M_9$};
    \node[rednode] (10) [right=5mm of 9] {};
    \node[below = 0.5mm of 10] {$M_{10}$};
    \node[rednode] (11) [right=5mm of 10] {};
    \node[below = 0.5mm of 11] {$M_{11}$};
    \node[bluenode] (12) [right=5mm of 11] {};
    \node[below = 0.5mm of 12] {$M_{12}$};
    \node[graynode] (13) [right=5mm of 12] {};
    \node[below = 0.5mm of 13] {$M_{13}$};
    \node(14) [right=5mm of 13] {};
    \draw (0) -- (1);
    \draw (1) -- (2);
    \draw (2) -- (3);
    \draw (3) -- (4);
    \draw (4) -- (5);
    \draw (5) -- (6);
    \draw (6) -- (7);
    \draw (7) -- (8);
    \draw (8) -- (9);
    \draw (9) -- (10);
    \draw (10) -- (11);
    \draw (11) -- (12);
    \draw (12) -- (13);
    \draw (13) -- (14);
    \draw [decorate,decoration={brace},xshift=0pt,yshift=-4pt]
      (0.6,0.5) -- (3,0.5) node [black,midway,xshift=0cm,yshift = .3cm]
      {\footnotesize $g_1$};
    \draw [decorate,decoration={brace},xshift=0pt,yshift=-4pt]
      (3.1,0.5) -- (5.5,0.5) node [black,midway,xshift=0cm,yshift = .3cm]
      {\footnotesize $g_2$};
    \draw [decorate,decoration={brace},xshift=0pt,yshift=-4pt]
      (5.6,0.5) -- (9.1,0.5) node [black,midway,xshift=0cm,yshift = .3cm]
      {\footnotesize $g_3$};
    \draw [decorate,decoration={brace},xshift=0pt,yshift=-4pt]
      (9.2,0.5) -- (11.8,0.5) node [black,midway,xshift=0cm,yshift = .3cm]
      {\footnotesize $g_4$};
    \draw [decorate,decoration={brace},xshift=0pt,yshift=-4pt]
      (11.9,0.5) -- (14.3,0.5) node [black,midway,xshift=0cm,yshift = .3cm]
      {\footnotesize $g_5$};
    \draw [decorate,decoration={brace},xshift=0pt,yshift=-4pt]
      (14.4,0.5) -- (16.7,0.5) node [black,midway,xshift=0cm,yshift = .3cm]
      {\footnotesize $g_6$};
  \end{tikzpicture}
  \caption{Visual illustration of the adjacent set: the red nodes correspond to $C=E$ and the blue ones represent $A(C)$. }
	\label{fig:adjset}
\end{figure}

Recall that $\calG$ refers to the set of SNP {\em clusters} identified by the first stage of GWAS.
We define $E$ to be the set of SNPs that appear in the clusters of $\calG$.
The construction of $\widetilde{X}_{\calG}$ proceeds in the following three steps:
\begin{enumerate}
  \item Sample the hidden Markov chain $M$ conditionally on $X$.
  \item Generate a (conditional) knockoff copy $\widetilde{M}_{\calG}$ of $M_{\calG}$ such that for any $g\in\calG$,
    \begin{align}\label{eq:condex}
      (M_{\calG},\widetilde{M}_{\calG})_{\text{swap}(g)}\mid M_{-\calG}\stackrel{\mathrm{d}}{=} (M_{\calG},\widetilde{M}_{\calG})\mid M_{-\calG}.
    \end{align}

  \item Conditional on $\widetilde{M}_{\calG}$, sample $\widetilde{X}_{\calG}$ from the emission probability of the HMM. 
\end{enumerate} 
\begin{proposition}\label{prop:conditionalkn}
  Suppose $X=(X_1,\ldots,X_p)$ follows a hidden Markov model and use
  $M = (M_1,\ldots,M_p)$ to denote the hidden Markov chain. The
  procedure above (steps 1--3) produces valid 
conditional knockoff copies as defined in \eqref{eq:cond_ex} in the
sense that 
  \begin{align}
    (X_{\calG},\widetilde{X}_{\calG})_{\text{swap}(g)} \mid X_{-\calG} \stackrel{\mathrm{d}}{=}(X_{\calG},\widetilde{X}_{\calG}) \mid X_{-\calG}.
  \end{align}
\end{proposition}
\begin{proof}
  Let $(x,x',m,m',a,b)$ be a possible realization of
  $(X_{\calG},\widetilde{X}_{\calG},M_{\calG},\widetilde{M}_{\calG},X_{-\calG},M_{-\calG})$. For
  simplicity, we use $\bar{x}$ and $\bar{x}'$ to denote the resulting
  vectors after swapping $x_g$ and $x'_g$. Define $\bar{m}$ and
  $\bar{m}'$ similarly. Then direct calculations give
  \begin{align*}
    &\p(X_{\calG} = x,\widetilde{X}_{\calG}=x',M_{\calG}=m,\widetilde{M}_{\calG}=m'\mid X_{-\calG}=a,M_{-\calG}=b)\\
    &\quad= \p(X_{\calG} = x,\widetilde{X}_{\calG}=x'\mid M_{\calG}=m,\widetilde{M}_{\calG}=m',X_{-\calG}=a,M_{-\calG}=b)\p ( M_{\calG}=m,\widetilde{M}_{\calG}=m'\mid X_{-\calG}=a,M_{-\calG}=b)\\
    &\quad= \p(X_{\calG} = x,\widetilde{X}_{\calG}=x'\mid M_{\calG}=m,\widetilde{M}_{\calG}=m')\p ( M_{\calG}=m,\widetilde{M}_{\calG}=m'\mid M_{-\calG}=b)\\
    &\quad= \p(X_{\calG} = \bar{x},\widetilde{X}_{\calG}=\bar{x}'\mid M_{\calG}=\bar{m},\widetilde{M}_{\calG}=\bar{m}')\p ( M_{\calG}=\bar{m},\widetilde{M}_{\calG}=\bar{m}'\mid M_{-\calG}=b)\\
    &\quad= \p(X_{\calG} = \bar{x},\widetilde{X}_{\calG}=\bar{x}'\mid M_{\calG}=\bar{m},\widetilde{M}_{\calG}=\bar{m}',X_{-\calG} =a,M_{-\calG}=b)\p ( M_{\calG}=\bar{m},\widetilde{M}_{\calG}=\bar{m}'\mid X_{-\calG}=a,M_{-\calG}=b)\\
    &\quad= \p(X_{\calG} = \bar{x},\widetilde{X}_{\calG}=\bar{x}',M_{\calG}=\bar{m},\widetilde{M}_{\calG}=\bar{m}'\mid X_{-\calG}=a,M_{-\calG}=b).
  \end{align*}
  Above, the second equality uses the hidden Markov structure and the third applies the conditional exchangeability of $M$. Taking expectation w.r.t. $M_{-\calG}$ and 
  marginalizing over $(M_{\calG},\widetilde{M}_{\calG})$ completes the proof.
\end{proof}

The implementation of Steps 1 and 3 is standard and can be carried 
out using internal functions of the \textsf{SNPknock} R-package;
for more details, refer to~\cite{sesia2020controlling}. 
We now elaborate on the implementation of Step 2.
As is shown in the Proposition \ref{prop:adjset} below, \eqref{eq:condex} is equivalent to
\begin{align}\label{eq:condex2}
  (M_{\calG},\widetilde{M}_{\calG})_{\text{swap}(g)} \mid M_{A(E)} \stackrel{\mathrm{d}}{=}(M_{\calG},\widetilde{M}_{\calG})\mid M_{A(E)}
\end{align}
for any $g\in\calG$. We proceed to discuss how to generate
$\widetilde{M}_{\calG}$ obeying \eqref{eq:condex2}. Let us call the
consecutive clusters in $\calG$ a clip; for example in Figure
\ref{fig:adjset}, $g_2$ is a clip, $g_4$ and $g_5$ form a clip, and so on.
Conditional on $A(E)$, the clips are independent of each other thanks 
to the Markov property.  It is then sufficient to generate
knockoff copies separately for each clip: in the example from Figure \ref{fig:adjset}, we shall thus sample
$\widetilde{M}_{g_2}$ and $\widetilde{M}_{(g_4,g_5)}$
independently. To sample $\widetilde{M}_{g_2}$---a knockoff copy for
$M_{g_2}$ conditional on $A(E)$---we shall generate
$(\widetilde{M}_2,\widetilde{M}_{g_2},\widetilde{M}_5)$, which is a
knockoff copy for $(M_2,M_{g_2},M_5)$. An immediate consequence is that the
subset $\widetilde{M}_{g_2}$ is a knockoff copy for $M_{g_2}$
conditional on $(M_2,M_5)$; i.e.,
\[
  (M_{g_2},\widetilde{M}_{g_2})\mid M_{(2,5)} \stackrel{d}{=} (\widetilde{M}_{g_2},M_{g_2})\mid M_{(2,5)},
\]
which by the Markov property is equivalent to 
\[
  (M_{g_2},\widetilde{M}_{g_2})\mid M_{A(E)} \stackrel{d}{=} (\widetilde{M}_{g_2},M_{g_2})\mid M_{A(E)}.
\]
The same reasoning can be used to construct knockoffs for
$M_{(g_4,g_5)}$ and all the other clips. In the end, we hold $\widetilde{M}_{\calG}$.

\begin{proposition}\label{prop:adjset}
Suppose $M=(M_1,\ldots,M_p)$ is a Markov chain. Then 
\begin{subequations}
\begin{align}\label{eq:condproof1}
(M_{\calG},\widetilde{M}_{\calG})_{\text{swap}(g)} \mid M_{A(E)} \stackrel{\mathrm{d}}{=}(M_{\calG},\widetilde{M}_{\calG})\mid M_{A(E)}
\end{align}
if and only if 
\begin{align}\label{eq:condproof2}
(M_{\calG},\widetilde{M}_{\calG})_{\text{swap}(g)} \mid M_{-\calG} \stackrel{\mathrm{d}}{=}(M_{\calG},\widetilde{M}_{\calG})\mid M_{-\calG}.
\end{align}
\end{subequations}
\end{proposition}
\begin{proof}
The derivation from \eqref{eq:condproof2} to \eqref{eq:condproof1} is straightforward since 
$A(E)$ is a subset of the complement of $E$ (and $M_{-\calG}=M_{-E}$). To see 
the reverse direction, note that by the Markov property, $M_{\calG}$ is independent of 
$M_{-(E\cup A(E))}$ conditional on $M_{A(E)}$. 
Additionally, the construction of $\widetilde{M}_{\calG}$ only depends on $M_{E \cup A(E)}$.\footnote{Technically,
  $\widetilde{M}_{\calG} = f(M_{E \cup A(E)},U)$,
  where $U \sim \text{Unif}[0,1]$ is independent of everything.}
Consequently, one has 
\begin{align*}
(M_{\calG},\widetilde{M}_{\calG})\mid M_{-\calG}
\stackrel{\mathrm{d}}{=} (M_{\calG},\widetilde{M}_{\calG})\mid
M_{A(E)} \stackrel{\mathrm{d}}{=} (M_{\calG},\widetilde{M}_{\calG})_{\text{swap}(g)}\mid
M_{A(E)} \stackrel{\mathrm{d}}{=}
  (M_{\calG},\widetilde{M}_{\calG})_{\text{swap}(g)}\mid M_{-\calG}. 
\end{align*} 
\end{proof}




\subsection{Additional figures}
\label{sec:additional_figs}

\paragraph{Large-scale simulation from Section \ref{sec:pfer}} 
Figure \ref{fig:pfer_amp_gaussian_large_ratio} and \ref{fig:pfer_amp_gaussian_large_pidist} plot the realized
ratio between $\p(\Pi_j\ge 0.5)$ and $\E[\Pi_j]$ (with $95\%$ confidence intervals), and the pooled histograms 
of all nonzero null $\Pi_j$'s. 
\begin{figure}[ht]
	\centering
	\includegraphics[width = .95\textwidth]{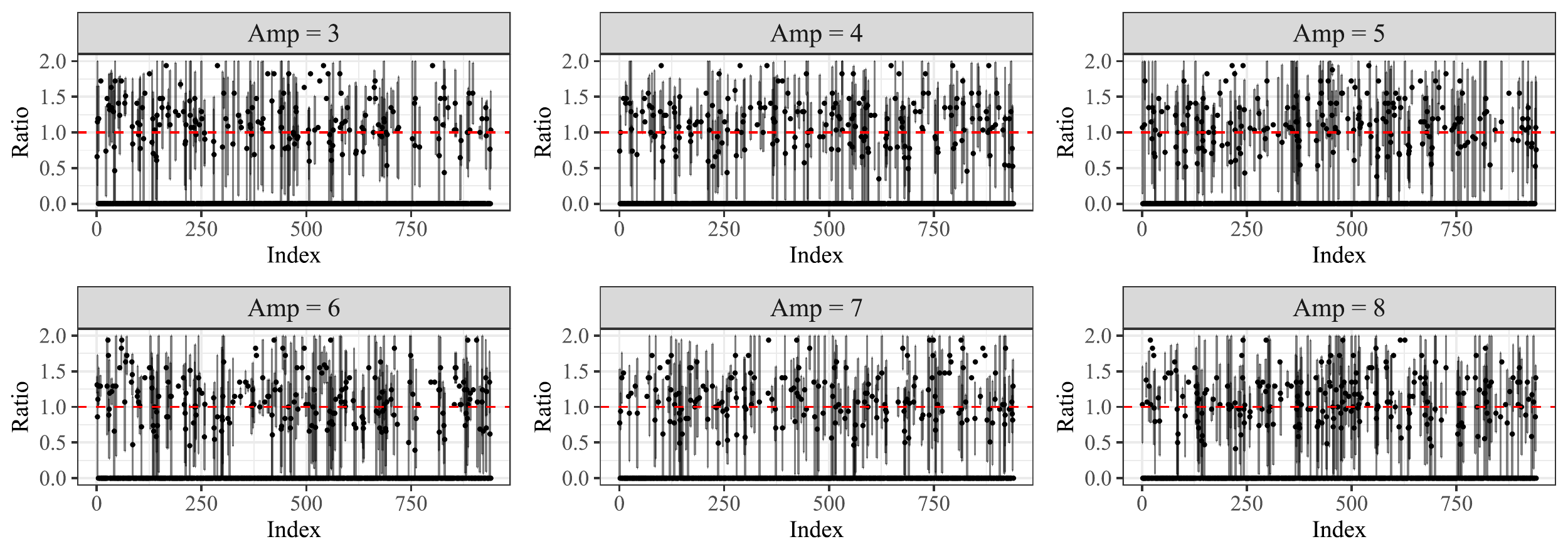}
        \caption{Realized ratios $\p(\Pi_j\ge 1/2)/\E[\Pi_j]$ with
          $95\%$ confidence intervals estimated from $200$
          repetitions. The experiment setting is the same as in Figure
          \ref{fig:pfer_amp_gaussian_large_within}.}
	\label{fig:pfer_amp_gaussian_large_ratio}
\end{figure}
\begin{figure}[ht]
	\centering
	\includegraphics[width = .95\textwidth]{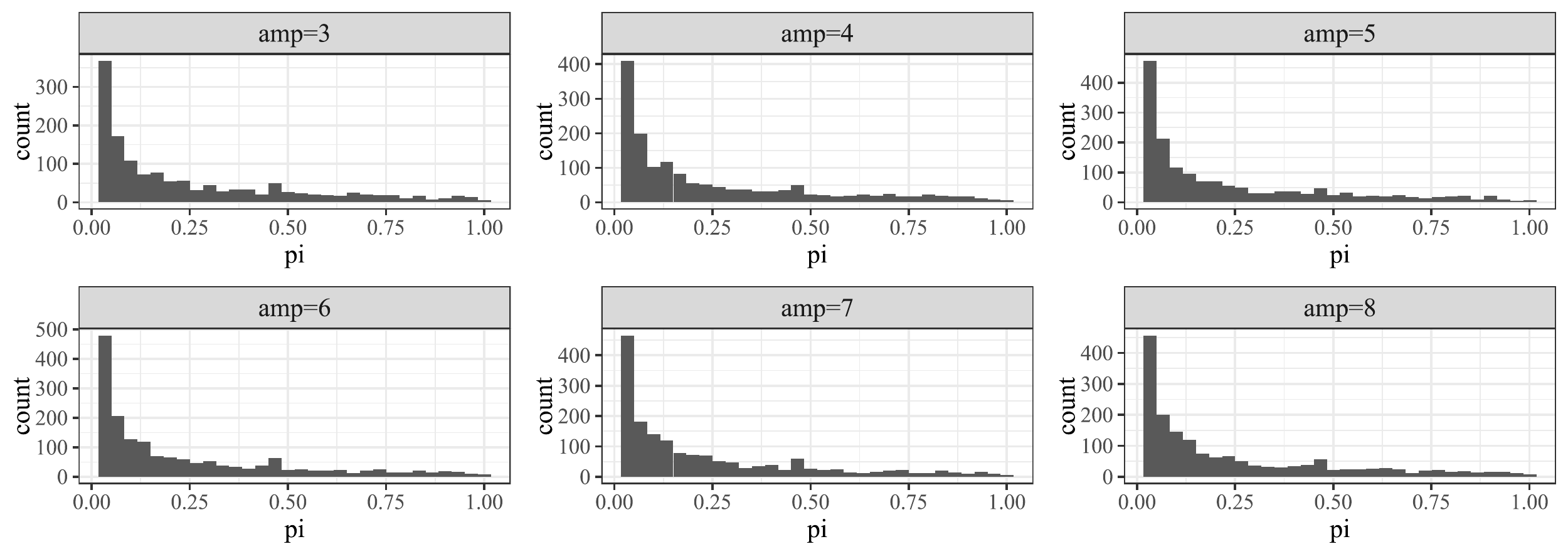}
	\caption{Pooled histograms of all nonzero null $\Pi_j$'s for
    different signal amplitudes. The
          experiment setting is the same as in Figure
          \ref{fig:pfer_amp_gaussian_large_within}.}
	\label{fig:pfer_amp_gaussian_large_pidist}
\end{figure}
\paragraph{Large-scale simulation from Section \ref{sec:fwer}}
Figure \ref{fig:fwer_amp_gaussian_large_ratio} plots the realized ratio between $\p(\Pi_j\ge 0.5)$ and 
$\E[\Pi_j]$ (with $95\%$ confidence intervals); Figure \ref{fig:fwer_amp_gaussian_large_pidist} and 
\ref{fig:fwer_amp_gaussian_large_Vdist} are the (pooled) histograms of all nonzero null $\Pi_j$'s 
and the number of false discoveries. 

\begin{figure}[ht]
	\centering
	\includegraphics[width = .95\textwidth]{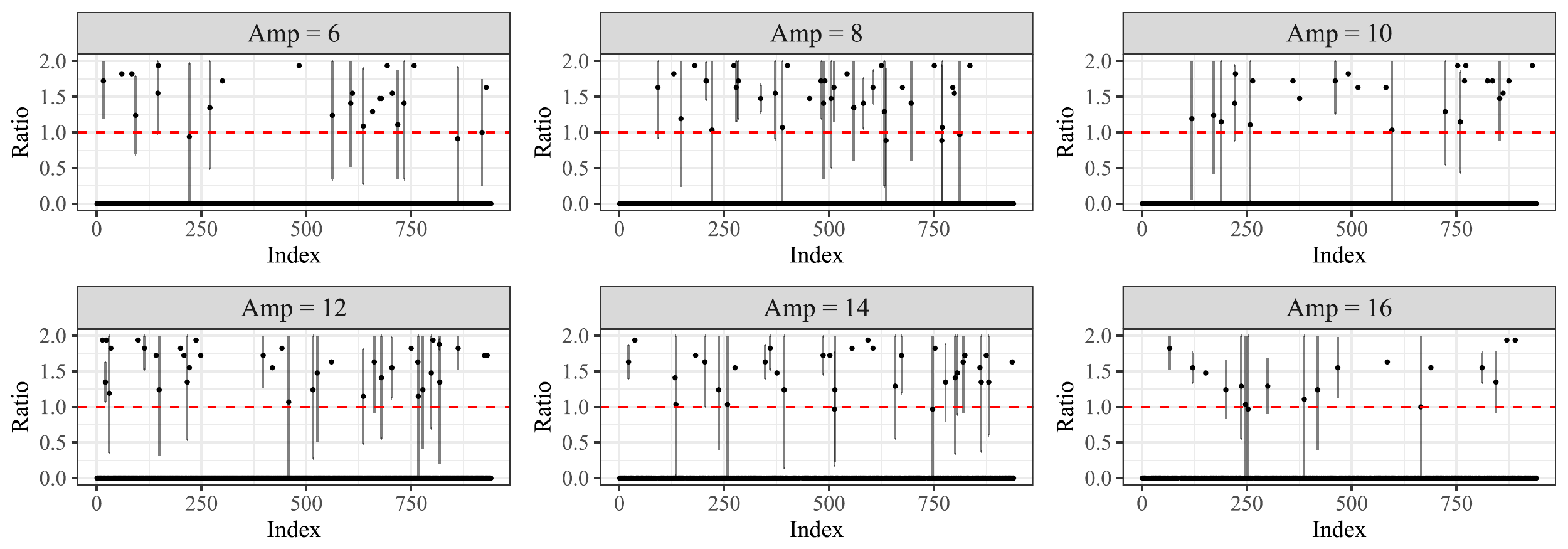}
        \caption{Realized ratios $\p(\Pi_j\ge 1/2)/\E[\Pi_j]$ with
          $95\%$ confidence intervals estimated from $200$
          repetitions. The experiment setting is the same as in Figure
          \ref{fig:fwer_amp_gaussian_large_within}.}
	\label{fig:fwer_amp_gaussian_large_ratio}
\end{figure}
\begin{figure}[ht]
	\centering
	\includegraphics[width = .95\textwidth]{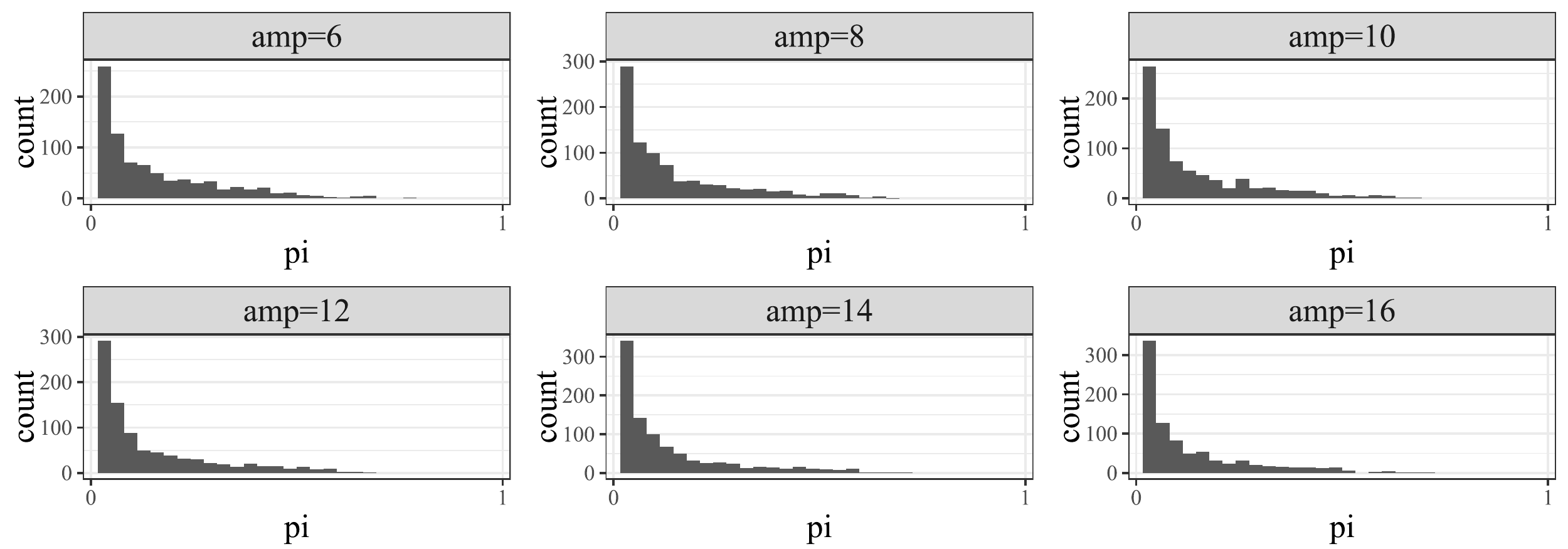}
	\caption{Pooled histograms of all nonzero null $\Pi_j$'s for
          different signal amplitudes. The experiment setting is the
          same as in Figure \ref{fig:fwer_amp_gaussian_large_within}.}
	\label{fig:fwer_amp_gaussian_large_pidist}
\end{figure}

\begin{figure}[H]
  \centering
  \includegraphics[width = .95\textwidth]{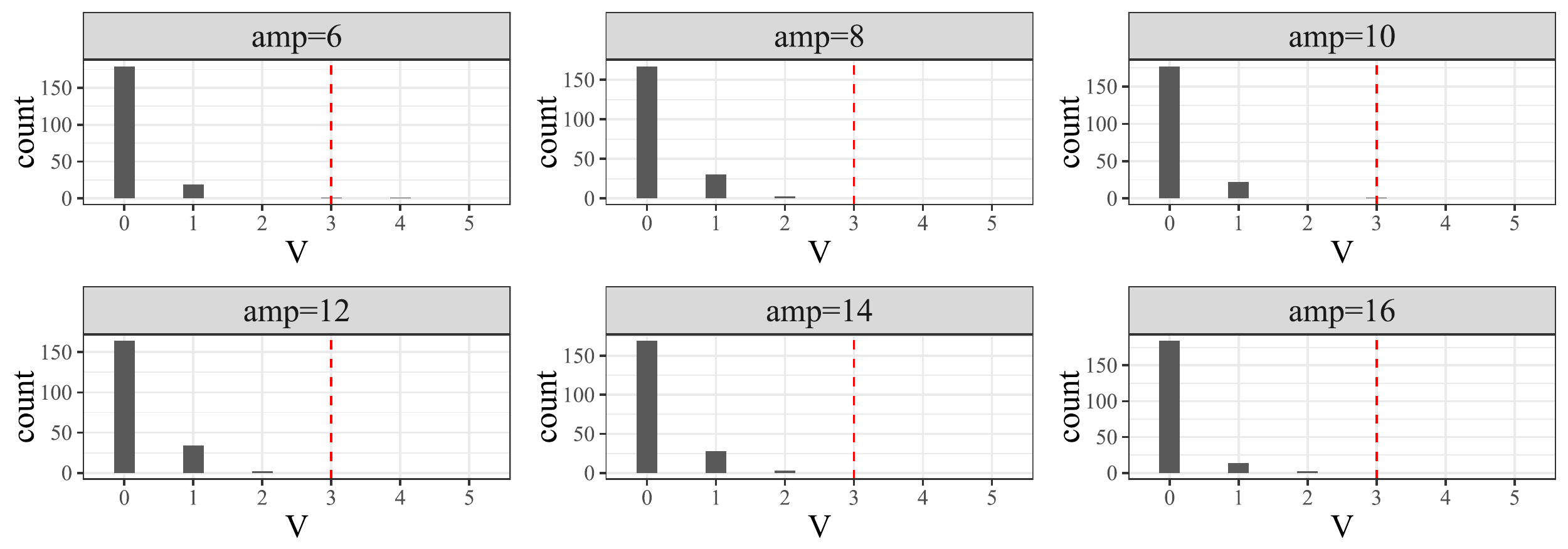}
  \caption{Histograms of the number of false discoveries $V$ for
    different signal amplitudes. The
    experiment setting is the same as in Figure
    \ref{fig:fwer_amp_gaussian_large_within}.}
  \label{fig:fwer_amp_gaussian_large_Vdist}
\end{figure}

\subsection{Additional tables}
\label{sec:extra_tables}
\begin{table}[ht]
  \centering
  
\begin{tabular}[t]{ccccccccccccc}
\toprule
\multicolumn{1}{c}{ } & \multicolumn{1}{c}{ } & \multicolumn{11}{c}{V} \\
\cmidrule(l{3pt}r{3pt}){3-13}
Amplitude & Method & 0 & 1 & 2 & 3 & 4 & 5 & 6 & 7 & 8 & 9 & 10+\\
\rowcolor{gray!6}
\midrule
 & Derandomized Knockoffs & 147 & 42 & 9 & 1 & 1 & 0 & 0 & 0 & 0 & 0 & 0\\

\rowcolor{gray!6}
\multirow{-2}{*}{\centering\arraybackslash 3} & Vanilla Knockoffs & 119 & 47 & 17 & 11 & 5 & 1 & 0 & 0 & 0 & 0 & 0\\

 & Derandomized Knockoffs & 159 & 33 & 8 & 0 & 0 & 0 & 0 & 0 & 0 & 0 & 0\\

\multirow{-2}{*}{\centering\arraybackslash 4} & Vanilla Knockoffs & 117 & 42 & 24 & 8 & 6 & 1 & 1 & 0 & 1 & 0 & 0\\

\rowcolor{gray!6}
 & Derandomized Knockoffs & 163 & 31 & 3 & 3 & 0 & 0 & 0 & 0 & 0 & 0 & 0\\

\rowcolor{gray!6}
\multirow{-2}{*}{\centering\arraybackslash 5} & Vanilla Knockoffs & 116 & 42 & 23 & 11 & 5 & 1 & 1 & 1 & 0 & 0 & 0\\

 & Derandomized Knockoffs & 153 & 35 & 10 & 2 & 0 & 0 & 0 & 0 & 0 & 0 & 0\\

\multirow{-2}{*}{\centering\arraybackslash 6} & Vanilla Knockoffs & 110 & 45 & 30 & 7 & 1 & 3 & 3 & 0 & 1 & 0 & 0\\

\rowcolor{gray!6}
 & Derandomized Knockoffs & 145 & 48 & 6 & 1 & 0 & 0 & 0 & 0 & 0 & 0 & 0\\

\rowcolor{gray!6}
\multirow{-2}{*}{\centering\arraybackslash 7} & Vanilla Knockoffs & 102 & 44 & 20 & 18 & 8 & 5 & 3 & 0 & 0 & 0 & 0\\

 & Derandomized Knockoffs & 152 & 40 & 7 & 1 & 0 & 0 & 0 & 0 & 0 & 0 & 0\\

\multirow{-2}{*}{\centering\arraybackslash 8} & Vanilla Knockoffs & 102 & 51 & 22 & 10 & 7 & 3 & 3 & 0 & 1 & 1 & 0\\
\bottomrule
\end{tabular}

  \caption{Frequencies (out of $200$ runs) of the number of false 
    discoveries. The simulation setting 
    is the same as in Figure \ref{fig:pfer_amp_gaussian_within}.}
  \label{tab:pfer_amp_gaussian}
\end{table}

\begin{table}[ht]
  \centering
  \begin{table}[H]
\centering
\resizebox{\linewidth}{!}{
\begin{tabular}[t]{ccccccccccccccccccc}
\toprule
\multicolumn{1}{c}{ } & \multicolumn{1}{c}{ } & \multicolumn{17}{c}{V} \\
\cmidrule(l{3pt}r{3pt}){3-19}
Amplitude & Method & 0 & 1 & 2 & 3 & 4 & 5 & 6 & 7 & 8 & 9 & 10 & 11 & 12 & 13 & 14 & 15 & 16+\\
\rowcolor{gray!6}
\midrule
 & Derandomized Knockoffs & 62 & 57 & 45 & 21 & 15 & 0 & 0 & 0 & 0 & 0 & 0 & 0 & 0 & 0 & 0 & 0 & 0\\

\rowcolor{gray!6}
\multirow{-2}{*}{\centering\arraybackslash 3} & Vanilla Knockoffs & 57 & 43 & 35 & 23 & 27 & 6 & 4 & 4 & 0 & 0 & 1 & 0 & 0 & 0 & 0 & 0 & 0\\

 & Derandomized Knockoffs & 41 & 77 & 51 & 22 & 8 & 1 & 0 & 0 & 0 & 0 & 0 & 0 & 0 & 0 & 0 & 0 & 0\\

\multirow{-2}{*}{\centering\arraybackslash 4} & Vanilla Knockoffs & 49 & 48 & 44 & 20 & 17 & 11 & 7 & 1 & 2 & 1 & 0 & 0 & 0 & 0 & 0 & 0 & 0\\

\rowcolor{gray!6}
 & Derandomized Knockoffs & 51 & 65 & 53 & 17 & 10 & 2 & 1 & 1 & 0 & 0 & 0 & 0 & 0 & 0 & 0 & 0 & 0\\

\rowcolor{gray!6}
\multirow{-2}{*}{\centering\arraybackslash 5} & Vanilla Knockoffs & 52 & 45 & 40 & 24 & 18 & 8 & 4 & 7 & 0 & 2 & 0 & 0 & 0 & 0 & 0 & 0 & 0\\

 & Derandomized Knockoffs & 48 & 69 & 44 & 22 & 15 & 1 & 1 & 0 & 0 & 0 & 0 & 0 & 0 & 0 & 0 & 0 & 0\\

\multirow{-2}{*}{\centering\arraybackslash 6} & Vanilla Knockoffs & 47 & 44 & 35 & 23 & 17 & 15 & 6 & 6 & 4 & 3 & 0 & 0 & 0 & 0 & 0 & 0 & 0\\

\rowcolor{gray!6}
 & Derandomized Knockoffs & 57 & 61 & 53 & 20 & 6 & 2 & 1 & 0 & 0 & 0 & 0 & 0 & 0 & 0 & 0 & 0 & 0\\

\rowcolor{gray!6}
\multirow{-2}{*}{\centering\arraybackslash 7} & Vanilla Knockoffs & 49 & 40 & 39 & 26 & 21 & 11 & 10 & 3 & 1 & 0 & 0 & 0 & 0 & 0 & 0 & 0 & 0\\

 & Derandomized Knockoffs & 40 & 71 & 54 & 20 & 9 & 5 & 1 & 0 & 0 & 0 & 0 & 0 & 0 & 0 & 0 & 0 & 0\\

\multirow{-2}{*}{\centering\arraybackslash 8} & Vanilla Knockoffs & 45 & 54 & 44 & 21 & 13 & 11 & 6 & 2 & 1 & 1 & 0 & 1 & 0 & 0 & 0 & 1 & 0\\
\bottomrule
\end{tabular}}
\end{table}

  \caption{Frequencies (out of $200$ runs) of the number of false 
    discoveries. The simulation setting 
    is the same as in Figure \ref{fig:pfer_amp_gaussian_large_within}.}
  \label{tab:pfer_amp_gaussian_large}
\end{table}

\begin{table}[ht]
  \centering
  
\begin{tabular}[t]{ccccccccccc}
\toprule
\multicolumn{1}{c}{ } & \multicolumn{1}{c}{ } & \multicolumn{9}{c}{V} \\
\cmidrule(l{3pt}r{3pt}){3-11}
Amplitude & Method & 0 & 1 & 2 & 3 & 4 & 5 & 6 & 7 & 8+\\
\rowcolor{gray!6}
\midrule
 & Derandomized Knockoffs & 196 & 4 & 0 & 0 & 0 & 0 & 0 & 0 & 0\\

\rowcolor{gray!6}
\multirow{-2}{*}{\centering\arraybackslash 10} & Vanilla Knockoffs & 171 & 20 & 9 & 0 & 0 & 0 & 0 & 0 & 0\\

 & Derandomized Knockoffs & 196 & 4 & 0 & 0 & 0 & 0 & 0 & 0 & 0\\

\multirow{-2}{*}{\centering\arraybackslash 15} & Vanilla Knockoffs & 172 & 18 & 6 & 2 & 1 & 1 & 0 & 0 & 0\\

\rowcolor{gray!6}
 & Derandomized Knockoffs & 194 & 6 & 0 & 0 & 0 & 0 & 0 & 0 & 0\\

\rowcolor{gray!6}
\multirow{-2}{*}{\centering\arraybackslash 20} & Vanilla Knockoffs & 166 & 18 & 9 & 4 & 3 & 0 & 0 & 0 & 0\\

 & Derandomized Knockoffs & 192 & 8 & 0 & 0 & 0 & 0 & 0 & 0 & 0\\

\multirow{-2}{*}{\centering\arraybackslash 25} & Vanilla Knockoffs & 152 & 22 & 15 & 6 & 4 & 0 & 0 & 1 & 0\\

\rowcolor{gray!6}
 & Derandomized Knockoffs & 193 & 7 & 0 & 0 & 0 & 0 & 0 & 0 & 0\\

\rowcolor{gray!6}
\multirow{-2}{*}{\centering\arraybackslash 30} & Vanilla Knockoffs & 161 & 25 & 7 & 5 & 1 & 0 & 0 & 1 & 0\\

 & Derandomized Knockoffs & 190 & 9 & 1 & 0 & 0 & 0 & 0 & 0 & 0\\

\multirow{-2}{*}{\centering\arraybackslash 35} & Vanilla Knockoffs & 155 & 27 & 4 & 10 & 3 & 1 & 0 & 0 & 0\\
\bottomrule
\end{tabular}

  \caption{Frequencies (out of $200$ runs) of the number of false 
    discoveries. The simulation setting 
    is the same as in Figure \ref{fig:fwer_amp_gaussian_within}.}
  \label{tab:fwer_amp_gaussian}
\end{table}

\begin{table}[ht]
  \centering
  
\begin{tabular}[t]{ccccccccccccc}
\toprule
\multicolumn{1}{c}{ } & \multicolumn{1}{c}{ } & \multicolumn{11}{c}{V} \\
\cmidrule(l{3pt}r{3pt}){3-13}
Amplitude & Method & 0 & 1 & 2 & 3 & 4 & 5 & 6 & 7 & 8 & 9 & 10+\\
\rowcolor{gray!6}
\midrule
 & Derandomized Knockoffs & 179 & 19 & 0 & 1 & 1 & 0 & 0 & 0 & 0 & 0 & 0\\

\rowcolor{gray!6}
\multirow{-2}{*}{\centering\arraybackslash 6} & Vanilla Knockoffs & 127 & 35 & 19 & 10 & 5 & 4 & 0 & 0 & 0 & 0 & 0\\

 & Derandomized Knockoffs & 167 & 30 & 3 & 0 & 0 & 0 & 0 & 0 & 0 & 0 & 0\\

\multirow{-2}{*}{\centering\arraybackslash 8} & Vanilla Knockoffs & 111 & 35 & 28 & 15 & 9 & 0 & 1 & 1 & 0 & 0 & 0\\

\rowcolor{gray!6}
 & Derandomized Knockoffs & 177 & 22 & 0 & 1 & 0 & 0 & 0 & 0 & 0 & 0 & 0\\

\rowcolor{gray!6}
\multirow{-2}{*}{\centering\arraybackslash 10} & Vanilla Knockoffs & 125 & 37 & 15 & 14 & 5 & 2 & 1 & 0 & 0 & 1 & 0\\

 & Derandomized Knockoffs & 164 & 34 & 2 & 0 & 0 & 0 & 0 & 0 & 0 & 0 & 0\\

\multirow{-2}{*}{\centering\arraybackslash 12} & Vanilla Knockoffs & 123 & 42 & 18 & 6 & 6 & 3 & 0 & 0 & 2 & 0 & 0\\

\rowcolor{gray!6}
 & Derandomized Knockoffs & 169 & 28 & 3 & 0 & 0 & 0 & 0 & 0 & 0 & 0 & 0\\

\rowcolor{gray!6}
\multirow{-2}{*}{\centering\arraybackslash 14} & Vanilla Knockoffs & 125 & 42 & 18 & 7 & 3 & 3 & 1 & 1 & 0 & 0 & 0\\

 & Derandomized Knockoffs & 184 & 14 & 2 & 0 & 0 & 0 & 0 & 0 & 0 & 0 & 0\\

\multirow{-2}{*}{\centering\arraybackslash 16} & Vanilla Knockoffs & 126 & 45 & 17 & 7 & 4 & 0 & 1 & 0 & 0 & 0 & 0\\
\bottomrule
\end{tabular}

  \caption{Frequencies (out of $200$ runs) of the number of false 
    discoveries. The simulation setting 
    is the same as in Figure \ref{fig:fwer_amp_gaussian_large_within}.}
  \label{tab:fwer_amp_gaussian_large}
\end{table}

\end{document}